\documentclass[11pt]{article}

\usepackage{setspace}
\usepackage[margin=1.1in]{geometry}

\usepackage[english]{babel}

\usepackage{amsmath}
\usepackage{amsfonts}
\usepackage{amssymb}
\usepackage{amsthm}
\usepackage{graphicx}
\usepackage{tikz-cd}
\usepackage{mathrsfs}
\usepackage{xfrac}
\usepackage{url}
\usepackage{color}
\usepackage[colorlinks=true, citecolor=blue]{hyperref}
\usepackage{enumitem}
\usepackage[usestackEOL]{stackengine}
\usepackage{accents}
\usepackage{bbm}
\usepackage{authblk}

\numberwithin{equation}{section}
\def\a{\alpha}

\def\be{\begin{equation}}
\def\ee{\end{equation}}
\def\ba#1\ea{\begin{align}#1\end{align}}
\def\no{\nonumber\\ }
\def\ra{\rangle}
\def\la{\langle}

\def\pu{\partial_u}
\def\pup{\partial_{u'}}
\def\ang{\hat r}
\def\mem{{\scr M}}

\newcommand{\ca}[1]{\mathcal{#1}}
\newcommand{\bb}[1]{\mathbb{#1}}
\newcommand{\fr}[1]{\mathfrak{#1}}
\newcommand{\scr}[1]{\mathscr{#1}}

\newcommand{\wt}[1]{\widetilde{#1}}

\newcommand{\ol}[1]{\overline{#1}}
\newcommand{\un}[1]{\underline{#1}}

\newtheorem{definition}{Definition}
\newtheorem{theorem}{Theorem}
\newtheorem{lemma}{Lemma}
\newtheorem{proposal}{Proposal}

\graphicspath{{Figures/}}

\newcommand{\barray}{\begin{array}}
\newcommand{\earray}{\end{array}}
\newcommand{\bea}{\begin{eqnarray}}
\newcommand{\eea}{\end{eqnarray}}

\newcommand{\bit}{\begin{itemize}}
\newcommand{\eit}{\end{itemize}}
\newcommand{\bd}{\begin{description}}
\newcommand{\ed}{\end{description}}

\newcommand{\re}{\mathrm{Re}}
\newcommand{\im}{\mathrm{Im}}

\def\la{\langle}
\def\ra{\rangle}

\def\w{\wedge}
\newcommand{\p}{\partial}

\newcommand{\R}{\mathbb{R}}

\newcommand{\f}{\frac}

\renewcommand{\a}{\alpha}    
\renewcommand{\d}{\delta}  \newcommand{\eps}{\epsilon} \newcommand{\z}{\zeta}
 \renewcommand{\th}{\theta}

 \newcommand{\s}{\sigma}      
     
 \let\Om=\Omega

\usepackage{slashed}

\usepackage{pgf}
\makeatletter
\newcommand*{\pgfunderleftarrow}{%
  \@ifstar
    {\let\ifpgf@depth\iftrue\mathpalette\@pgfunderleftarrow}
    {\let\ifpgf@depth\iffalse\mathpalette\@pgfunderleftarrow}%
}
\newcommand*{\@pgfunderleftarrow}[2]{%
  #2%
  \edef\pgf@math@fam{\the\fam}%
  \pgfpicture
    \pgfsetbaseline{0pt}
    \pgf@relevantforpicturesizefalse      
    \pgfsetroundcap                       
    \pgfsetarrowsend{to}
    \pgfutil@tempdima=0.28pt%
    \advance\pgfutil@tempdima by.8\pgflinewidth%
    \pgfutil@tempdima-4\pgfutil@tempdima
    \sbox\pgfutil@tempboxa{$\m@th\fam\pgf@math@fam#1#2$}%
    \advance\pgfutil@tempdima-\dp\pgfutil@tempboxa
    \pgfutil@tempdimb\wd\pgfutil@tempboxa
    \pgfpathmoveto{\pgfqpoint{0pt}{\pgfutil@tempdima}}%
    \pgfpathlineto{\pgfqpoint{-\pgfutil@tempdimb}{\pgfutil@tempdima}}%
    \pgfusepath{stroke}
    \ifpgf@depth
      \pgf@relevantforpicturesizetrue
      \pgfpathmoveto{\pgfqpoint{0pt}{-\pgfutil@tempdimb}}%
      \pgfusepath{use as bounding box}%
    \fi
  \endpgfpicture
}
\makeatother

\newcommand{\sscr}{\scriptscriptstyle\rm}

\usepackage{manfnt}
\reversemarginpar

\DeclareFontFamily{U}{matha}{\hyphenchar\font45}
\DeclareFontShape{U}{matha}{m}{n}{
      <5> <6> <7> <8> <9> <10> gen * matha
      <10.95> matha10 <12> <14.4> <17.28> <20.74> <24.88> matha12
      }{}
\DeclareSymbolFont{matha}{U}{matha}{m}{n}
\DeclareMathSymbol{\oright}       {2}{matha}{"69}

\usepackage{pifont}

\usepackage{mathrsfs}
\newcommand{\Dd}{D}
\newcommand{\scri}{{\mathscr{I}}} 
 
\def\con{K}
\newcommand{\DY}{\text{div}(Y)}

\newcommand{\se}{s^{\sscr e}} 
\newcommand{\sm}{s^{\sscr m}} 

\title{An algebra of proper observables at null infinity:\\ Dirac brackets, memory and Goldstone probes
\vspace{5mm}
}
\date{November 5, 2026}
\author[1]{Rodrigo Andrade e Silva\thanks{randradeesilva@perimeterinstitute.ca}}
\author[2]{Simone Speziale\thanks{simone.speziale@cpt.univ-mrs.fr}}
\affil[1]{\it \small Perimeter Institute for Theoretical Physics, Waterloo, ON, Canada}
\affil[2]{\it \small Aix Marseille Univ., Univ. de Toulon, CNRS, CPT, UMR 7332, 13288 Marseille, France}

\begin{document}
\begin{titlepage}
\maketitle
\thispagestyle{empty}

\begin{abstract}
We develop a rigorous evaluation of Dirac brackets for classical observables on the phase space of radiative gravitational modes at null infinity that naturally incorporates memory effects. 
Considering the Ashtekar-Streubel phase space, with boundary conditions in time given by vanishing {\it news} and purely electric {\it shear}, and taking into account the infinite dimensionality of the phase space, we identify the algebra of proper observables (understood as
functions on phase space that can be associated with smooth symplectic flows).
We show that the action of supertranslation charges generates the correct transformations on the shear.
We also show that the conventional definition of the ``Goldstone mode'' adopted in the literature cannot be associated with a proper observable, but nevertheless there exists an infinite family of proper observables, which we call {\it Goldstone probes}, that are capable of measuring the Goldstone mode.
We notice that there are no Goldstone probes constructed only out of the shear {\it or} the news, providing a possible explanation for why attempts to construct a (separable) Hilbert space with different memory states have failed so far.
Finally, we derive formulas for distributional Dirac brackets between local shear and news, and show that they contain non-local corrections. 
\end{abstract}

\end{titlepage}

\tableofcontents

\section{Introduction}
\label{SecIntro}

Defining quantum states carrying gravitational memory is an important open question, crucial for the understanding of infrared physics and the existence of the gravitational $S$-matrix, see e.g. \cite{Ashtekar:1981sf,Ashtekar84book,He:2014laa,Strominger:2017zoo,Ashtekar:2018lor,Prabhu:2022zcr}.
This problem has a classical counterpart which, in our opinion, remains open: there is no satisfactory characterization of the algebra of proper observables on the radiative phase space at null infinity $\scri$ that includes memory as a non-central element.
By proper observables we mean functions on phase space that can be associated with (smooth) Hamiltonian vector fields, and therefore inherit a natural algebraic structure from the symplectic form, via the Poisson brackets.
Since quantization relies fundamentally on the classical algebra of (proper) observables, resolving the classical problem could shed some light on the quantum problem as well.
A possible step forward appeared in \cite{He:2014laa}, as a consequence of attempting to solve a seemingly different problem. Namely, they noticed an apparent issue with the action of supertranslations on the shear (a missing factor of $2$), via a computation of Dirac brackets, which led them to postulate a {\it Goldstone mode} canonically conjugated to the memory.  
While this idea has been quite influential, an actual derivation of the brackets including the Goldstone mode 
has not been provided in the literature yet.
(Further investigations of this matter, and some proposals, have appeared in \cite{Campiglia:2021bap,Freidel:2021dfs,Campiglia:2024uqq}.)
It is thus crucial to ascertain the status of the Goldstone mode, and moreover to characterize the complete algebra of radiative gravitational observables inherited from the symplectic structure at null infinity. 
In this paper we will address these points.

We consider the Ashtekar-Streubel phase space of radiative gravitational modes at $\scri$, with non-vanishing asymptotic limits of the shear, necessary to include memory. (While we assume purely electric shear in the limits, our analysis can be straightforwardly extended to the case with arbitrary asymptotic limits of the shear, with essentially no modifications to the results.)
Taking into careful account the infinite dimensionality of the phase space and the non-trivial boundary conditions, we obtain a complete characterization of the {\it algebra of proper gravitational observables} at $\scri$.
We furthermore clarify non-trivial properties of the Dirac brackets on this phase space, and their evaluation leads to significant novelties compared to the literature.
With our construction, the action of supertranslations on the shear is correct, and one has a clear understanding of what the observables are. 
Notably, we find that there exists an infinite family of proper observables, all built from mixed shear {\it and} news, that do not commute with the memory. 
As they are capable of ``indirectly'' measuring the Goldstone mode, we call them {\it Goldstone probes}.
Importantly, it follows that the memory, which belongs to the algebra of proper observables, is not central.
On the other hand, we find that the Goldstone mode itself is not a proper observable, and cannot be included in the algebra. 
Nevertheless, the Goldstone-memory bracket proposed in \cite{He:2014laa} admits a precise interpretation as a ``directional bracket'' compatible with the Ashtekar-Streubel symplectic structure. (This neither includes the Goldstone mode in the algebra of proper observables nor defines brackets among Goldstone modes.) In any case, no bracket involving the Goldstone mode is needed to establish the non-centrality of the memory or to recover the correct supertranslation action.

The gravitational radiative phase space and its basic properties were established a long time ago \cite{Ashtekar:1981bq,Ashtekar:1981hw,Ashtekar:1990gc}. 
The first observables that were constructed were the news smeared against rapidly-decaying test functions \cite{Ashtekar:1981sf,Ashtekar:2018lor}, and their algebra is the starting point for the Fock quantization of the radiative phase space \cite{Ashtekar84book,He:2014laa,Strominger:2017zoo,Ashtekar:2018lor,Prabhu:2022zcr}.
This construction excludes memory, and 
in existing attempts
to restore it, it appears as a central element in the classical algebra, and super-selected at the quantum level.
If one uses a Dirac delta as a smearing function, one gets schematically
\be\label{NNbra}
\{N(u),N(u')\}\sim \p_u\d(u,u').
\ee
This is the bracket used in \cite{He:2014laa}. 
The problem with using conclusions obtained in this manner is that the bracket 
\eqref{NNbra} is \emph{inconsistent with the presence of memory}, as we establish in this paper. 
It is thus ill-fated to use it to compute brackets on $\scri$ when memory is allowed.
To gain some intuition about the problem, notice that 
the right-hand side structure in \eqref{NNbra} is typical of any light-front canonical analysis, and indeed the relation between the news and the shear has precisely the structure of the second-class constraint known as light-front condition \cite{Torre:1985rw,Heinzl:1993px,Heinzl:2000ht,Alexandrov:2014rta,DePaoli:2017sar}. Work on the light-front canonical analysis shows that a distributional Dirac bracket like \eqref{NNbra} is obtained under trivial or periodic boundary conditions \cite{Heinzl:1993px,Heinzl:2000ht,Barnich:2024aln}, both incompatible with memory. 
When discussing observables, non-trivial boundary conditions on the fields typically imply also non-trivial boundary conditions on the smearings, and we will see in detail how this happens. Then, a fatal flaw of an expression like \eqref{NNbra} becomes apparent: for spaces of smearings that do not vanish asymptotically, derivatives of Dirac deltas do {\it not} satisfy the familiar property $\pup\delta(u,u') = - \pu\delta(u,u')$, and consequently such a bracket is not even anti-symmetric. As we will see, to construct proper observables that do not commute with memory (our Goldstone probes) it is essential to consider such smearings of the news that do not vanish asymptotically.

In order to extend the construction of \cite{Ashtekar:1981sf,Ashtekar84book} to include non-central memory, we
employ a rigorous extension of the Dirac brackets analysis to infinite dimensions.
As mentioned before, the crux of the matter is that observables should correspond to functions with smooth symplectic flows tangent to the physical phase space.
In finite dimensional systems this is a relatively trivial step, but the situation changes in field theories because the symplectic $2$-form tends to be only {\it weakly} non-degenerate. 
While this fact is well-known, in many situations the consequences are mild and do not affect the physics in essential ways. For example, it may require that observables are smooth functions with smooth functional derivatives, and one can still manipulate local fields $\phi(x)$ as distributions, without much care, and still obtain correct results.
What makes it peculiar, and interesting, in the gravitational radiative phase space is that the non-trivial boundary conditions, together with the light-front constraint, imply a non-local constraint on the space of smearings, ``dual'' in a certain sense to the phase space constraint.
This, in particular, is what prevents Goldstone probes from being constructed only out of shears or only out of news.
Moreover, the need to work with spaces of smearings that do not vanish asymptotically implies that distributions behave in possibly unfamiliar ways, as indicated earlier.
In view of those subtleties, we find that a certain amount of mathematical precision seems to be helpful and actually necessary in order to make progress on this matter.

\subsection{Summary}

We will now proceed with a brief summary of the paper.

In Sec.~\ref{SecSympScri}, we begin with a basic review of the structure of null infinity, $\scri$, and describe the phase space of radiative gravitational modes according to Ashtekar-Streubel. 
While their symplectic form is constructed using covariant phase space methods, it can also be taken as a starting point for a canonical Dirac analysis by considering the news and the shear as independent variables related by a primary constraint. This allows us to frame the results in a language that is familiar to a large part of the community.
In Sec.~\ref{SecBMS}, we describe the action of the BMS algebra on the phase space.
In Sec.~\ref{Seccomplexvariables}, we review the formalism of Newman-Penrose, which repackages the two degrees of freedom of the shear into a single complex object (and similarly for the news). This offers a convenient parametrization for the fields, simplifying the manipulations and expressions.
In Sec.~\ref{SecBoundaryC}, we discuss more carefully the boundary conditions necessary for our analysis, allowing for memory states and admitting a regular action of the BMS algebra. The resulting choices for the spaces of shears and news are inspired by the original proposal of Ashtekar-Streubel, with a slight modification.

In Sec.~\ref{SecDirbracketsscri}, we study Dirac brackets on $\scri$. We begin with a quick review of the standard Dirac brackets formula, explain subtleties in extending it to infinite dimensions and summarize our proposal for such a generalization.
(The details of this proposal are explained in App.~\ref{AppDirbrackets}.)
We also introduce {\it directional brackets}, which are not brackets in the standard sense but merely a notation to indicate how regular functions (i.e., associated with smooth Hamiltonian flows) act on smooth (not necessarily regular) functions.
In Sec.~\ref{SecKreg}, we study the conditions for (smeared) constraints to be regular and prove, in the language of our formalism, that the set of regular constraints is second class.
In Sec.~\ref{SecTEshearnews}, we consider the {\it shear-news functions}, $\Gamma_{\lambda,\alpha}$, defined by smearings of the shear and news with functions $\lambda$ and $\alpha$, respectively. We derive the conditions for the (kinematical) regularity and tangential extensibility of $\Gamma_{\lambda,\alpha}$ (i.e., such that there is a regular function $\wt \Gamma_{\lambda,\alpha}$, equal on-shell to $\Gamma_{\lambda,\alpha}$, whose symplectic flow is tangent to the constraint surface).
Importantly, we find that $\lambda$ and $\alpha$ must satisfy a non-local, ``dual'' constraint.
In Sec.~\ref{SecDBsN}, we construct their Dirac brackets, and also their directional Dirac brackets.
In Sec.~\ref{SecSuperTflow}, we show how the supertranslation charge generates the correct transformation of the shear.
We see how non-local contributions to the brackets, related to the tangential extensions, are essential for obtaining the correct ``factor of $2$'' for the action of the soft component of the supertranslation.

In Sec.~\ref{SecRedPS}, we test the validity of our Dirac brackets formalism by considering a reduced phase space analysis. Indeed, we confirm that tangential extensibility on the kinematical phase space is (on-shell) equivalent to regularity on the reduced phase space. Also, we show that the Dirac brackets match (on-shell) to the Poisson brackets on the reduced phase space.
The analysis of the algebra of proper observables is actually simpler from the perspective of the reduced phase space, and accordingly some readers may prefer to skip to this section to find a derivation of the conditions for the shear-news functions to be proper observables, as well as a derivation of their brackets.

In Sec.~\ref{Secalgebra}, we discuss more generally the algebra of proper observables and some of its properties.
In Sec.~\ref{SecEntireA}, we describe the full algebra of proper observables. Their Dirac brackets are obtained analogously to the derivation of the brackets for the shear-news functions.
In Sec.~\ref{SecMemGP}, we discuss the status of the memory and the Goldstone mode. We prove that the only functionals of the asymptotic values of the shear that are equal (on-shell) to proper observables are functionals of the memory mode alone. (The proof uses a lemma established in App.~\ref{AppLemma1formPK}.) It follows that the Goldstone mode cannot be turned into a proper observable, and therefore its brackets are not well-defined. 
We also introduce the concept of {\it Goldstone probes}, which are proper observables that do not commute with memory. 
Importantly, we notice that there are no Goldstone probes constructed out of the shear or the news alone. 
Finally, we mention how one can construct a family of proper observables that, in a limiting sense, converge (pointwise) to a ``dressed shear'' and, further, to the standard ``Goldstone mode'', but it is still not possible to consistently define brackets for those limits from the perspective of the full algebra.
We also show that, when restricting to sequences of {\it linear} observables, brackets can be defined via this limiting procedure and they agree, when the limit exists, with the directional brackets. In particular, the directional bracket of the memory with the Goldstone mode reproduces the bracket posited in \cite{He:2014laa}.
In Sec.~\ref{SecCompleteSubA}, we discuss complete subalgebras of proper observables. In particular, we show that the set of shear-news functions (together with the constant function $1$) forms a complete subalgebra.
This subalgebra is not minimal, and we describe two smaller complete subalgebras of it.

In Sec.~\ref{SecDDIB}, we extend the Dirac brackets to complex and ``localized'' functions. 
In Sec.~\ref{SecComplexB}, we consider an algebraic extension to complex valued shear-news functions. Those can be used as tools for computing real brackets, but only via suitable linear combinations.
In Sec.~\ref{SecDistbrackets}, we consider {\it distributional brackets}, giving a precise meaning to brackets between ``localized'' functions, like $\{\sigma(x), N(x')\}_D$.
In Sec.~\ref{SecDefDistB}, we explain that those are not brackets in the standard sense, but should be fundamentally understood as bidistributional objects. Specifically, their defining property is that brackets between proper observables can be computed by suitably smearing the distributional brackets.
In Sec.~\ref{Secfuncanalysis}, we discuss aspects of the functional analysis necessary to rigorously define those brackets. In particular, we define a series of relevant distributions on these spaces of smearings (some of which contain functions that do not vanish asymptotically). We highlight some unfamiliar properties that these distributions may have, such as  $\pup\delta(u,u') \ne - \pu\delta(u,u')$ and $\pu H(u,u') \ne \delta(u,u')$. 
In Sec.~\ref{DistBracketsFormulas}, we display the formulas for the distributional proper Dirac brackets and the distributional directional Dirac brackets. It is apparent that these formulas contain boundary corrections, not encountered in the standard formulas of the literature such as in \eqref{NNbra}. 
(The derivation of these formulas is given, in great detail, in App.~\ref{DistBracketsEval}.)

In Sec.~\ref{SecDiscussion}, we conclude by recapitulating the main results of the paper and commenting on some of their aspects and prospects.

\subsection{Conventions and notation}

\begin{itemize}[
	label=,
     labelsep=1.7em,
     itemsep=7pt,
     leftmargin=2em,
     itemindent=-2em
]

\item[] \un{\it Scri}~~ 
We denote points in $\scri$ as 
\[
x = (u,\ang) \in \scri
\]
where $u \in \bb R$ is the affine parameter along null rays and $\ang \in S^2$. Tensors will be denoted with Latin indices, e.g. $\sigma_{ab}$ for the shear and $N_{ab}$ for the news. The future-pointing null vector is denoted by $n = \pu$.
The volume elements on $\scri$ and $S^2$ are respectively $\epsilon_\scri$ and $\epsilon_S$.
Integrals without a limit indicate integration over the whole domain, e.g.
\[
\int\!du := \int_{-\infty}^{+\infty}\!du \,,
\]
and integrals over the sphere are indicated by $\int_S$. 
The Lie derivative along $n$ will be denoted simply as $\pounds_n = \pu$. 

\item[] \un{\it Newman-Penrose}~~ 
We work in the language of the Newman-Penrose formalism, where the shear $\sigma$ and news $N$ are expressed as complex functions on $\scri$, with spin-weights $2$ and $-2$, respectively. The space of smooth, complex spin-$s$ functions on a manifold $\ca M$ is denoted by
\[
C^\infty(\ca M,\bb C_s) \,.
\]
(Such functions may also carry a conformal weight, but we keep it implicit.) 
The space of smooth, purely electric, spin-$2$ functions on the sphere is denoted by $\ca E$. Complex-conjugation is denoted with a bar, $\sigma \mapsto \bar\sigma$.

\item[] \un{\it Phase spaces}~~ 
The (kinematical) phase space is denoted by $\ca P$, with symplectic form $\Omega$, and the constrained phase space by $\ca P_K$, with constraint $K$. Equality on-shell of the constraint is expressed with $\doteq$. Points in phase space are denoted by $p$ and vectors by $X$. The exterior derivative is $\delta$ and the interior product is $I$ (where $I_X$ inserts $X$ into the first slot of a form). Products of forms on the phase space are understood as exterior products, so $\delta N \,\delta \sigma (X,X') = \delta N(X) \delta\sigma(X') - \delta N(X') \delta\sigma(X)$. We denote the action of vectors on phase space functions simply by
\[
XF := \delta F(X)\,.
\]
The reduced phase space is denoted by $\ca P_R$, with induced symplectic form $\Omega_R$, and tangent vectors are denoted by $V$.

\item[] \un{\it Brackets}~~ 
Functions on the phase space are generically denoted by $F$ or $G$.
We say that $F$ is {\it regular} if it can be associated with a smooth Hamiltonian vector field $X_F$. 
We say that $F$ is {\it tangentially extensible} if it admits a tangential extension $\wt F$. The Hamiltonian vector field of $\wt F$ is denoted by $\wt X_F$.
Poisson and Dirac brackets are denoted with a comma, respectively $\{\,\,,\,\}$ and $\{\,\,,\,\}_D$. Directional brackets are denoted with a semicolon, $\{\,\,;\,\}$ and $\{\,\,;\,\}_D$, with the convention
\[
\{F;G\} := - X_F G\,,
\]
that is, ``the first entry flows the second (with a minus sign)''.

\item[] \un{\it Function spaces}~~ 
The space of shears $\sigma$ is denoted by $\ca F_0$ and the space of (conjugated) news $\bar N$ is denoted by $\ca F_1$. 
The letter $\varphi$ is used for the smearing parameter of the constraint $K$, and the space of $\varphi$ which make $K_\varphi$ regular is called $\ca K$.
The letters $\lambda$ and $\alpha$ are used for the smearing parameters of the shear and news, respectively, in the shear-news function $\Gamma_{\lambda,\alpha}$. The spaces of smooth $(\lambda, \alpha)$ which make $\Gamma_{\lambda,\alpha}$ regular and tangentially extensible are called, respectively, $\ca G_R$ and $\ca G_{TE}$.

\item[] \un{\it Other}~~ 
For the smearing parameter $\lambda$, of the shear in $\Gamma_{\lambda,\alpha}$, we introduce the notation 
\[
\Lambda(x) := \int_{-\infty}^u\!du'\,\lambda(u',\ang)
\]
For any function $f$ on $\scri$, 
\[
f_\pm(\ang) := \lim_{u\to\pm\infty}f(u,\ang)\,.
\]
We use units where $G=c=1$. The set of natural numbers (not including zero) is denoted by $\bb N$.
\end{itemize}

\section{The symplectic structure on $\scri$}
\label{SecSympScri}

We will work with a standard Bondi completion, in which the induced metric at null infinity, $\scri$, has unit-round sphere cross sections. (For definiteness, we shall restrict our attention to future null infinity, $\scri^+$, and denote it $\scri$ for short, but all the constructions can naturally be carried over to past null infinity, $\scri^-$.)
As $\scri$ has topology $\bb R\times S^2$, each point $x\in \scri$ can be decomposed as 
\be
x=(u,\ang) \in\scri
\ee
with $u\in \bb R$ being a (complete) affine parameter and $\ang\in S^2$ a point on the sphere.
We use coordinates $(u,\th,\phi)$, so that the null vector tangent to $\scri$ is $n=\p_u$, $(\theta,\phi)$ are standard angular coordinates on $S^2$ and $q = d\theta^2 + \sin^2\theta d\phi^2$. 
The gravitational radiative degrees of freedom on $\scri$ can be described in terms of the asymptotic shear $\s_{ab}$ associated with the $u$ foliation via $l_a=-\p_au$,
and the news tensor $N_{ab} = 2 \pounds_n\s_{ab}$. The shear and news tensors are  symmetric, 
spatial (in the sense that $\sigma_{ab}n^b = N_{ab}n^b = 0$), and traceless.\footnote{For spatial tensors, raising indices and taking traces are independent of the choice of inverse metric. If need be to make a choice of inverse, we will use the one that annihilates $l_a$.} 
They thus contain two independent components each, which can be thought of as their ``electric'' and ``magnetic'' components. For example, for the shear 
we denote them respectively $\se$ and $\sm$, defined via the decomposition
\be\label{shearem}
\s_{ab}=\Dd_{\la a}\Dd_{b\ra} \se+ \eps_{(a}{}^{c}\Dd_{b)}\Dd_c \sm
\ee
 where $\Dd$ is the spatial covariant derivative compatible with $q$ and $\Dd_{\la a}\Dd_{b\ra} := \Dd_{a}\Dd_{b} - \frac{1}{2}q_{ab}\Dd_c\Dd^c$ is the trace-free part of $\Dd_{a}\Dd_{b}$.

The Ashtekar-Streubel radiative phase space at $\scri$ \cite{Ashtekar:1981hw,Ashtekar:1981bq} is defined including boundary conditions at early and late times, characterized by, at least, vanishing news and purely electric shear
\be\label{bc}
\lim_{u\to\pm \infty} N_{ab} =0, \qquad \lim_{u\to\pm\infty} \s_{ab}=\Dd_{\la a}\Dd_{b\ra} \se_\pm.
\ee
Here $\se_\pm$ are two real functions on the sphere. 
The difference between the final and initial values of the shear defines the \emph{displacement memory} \cite{Christodoulou:1991cr}
\be\label{defmem}
\mem := \Dd^{a}\Dd^{b}{\s}_{ab} \Big|^{\infty}_{-\infty} = \f12(\Dd^2+2)\Dd^2(\se_+-\se_-),
\ee
which is generically produced in dynamical situations such as 
gravitational coalescence or scattering.
The boundary configurations \eqref{bc} correspond to vacua of the gravitational field \cite{Ashtekar:1981hw}, identified from the vanishing `outgoing' Weyl tensor components. Weaker boundary conditions allowing for magnetic shear are called weak vacua, and have received some attention in the more recent literature, see e.g. \cite{Winicour:2014ska,Satishchandran:2019pyc,Bieri:2020zki}. (While we will not consider such boundary conditions explicitly in this paper, our analysis applies to this case as well, with the results essentially unchanged.)
Stronger boundary conditions with fully vanishing shear are on the other hand too restrictive, describing a solution space without memory, of limited physical interest.
In order to have a sufficiently well-behaved phase space including memory, while still admitting a large class of regular observables, we will have to assume a restriction of the phase space where the news approaches its vanishing asymptotic limits in a sufficiently fast manner. The particular restriction will be made precise in Sec.~\ref{Seccomplexvariables}.

The Ashtekar-Streubel symplectic form on $\scri$ is \cite{Ashtekar:1981bq}\footnote{The sign conventions for the symplectic form and the brackets are not uniform in the literature. We take the sign of the symplectic form as in the original paper \cite{Ashtekar:1981bq}, and define the brackets so that global translations are generated by a negative energy flux.}
\be\label{OmASos}
\Omega_R = \frac{1}{8\pi} \int \eps_\scri\, \d\sigma^{ab}\pounds_n\d \s_{ab},
\ee
where $\eps_\scri = du \w \left(\sin\th d\th \w d\phi\right)$ is the volume element compatible with $q_{ab}$ and $n$, $\delta$ is the exterior derivative in field space and the exterior product is implicit (whenever there are products of field-space forms).
Hereinafter an undefined integral means integration over the entire domain.
For the Dirac analysis, it is convenient to write this symplectic form in canonical form as
\be\label{OmAS}
\Omega = -\frac{1}{16\pi} \int \eps_\scri \, \delta N_{ab}\,\d\sigma^{ab},
\ee
where the news and shear appear as canonically conjugated variables, but are related by the constraint
\be\label{constraint}
\con_{ab}:=N_{ab} - 2 \pounds_n\sigma_{ab} = 0\,.
\ee
Then \eqref{OmASos} is recovered on the reduced phase space (on-shell of this constraint).
The canonical analysis based on \eqref{OmAS} and \eqref{constraint} is a Dirac-type analysis, and the one based on \eqref{OmASos} is a reduced phase space analysis,
hence our label $R$ in \eqref{OmASos}. They give equivalent results, and we will develop both in this paper.

\subsection{The BMS algebra}
\label{SecBMS}

In this paper we restrict attention to the Bondi-van der Burg-Metzner-Sachs (BMS) boundary conditions, and their associated symmetry group and universal structure 
\cite{Bondi:1962px,Sachs:1962zza,Newman:1966ub,Geroch:1977jn,Ashtekar:1981bq} (see \cite{Speziale:2025lkm} for a recent review). 
Using Bondi coordinates, the BMS algebra can be labeled by  
vector fields 
\be
\xi^a = \left(T + \frac{u}{2}\, \text{div}(Y)\right) n^a + Y^a,
\ee
where $T$ is a $u$-independent scalar function on $\scri$ (of conformal weight $1$),\footnote{For the purposes of our analysis it behaves like a standard real function. In this paper, we will not label explicitly the conformal weights of functions (in particular, $\sigma$ and $N$ should carry conformal weights $-1$ and $-2$, respectively).} $Y^a$ is a 
conformal Killing vector of the sphere, satisfying $l_aY^a = 0$, $\pounds_nY^a = 0$ and $\pounds_Y q_{ab} =\DY\, q_{ab}$,
and we introduced the short-hand notation 
$\text{div}(Y) := D_aY^a$ for its (spatial) divergence.

The radiative phase space carries a symplectic realization of this symmetry group \cite{Ashtekar:1981bq,Wald:1999wa,Barnich:2011mi,Flanagan:2015pxa,Grant:2021sxk,Chandrasekaran:2021vyu,Odak:2022ndm,Rignon-Bret:2024gcx,Ashtekar:2024stm}.
Let $X_\xi$ be the vector field tangent to the phase space generating the BMS transformation associated with $\xi$. 
One can prove using the fall-off conditions \eqref{bc} that the flow is Hamiltonian, namely
\be
\delta F_\xi \doteq - I_{X_\xi}\Omega.
\ee
The canonical generator is the Ashtekar-Streubel flux
\be\label{ASflux}
F_\xi =-\frac{1}{16 \pi} \int\!\epsilon_\scri\, N_{ab} \delta_\xi \sigma^{ab},
\ee
where
\be\label{dxis}
\d_{\xi} \s_{ab} = \frac{1}{2}\left(T + \frac{u}{2} \,\DY \right) N_{ab}+ \pounds_Y\sigma_{ab} -  \frac{1}{2} \DY \sigma_{ab} + {\Dd}_{\langle a} \Dd_{b\rangle} T 
\ee
is the BMS transformation of the shear. 
We also recall that the news transforms as
\be
\d_\xi N_{ab} = \left(T + \frac{u}{2} \,\DY \right) \pounds_n N_{ab} + \pounds_YN_{ab}
\ee
and the metric $q_{ab}$ is fixed.

The fluxes provide a symplectic representation of the BMS algebra, without central extensions, as
\be
\{F_\xi,F_\chi\}_D = F_{[\xi,\chi]}.
\ee
where $[\xi, \eta]$ is the Lie bracket of vector fields on $\scri$.
The $D$ marks that the bracket should be understood in the sense of Dirac, as we will review below.
If news and shear were proper observables, one would expect  their transformations to be canonically generated, 
\be\label{dxiF}
\{F_\xi, N_{ab}\}_D \stackrel{?}= -\d_\xi N_{ab}\,, \qquad  \{F_\xi, \s_{ab}\}_D \stackrel{?}= -\d_\xi \s_{ab}.
\ee
As we will see, these brackets are technically ill-defined since the local shear and news are not proper observables. 
Nevertheless, we will give a meaning to an analogous expression via the ``directional brackets'', defined later.

Specifically for a supertranslation $T$, 
the flux reads
\be\label{STflux}
F_T \doteq - \frac{1}{32\pi} \int\!\epsilon_\scri\, T \left( N_{ab} N^{ab} + 2 D_a D_b N^{ab} \right). 
\ee
The \emph{hard} and \emph{soft} parts are given respectively by the terms quadratic and linear in the news.
In the literature, this is referred to as supermomentum flux, but also as supertranslation charge.

\subsection{Complex variables}
\label{Seccomplexvariables}

It is convenient to re-express the shear and news in terms of (spin-weighted) complex functions. 
To that end, we introduce a complex (local) basis of vectors on the sphere $(m^a,\bar m^a)$, satisfying
\be
m^a m_a = 0 \qquad \bar m^a \bar m_a = 0 \qquad m^a \bar m_a = 1.
\ee
The basis is determined by $q_{ab}$ up to an arbitrary phase, which we fix choosing
\be
m = \frac{1}{\sqrt{2}}\left(\frac{\partial}{\partial\theta} - \frac{i}{\sin\theta} \frac{\partial}{\partial\phi}\right)
=P\f{\p}{\p\z}, \qquad P:= \f{\z(1+|\z|^2)}{\sqrt 2|\z|},
\ee
where $\z=-e^{i\phi}\tan\f\th2$ is the stereographic projection from the south pole.
Using this basis, we define the decomposition
\ba
\sigma_{ab} &= - \bar\sigma m_a m_b - \sigma \bar m_a \bar m_b \label{sabs} \\
N_{ab} &= 2N m_a m_b + 2 \bar N \bar m_a \bar m_b. \label{NabN}
\ea
Under a complex rotation $m^a \mapsto e^{i\psi}m^a$  of the basis, we have $(\s,N)  \mapsto (e^{2i\psi}\s, e^{-2i\psi}N)$; thus, $\sigma$ has spin-weight $2$, while $N$ has spin-weight $-2$. Note that we are taking $m^a$ to be constant in $u$, so we restrict to time-independent complex rotations.
The complex variables are convenient for simplifying some calculations and expressions.\footnote{The relation to the ``holomorphic'" notation of \cite{He:2014laa} is
\[
\s
=\f{P^2}2C_{ z z}\,, \qquad \pu {\bar\s} =\f{P^2}2 N_{\bar z\bar z}\,.
\]
We also notice that these variables are related to the metric component $C_{AB}$ used in Bondi coordinates by $\s_{ab}=-\f12 \d_a^A\d_b^B C_{AB}$, $N_{ab}=-\d_a^A\d_b^B\pu C_{AB}$.}

In these variables, the symplectic form \eqref{OmAS} reads 
\be\label{OmASNP}
\Om = \frac{1}{4\pi} \re\!\int\! \eps_\scri\, \delta N \,\delta\s,  
\ee
and the constraint \eqref{constraint} reads
\be\label{constraint1}
\con := \bar N + \pu\sigma = 0, 
\ee
where we have used that $\pounds_n$ acts on complex, spin-weighted variables as $\pu$. Throughout the rest of the paper, we will always use $\pu$ instead of $\pounds_n$. 
The BMS flux generator \eqref{ASflux} reads
\be\label{ASfluxNP}
F_\xi = \f1{4\pi}\re\!\int\!\epsilon_\scri\, N\,\d_\xi \s\,,
\ee
with shear transformation \eqref{dxis}
\be\label{ASfluxNP2}
\d_{\xi} \s = -\left(T+\f u2\DY\right)\bar N+\left(\pounds_Y-\f12\DY\right)\s -\eth^2 T\,.
\ee
where $T \in C^\infty(S^2,\bb R)$. 
Specializing to supertranslations,
\be\label{STfluxNP}
F_T= -\frac{1}{4\pi} \re\!\int\!\epsilon_\scri\, T \left( N \bar N + \eth^2 N \right), \qquad \d_T \s = -T\bar N -\eth^2 T,
\ee
For instance, \eqref{STfluxNP} can be obtained from \eqref{STflux} using \eqref{NabN} and $\eth^2 N = D_a D_b(N m^a m^b)$, where $D_a$ is the spatial derivative.

Another advantage of the complex variables is to simplify the notion of electric/magnetic decomposition introduced in \eqref{shearem}. The covariant derivative on spin-weighted quantities is represented by the $\eth$ ({\it eth}) operator \cite{PenroseRindler2}. For round spheres and a quantity of spin weight $s$, it is given by 
\be
\eth\eta := (\sin\theta)^s m[(\sin\theta)^{-s}\eta].
\ee
We also define $\bar\eth$ by conjugation, $\bar\eth\eta = \overline{\eth\bar\eta}$. An important property that follows from the topology of the sphere 
is that for any $\eta$ with spin $s \ge 0$, there always exists a spin-$0$ function $\zeta$ such that $\eta = \eth^s\zeta$. Then we can define the electric and magnetic parts of $\eta$ as the unique decomposition
\be
\eta=\eta_e+\eta_m
\ee
with
\be
\eta_e := \eth^s \re(\zeta), \qquad
\eta_m := i \eth^s \im(\zeta).
\ee
We also have\footnote{This follows from the fact that $[\eth^s, \bar\eth^r] = 0$ whenever it acts on a quantity of spin-weight $(r-s)/2$.}
\be\label{emkernel}
\im[\bar\eth^s\eta_e] = \re[\bar\eth^s\eta_m] = 0.
\ee
Accordingly, $\eta$ is said to be {\it purely electric} if $\eta = \eta_e$ (or $\im[\bar\eth^s\eta] = 0$), and {\it purely magnetic} if $\eta = \eta_m$ (or $\re[\bar\eth^s\eta] = 0$). 
An analogous definition works for $\eta$ with $s<0$. For scalars, $s=0$, being purely electric or purely magnetic are equivalent, respectively, to being purely real or purely imaginary.

Another useful property is the orthogonality of the two components under integration on the sphere, that is
\be\label{emort}
\re\int_S\!\epsilon_S\, w \,\bar\eta = 0
\ee
for all purely electric $\eta$ implies that $w$ is purely magnetic.\footnote{Let $w$ and $\eta$ be smooth spin-$s$ functions on $S^2$, with $s\ge 0$. If $\re\int_S\!\epsilon_S\, w \,\bar\eta = 0$ for all purely electric $\eta$, then as $\eta = \eth^s\zeta$ for some {\it real} scalar function $\zeta$, we have
\[
0 = \re\int_S\!\epsilon_S\, w\, \bar\eta = \re\int_S\!\epsilon_S\,  w\, \bar\eth^s\zeta = (-1)^s\re\int_S\!\epsilon_S\, \bar\eth^s w\, \zeta = (-1)^s\int_S\!\epsilon_S\, \re[\bar\eth^s w] \,\zeta
\]
Since $\zeta$ can be any real function, $\re[\bar\eth^s w] = 0$, so $w$ must be purely magnetic. An analogous proof works for $s<0$.}

\subsection{Boundary conditions}
\label{SecBoundaryC}

The boundary conditions \eqref{bc} read, in complex variables,
\be\label{bc2}
\lim_{u\to\pm \infty} N =0, \qquad
\lim_{u\to\pm\infty} \s = \eth^2 \se_\pm.
\ee
where $\se_\pm$ are real smooth functions on the sphere.
As we explain, it is necessary \cite{Ashtekar:1981bq} to strengthen these boundary conditions slightly so that, at least, the following properties are satisfied:
\begin{enumerate}

\item To consider observables constructed by smearing the news with functions that do not vanish asymptotically, such as memory, we need
$\int\!\eps_\scri\, |N| < \infty$. 

\item To have a well-defined supertranslation flux \eqref{STflux}, we need $\int\!\eps_\scri\, |N|^2 < \infty$. 

\item Infinitesimal supertranslations generate $u$-derivatives of the fields. Thus, in order to have a proper action on the phase space, 
we must assume that the space of shear and the space of news are closed under $u$-differentiation. 

\item Similarly, in order to have a proper infinitesimal action of the Lorentz part of BMS transformations, we also assume closure under the operations
$\sigma \mapsto  u\pu\sigma$, $\sigma \mapsto \pounds_Y\sigma$ and $N\to\pounds_Y N$.

\item We wish to assume that the news and shears approach their respective asymptotic limits \eqref{bc2} in a sufficiently regular manner, to avoid pathological ``angular spikes'' reaching the boundaries. To that end, we require that the limits $N(u,\ang) \to 0$ and $\sigma(u,\ang) \to \sigma_\pm(\ang)$ hold 
within $C^\infty(S^2)$ as $u$ approaches $\pm\infty$.

\item In view of the constraint, we require that for every $\sigma$ in the space of shears, $\pu\sigma$ belongs to the space of news $\bar N$.

\item Finally, as customary in physics, we assume that the spaces of shears and news admit natural Fr\'echet topologies. In this manner, Cauchy sequences of states (and observables) will converge to admissible states (and observables).
\end{enumerate}

A natural choice satisfying all the properties above is the Ashtekar-Streubel space of shears \cite{Ashtekar:1981bq},\footnote{Here we use a combinatorial description of the angular dependence, equivalent to the original analytical definition. We do however need a small modification from the original definition: in the $u$-derivatives, the lowest value of $k$ should be 1 and not 0, in order to allow for non-trivial memory. The family of seminorms that induces the Fr\'echet topology on this space, described in Sec.~\ref{Secfuncanalysis}, is modified in the same way.}
given by the space of functions
\ba\label{caF0def}
\ca F_0 &:= \Big\{ f \in C^\infty(\scri, \bb C_2);\, \text{\it such that } \sup_\scri\left|(1+u^2)^{\frac{k+\varepsilon}{2}}\pu^k \scr D f\right| < \infty   \no
&\hspace{11.5em}\text{\it for all $k \ge 1$, $\scr D \in \text{\rm words}(\eth,\bar\eth)$ and $f_\pm \in \ca E$} \Big\}
\ea
where $C^\infty(\scri, \bb C_s)$ denotes the space of smooth complex, spin-$s$ functions on $\scri$,\footnote{Recall that we are restricting to time-independent frames, so a frame rotation $m^a \mapsto e^{i\psi}m^a$ 
induces a global change by $e^{is\psi}$ in such functions.}
$\varepsilon>0$ is a (fixed) arbitrarily small number, 
$\text{\rm words}(\eth,\bar\eth)$ denote all finite strings of $\eth$ and $\bar\eth$ (e.g., $\scr D = \eth^2\bar\eth \eth$, including $\scr D = 1$),
\be
f_\pm := \lim_{u\to\pm\infty}f
\ee
and
\be
\ca E := \eth^2C^\infty(S^2,\bb R)
\ee
is the space of purely electric, spin-$2$ smooth functions on the sphere.
Thus, we assume 
\[
\sigma \in \ca F_0
\]
in the rest of this paper.
In view of the constraint \eqref{constraint1}, the space of news can, without loss of generality, be defined by $u$-differentiating the space of shears. That is, define
\be\label{caF1def}
\ca F_1 := \pu\ca F_0 
\ee
meaning that $g \in \ca F_1$ if and only if $g = \pu f$ for some $f \in \ca F_0$.\footnote{The labels ``$0$'' and ``$1$'' in $\ca F_i$
are motivated by their relationship via $u$-differentiation.}
More explicitly, we have
\ba\label{caF1def1}
\ca F_1 &= \Big\{ g \in C^\infty(\scri, \bb C_2);\, \text{\it such that } \sup_\scri\left|(1+u^2)^{\frac{k+1+\varepsilon}{2}}\pu^k \scr D g\right| < \infty   \no
&\hspace{10em}\text{\it for all $k \ge 0$, $\scr D \in \text{\rm words}(\eth,\bar\eth)$ and $\int\!du\,g(x)  \in \ca E$} \Big\}
\ea
Accordingly, we take
\[
\bar N \in \ca F_1.
\]
The precise topologies on the spaces $\ca F_i$ are not particularly relevant for most of the paper, but they do play a role in our analysis of the distributional brackets and we will discuss them in
Sec.~\ref{SecDistbrackets}.

In view of the above, the kinematical phase space, $\ca P$, will be taken to be
\be
\ca P = \big\{ (\sigma, \bar N) \in \ca F_0 \times \ca F_1 \big\} \simeq \ca F_0 \times \ca F_1
\ee
This means in particular that the asymptotic behaviors of news and shears are respectively
\be
N = O\left(|1/u|^{1+\varepsilon}\right)\,, \qquad \sigma - \lim_{u\to\pm\infty}\sigma = O\left(|1/u|^{\varepsilon}\right)\,, 
\ee 
as $u\to\pm\infty$. More generally, $\sigma$ is bounded and $\pu^k\sigma$ falls off at least as fast as $|1/u|^{k+\varepsilon}$, for $k\ge 1$, as $u\to\pm\infty$, and $\pu^kN$ falls off at least as fast as $|1/u|^{k+1+\varepsilon}$, for $k\ge 0$, as $u\to\pm\infty$.
The constrained phase space, denoted by $\ca P_\con$, is given by
\be
\ca P_\con = \ca F_0 \times_K \ca F_1 := \big\{ (\sigma, \bar N) \in \ca F_0 \times \ca F_1; \,\,\text{\it satisfying } K := \bar N + \pu\sigma = 0 \big\}
\ee
It is important to keep in mind that the phase spaces, $\ca P$ and $\ca P_K$, and the function spaces $\ca F_i$, should be regarded as real vector spaces,\footnote{The actual radiative phase space is an affine space, but it can be turned into a vector space with a suitable choice of origin in a given gauge \cite{Ashtekar:1981hw}. In our case this step is done by the Bondi coordinates, and the origin is the vanishing of the shear associated with $u$.} even though they are parametrized by complex-valued functions. In particular, note that if $f \in \ca F_0$ has non-zero asymptotic limits, then $af$ belongs to $\ca F_0$ if and only if $a$ is real (because the asymptotic limits must be purely electric). We denote the complex conjugated spaces to $\ca F_i$ by $\bar{\ca F}_i$, so $f \in \bar{\ca F}_i$ if and only if $\bar f \in \ca F_i$. 

We point out that all the results of the paper are expected to hold, without modification, for any sufficiently well-behaved spaces satisfying the properties $1$-$7$ above. It is mainly for the sake of concreteness that we will present the paper in terms of the particular function spaces $\ca F_0$ and $\ca F_1$ introduced above. 
Two examples of ``smaller'' choices, for which our results also apply, are:
\begin{enumerate}[label=\roman*.]
\item A space of shears that is isomorphic to the space of smooth functions on a compactification $\scri$. For example, the map $u\mapsto \arctan(u)$ turns $\bb R$ into $(-\sfrac{\pi}{2}, \sfrac{\pi}{2})$, so the space of shears can be identified with $C^\infty(\scri_c, \bb C_2)$, where $\scri_c := [-\sfrac{\pi}{2}, \sfrac{\pi}{2}]\times S^2$, and $f(\pm\sfrac{\pi}{2}) \in \ca E$;

\item A space of news that is isomorphic to the space of Schwartz functions whose total integral in $u$ belongs to $\ca E$. Then, the space of shears can be taken to be primitives of these Schwartz functions, with purely electric limits as $u\to-\infty$. 
\end{enumerate}
Finally, we emphasize that all the results of the paper still hold when allowing shears with arbitrary asymptotic limits (i.e., including magnetic components). In this case, the space $\ca F_0$ can be taken as in \eqref{caF0def} removing the condition that $f_\pm \in \ca E$, and similarly $\ca F_1$ can be given as in \eqref{caF1def1} without the condition on the total integral of $g$.

\section{Dirac brackets on $\scri$}
\label{SecDirbracketsscri}

In a finite dimensional phase space $\ca P$, with a (real) non-degenerate symplectic form $\Omega$, every smooth function $F \in C^\infty(\ca P, \bb R)$ can be associated with a Hamiltonian vector field, $X_F$, via
\be\label{FXFrel}
\delta F = - I_{X_F}\Omega
\ee
On the other hand, in infinite dimensional phase spaces, the symplectic form is usually only {\it weakly} non-degenerate. 
This means that the map from the tangent spaces to the cotangent spaces, $X \mapsto I_X\Omega$, is injective but not necessarily surjective (i.e., onto). In turn, this implies that every smooth Hamiltonian vector field (i.e., generating a symplectomorphism of $\Omega$) can be (locally) associated with a phase space function via \eqref{FXFrel}, but not all smooth functions can be associated with Hamiltonian vector fields. 
Those functions which can be associated with a smooth Hamiltonian vector field will be called {\it regular functions}. 
Notice that, given two regular functions $F$ and $G$, their Poisson bracket is defined by
\be
\{F, G\} := - \Omega(X_F, X_G)
\ee
It is important to emphasize that there is no reasonable way to define Poisson brackets from the symplectic structure between non-regular functions. 

If the phase space $\ca P$ has constraints, particularly of the second class, the physically relevant brackets are the {\it Dirac brackets}. On-shell, they match with the Poisson brackets on the reduced phase space. The standard formula for Dirac brackets in finite dimensions reads
\be\label{DBdiscrete0}
\{F, G\}_D := \{F, G\} - \sum_{ij} \{F, K_i\} \Delta^{-1}_{ij} \{K_j, G\}
\ee
where
\be
\Delta_{ij} := \{K_i, K_j\}
\ee
is an (invertible) matrix produced from the set of second class constraints $K_i$.

This formula is often assumed to directly extend, in a formal sense, to infinite dimensional phase spaces, with infinitely many constraints, and that may even produce sensible results in some situations.
However, we point out that there are many subtleties with such an extension that need to be addressed in order to make it rigorous. First, one would wish to convert the summation to an integration, which requires one to decide which measure to use. Second, the $\Delta$ matrix becomes a differential operator, and inverting it requires a careful consideration of the boundary conditions. Third, and most importantly, the Poisson brackets involving ``localized'' constraints, $\Delta_{ix,jy} = \{K_i(x), K_j(y)\}$ and $\{F, K_i(x)\}$, are typically not well-defined according to the discussion above.
In App.~\ref{AppDirbrackets}, we describe a natural generalization of the Dirac brackets formula for the infinite dimensional case, in a way that resolves these subtleties.
The procedure can be summarized as follows:
\begin{enumerate}
\item Consider the smeared constraints
\be\label{Ksmeardef}
K_\varphi := \frac{1}{4\pi}\re \int \epsilon_\scri\,\bar\varphi \left(\bar N + \pu \sigma\right)
\ee
where $\varphi$ is a smooth, spin-weight $2$ complex function on $\scri$ such that $K_\varphi$ is well-defined. Then, determine the ``largest'' set of $\varphi$ for which $K_\varphi$ is regular on the kinematical phase space. We denote
\be
\ca K = \text{\it ``Space of smearings for which $K_\varphi$ is regular on $\ca P$''} \nonumber
\ee
It is then possible to define Poisson brackets among $K_\varphi$'s, provided that $\varphi\in \ca K$.

\item Classify the constraints in terms of first versus second class. 
We say that $K$ defines a set of {\it second-class} constraints if, for $\varphi \in \ca K$,
\be\label{secondclass0}
\{K_\varphi, K_{\varphi'}\} \doteq 0\,\,\,\text{\it for all $\varphi' \in \ca K$} \quad \Longleftrightarrow \quad \varphi = 0
\ee

\item Ensure that the flow of functions preserves the constraint surface.
Let $F$ be a regular function on $\ca P$. Define, if possible, its {\it tangential extension} $\wt F$ by
\be
\wt F := F + K_{\varphi_F}
\ee
where $\varphi_F \in \ca K$ is a solution of
\be\label{varphieq0}
\{K_{\varphi_F} ,K_{\varphi'}\} \doteq - \{F, K_{\varphi'}\}\,\,\,\text{\it for all $\varphi' \in \ca K$}
\ee
If the constraints are second-class, a solution $\varphi_F$ is unique provided it exists. If $\varphi_F$ exists, then $F$ is said to be {\it tangentially extensible}.\footnote{In the simplest example of $T^*\R^2$ with the second class constraints $y=p_y=0$, the functions $x$ and $x+y$ coincide on the constraint surface, but only the former defines a flow tangent to it. The tangential extension of the second function yields the former.\label{footTR2}}

\item If $F$ and $G$ are two tangentially extensible functions, their Dirac bracket is defined by
\be
\{F, G\}_D := \{\wt F, \wt G\}
\ee
Note, in particular that $\{F, K_\varphi\}_D \doteq 0$ for all $\varphi \in \ca K$.
\end{enumerate}

\noindent
The property mentioned at the end of Point 4 is the motivation for the nomenclature: $\wt F$ is an extension of $F$, by the addition of a constraint, such that its symplectic flow is tangent to the constraint surface.
This property also implies that, if $F$ and $G$ are tangentially extensible, then it suffices to extend only one of them in the computation of the Dirac brackets, i.e.,
\be
\{F, G\}_D \doteq \{\wt F, \wt G\} \doteq \{\wt F,  G\} \doteq \{ F, \wt G\}
\ee
It is also worth mentioning that, if $F$ and $G$ are tangentially extensible, then $\{F, G\}_D$ is also tangentially extensible.
Thus, the space of tangentially extensible functions forms an algebra under the Dirac brackets. We call those functions {\it proper observables}:
\begin{definition}\label{properobsdef}
A ``proper observable'' is a function on phase space which is tangentially extensible.
\end{definition}
\noindent
As we will see, in Sec.~\ref{SecRedPS}, proper observables are equal, on-shell, to regular functions on the reduced phase space.

We also introduce the notation of {\it directional brackets}. If  $F$ is regular, while $G$ is merely smooth, we can define the directional Poisson brackets by
\be
\{F; G\} := - X_FG
\ee
where a semicolon is used instead of the comma. 
Note that, in general, those are {\it not} Poisson brackets, but simply a notation to indicate how $F$ acts on $G$ (with a minus sign). 
In particular, these brackets will generally not define an algebra, as $\{G; F\}$ may not even make sense.
The convention we use is that the first entry is regular, so
\[
\text{\it the first entry flows the second (with a minus sign)}
\]
The same definition applies also for Dirac brackets. If $F$ is tangentially extensible and $G$ is smooth, the directional Dirac brackets are defined by
\be\label{Dirbrackets0}
\{F; G\}_D := \{\wt F; G\} = -\wt X_FG
\ee
On-shell, this expresses how $F$ acts on $G$ (with a minus sign), while preserving the constraints. If $G$ is also regular or tangentially extensible, then the directional brackets match with the standard brackets.

In this section we will focus on linear phase space functions. Define the {\it shear-news functions} by
\be\label{shearnewsdef}
\Gamma_{\lambda,\alpha} := \re\int\!\epsilon_\scri \left(\bar\lambda \sigma + \alpha  N \right) 
\ee
where $\lambda$ and $\alpha$ are field-independent, smooth, spin-weight $2$ complex functions on $\scri$ such that $\Gamma_{\lambda,\alpha}$ is well-defined. 
We will determine the conditions for their (kinematical) regularity and tangential extensibility. 
The relevant spaces are 
\ba
\ca G_R &= \text{\it ``Space of smearings for which $\Gamma_{\lambda,\alpha}$ is regular on $\ca P$''} \no
\ca G_{TE} &= \text{\it ``Space of smearings for which $\Gamma_{\lambda,\alpha}$ is tangentially extensible''} \nonumber
\ea
The full algebra of proper observables will be discussed in Sec.~\ref{Secalgebra}. The analysis for non-linear functions is essentially the same as for linear functions, since the conditions for regularity and tangential extensibility refer to tangent spaces of each point in phase space (but must hold smoothly on the constraint surface).

\subsection{Constraint regularity and classification}
\label{SecKreg}

The first step in our formalism is to determine the space of smearings, $\ca K$, to produce regular constraints $K_\varphi$, defined as in \eqref{Ksmeardef}.
To ensure that $K_\varphi$ is well-defined it suffices, given the assumptions on the finiteness of the total integral of $|N|$ and $|\pu\sigma|$, that $\varphi$ is bounded.
The reason for taking the real part in \eqref{Ksmeardef} is so that its symplectic flow is tangent to the real phase space.\footnote{That is, it must be true that the flow of $N$ through a parameter $\tau$, $N \mapsto N_\tau$, and of its complex conjugate, $\bar N \mapsto (\bar N)_\tau$, preserves the relation $(\bar N)_\tau = \overline{N_\tau}$; and similarly for $\sigma$. 
Equivalently, if $X$ is the symplectic flow of an observable $F$ (i.e., $\delta F = - I_X\Omega$), then the reality of $F$ (plus the reality of $\Omega$) ensures that $X$ is real, i.e., $\delta\bar N (X) = \overline{\delta N(X)}$ and $\delta\bar\sigma (X) = \overline{\delta\sigma(X)}$.}
Note that $K_\varphi=0$ for all $\varphi$ implies the local constraint \eqref{constraint1} everywhere.

The variation of $K_\varphi$ gives,
\be
\delta K_\varphi := \frac{1}{4\pi}\re \int \epsilon_\scri\,\varphi \left(\delta N + \pu \delta\bar\sigma\right)
\ee
As the symplectic form \eqref{OmASNP} is given by an integral over variations of fields with no derivatives or boundary terms, it is necessary that the expression above can be integrated by parts 
in order for $K_\varphi$ to be associated with a symplectic flow.
This requires, at minimum, the assumption that $\varphi$ has well-defined asymptotic limits.
Then, integration by parts leads, in principle, to boundary terms
\be
\delta K_\varphi := \frac{1}{4\pi}\re \int \epsilon_\scri\, \left(  \varphi \delta N - \pu \varphi \delta\bar\sigma \right) + \frac{1}{4\pi}\re \int_{S^+}\! \epsilon_S\,  \varphi \delta\bar\sigma - \frac{1}{4\pi}\re \int_{S^-}\! \epsilon_S\, \varphi \delta\bar\sigma 
\ee
Let the candidate flow of $K_\varphi$ be denoted simply by $X_\varphi$.
Inserting it into the symplectic form gives
\be
I_{X_\varphi}\Omega = \frac{1}{4\pi} \re\int \epsilon_\scri \left( X_\varphi N\, \delta\sigma - X_\varphi\sigma \delta N \right)
\ee
where $X_\varphi N := \delta N(X_\varphi)$, and similarly for $X_\varphi\sigma$ (or any vector field followed by a function).
We see that, in order for $\delta K_\varphi = - I_{X_\varphi}\Omega$, it must be true that the boundary terms in $\delta K_\varphi$ vanish. 
That is,
\be
\re\int_{S^\pm}\!\epsilon_S\, \varphi \delta\bar\sigma = 0
\ee
As we have that $\sigma$ is purely electric at the boundaries, we conclude 
from \eqref{emort} that
\be\label{fpurelym}
\varphi_\pm\,\, \text{\it are purely magnetic}
\ee
Now, with the boundary terms of $\delta K_\varphi$ vanishing, we find\footnote{In this paper we will always characterize vectors by the way they act on the field variables. Equivalently, a vector $X$ acting as $X\sigma = f$ and $X\bar N = g$ can be expressed in ``components'' as
\[
X = 2\,\re\!\int\!\epsilon_\scri\left(f(x)\frac{\delta}{\delta\sigma(x)} + g(x)\frac{\delta}{\delta\bar N(x)}\right)
\]}
\ba
& X_\varphi\sigma(x) =  \varphi(x) \\
& X_\varphi \bar N(x) = \pu \varphi(x)
\ea
We must also ensure that $X_\varphi$ preserves the boundary conditions. That is, the action above must preserve the respective spaces of shears and news, which requires
\be
\varphi \in \ca F_0
\ee
In particular, this means that $\varphi$ must have purely electric asymptotic limits.
In combination with \eqref{fpurelym}, we conclude that
\be
\varphi_\pm = 0
\ee
Therefore, the maximum allowed space of smearings is
\be
\ca K = \ca F_0^0 :=  \Big\{\varphi \in \ca F_0;\,\, \text{\it with } \varphi_\pm = 0 \Big\}
\ee

Now we wish to verify that these constraints are second-class. We have, for any $\varphi,\varphi'\in\ca K$,
\ba
\{K_\varphi, K_{\varphi'}\} = \delta K_\varphi(X_{\varphi'}) &= \frac{1}{4\pi}\re \int \epsilon_\scri\,\varphi\, \left( X_{\varphi'}N +\pu X_{\varphi'}\bar\sigma \right) \no
&= \frac{1}{2\pi}\re \int \epsilon_\scri\,\varphi\pu \bar \varphi'
\ea
If this vanishes (on-shell) for all $\varphi \in \ca K$, then
\be
\pu  \bar \varphi' = 0
\ee
and the only solution in $\ca K$ is $\varphi' = 0$. Hence, this system of constraints is indeed second-class.

\subsection{Tangential extensions}
\label{SecTEshearnews}

Here we consider the tangential extensibility of the shear-news function, $\Gamma_{\lambda, \alpha}$, defined in \eqref{shearnewsdef}. 
Since $\Omega$ has no boundary terms, $\Gamma_{\lambda, \alpha}$ is the most general form of a linear combination of shear and news that could, in principle, be regular.
In order for this function to be well-defined, we need to assume that 
$\int\!du\,|\lambda(x)| < \infty$ and $\alpha$ is bounded.

The symplectic flow of $\Gamma_{\lambda,\alpha}$, solution of $\delta \Gamma_{\lambda,\alpha} = - I_{X_{\lambda,\alpha}}\Omega$, acts on the fields as\footnote{We shall consistently reserve $\varphi$ for the labels of the constraint, and $(\lambda,\alpha)$ for the labels of the shear-news function. In this way, the labels of the vector fields determine which kind of functions they are associated with.}
\ba
& X_{\lambda,\alpha}\sigma(x) = 4\pi\alpha(x) \\
& X_{\lambda,\alpha} N(x) = -4\pi\bar\lambda(x)
\ea
In order for $X_{\lambda, \alpha}$ to be a regular flow, it needs to preserve the boundary conditions for the shears and news.
The space of smearings corresponding to (kinematically) regular shear-news functions is thus
\be
\ca G_R = \big\{(\lambda,\alpha) \in \ca F_1 \times \ca F_0 \big\} \simeq \ca F_1 \times \ca F_0 
\ee

The tangential extension of $\Gamma_{\lambda,\alpha}$ is given by $\wt\Gamma_{\lambda,\alpha} = \Gamma_{\lambda,\alpha} + K_{\varphi_{\lambda,\alpha}}$, with
\be
\{\Gamma_{\lambda,\alpha}, K_{\varphi'}\} = - \{K_{\varphi_{\lambda,\alpha}}, K_{\varphi'}\} \,,\,\,\text{\it for all  $\varphi'\in \ca K$}
\ee
We have
\be
\re\int\!\epsilon_\scri \left(\bar\lambda \varphi' + \alpha \pu \bar \varphi'\right) = - \frac{1}{2\pi}\re \int \epsilon_\scri\,\varphi_{\lambda,\alpha} \pu \bar \varphi'
\ee
Integrating by parts in $u$, using the properties of $\ca K$, we obtain
\be
\lambda - \pu  \alpha = \frac{1}{2\pi} \pu  \varphi_{\lambda,\alpha}
\ee
whose general solution is
\be
\varphi_{\lambda,\alpha} = 2\pi\left(\Lambda - \alpha +  c\right)
\ee
where $c \in C^\infty(S^2,\bb C_{2})$ is an integration constant and
\be
\Lambda(x) := \int_{-\infty}^u\!du' \lambda(u',\ang)
\ee
in which $x = (u,\ang)$, with $\ang\in S^2$.
Since $\lim_{u\to-\infty} \varphi_{\lambda,\alpha} = 0$, we have
\be
c = \alpha_-
\ee
However, for this $\varphi_{\lambda,\alpha}$ to exist in $\ca K$ we also need $\lim_{u\to+\infty} \varphi_{\lambda,\alpha} = 0$, which requires
\be\label{GTEintcond}
\Lambda_+ = \alpha_+ - \alpha_- 
\ee
where $\Lambda_+(\ang) := \lim_{u\to+\infty}\Lambda(u,\ang)$.
Thus, in order for $\Gamma_{\lambda,\alpha}$ to be tangentially extensible, $(\lambda,\alpha)$ must belong to
\be\label{GTEdef}
\ca G_{TE} = \ca F_1 \times_{\wt K} \ca F_0 := \Big\{ (\lambda, \alpha) \in \ca F_1 \times \ca F_0; \,\,\text{\it satisfying } \wt K = 0 \Big\}
\ee
where 
\be\label{dualconstraint}
\wt K := \int\!du\, \lambda - \alpha_+ + \alpha_-
\ee
can be understood as a ``{\it dual constraint}'' on the space of smearings.

The tangential extension of the shear-news function is then
\be
\wt\Gamma_{\lambda,\alpha} = \Gamma_{\lambda,\alpha} + 2\pi K_{\Lambda - \alpha + \alpha_-}
\ee
for all $(\lambda,\alpha) \in \ca G_{TE}$. 
Note that $\ca G_{TE} \subsetneq \ca G_R$, so not all regular $\Gamma_{\lambda,\alpha}$ can be tangentially extended. 
This is because not all regular functions on the kinematical phase space are guaranteed to restrict to regular functions on the reduced phase space (see Sec.~\ref{SecRedPS}). 
The flow $\wt X_{\lambda, \alpha}$ of $\wt\Gamma_{\lambda, \alpha}$ acts on the fields as
\ba
& \wt X_{\lambda,\alpha}\sigma(x) = 2\pi\left(\Lambda(x) + \alpha(x) + \alpha_-(x)\right) \label{wtXlambdaalphasigma}\\
&\wt  X_{\lambda,\alpha} N(x) = -2\pi \left(\bar\lambda(x) + \pu \bar\alpha(x) \right) 
\label{wtXlambdaalphaN}
\ea

It may be worth pointing out that the condition for tangential extensibility, $(\lambda,\alpha)\in\ca G_{TE}$, implies that ``pure shears'', $\Gamma_{\lambda,0}$, are tangentially extensible only if $\int\!du\,\lambda(x) = 0$ and ``pure news'', $\Gamma_{0,\alpha}$, are tangentially extensible only if $\alpha_+ - \alpha_- = 0$. The ``localized news'', $N(x) \sim \Gamma_{0,\delta_x}$, where $\delta_x(x') := \delta(x',x)$, is technically {\it not} an observable, since $(0,\delta_x) \notin \ca G_{TE}$, due to its irregularity. Nevertheless, it does formally satisfy the dual constraint ($\wt K = 0 - 0 + 0$), so one could perhaps think of it as a ``generalized observable''. However, we do {\it not} find an application, in our analysis, for such objects. (In particular, they play no role in the more useful notion of ``distributional brackets'', described in Sec.~\ref{SecDistbrackets}.) 
On the other hand, the ``localized shear'', $\sigma(x) \sim \Gamma_{\delta_x,0}$, is also not an observable, but not even in such a ``generalized'' sense, since it does not satisfy the dual constraint ($\wt K = 1 - 0 + 0$).\footnote{The same conclusion applies to the ``covariant'', or relative, shear of \cite{Compere:2018ylh,Rignon-Bret:2024mef}.}

\subsection{Dirac brackets for shear-news functions}
\label{SecDBsN}

The Dirac brackets for tangentially extensible shear-news functions are
\ba
\{\Gamma_{\lambda,\alpha}, \Gamma_{\lambda',\alpha'}\}_D &= \{\wt\Gamma_{\lambda,\alpha},\Gamma_{\lambda',\alpha'}\} \no
&= - \delta \Gamma_{\lambda',\alpha'}(\wt X_{\lambda,\alpha}) \no
&= -2\pi \re \int\!\epsilon_\scri\left( \bar\lambda' (\Lambda + \alpha + \alpha_-) - \alpha' (\bar\lambda + \pu \bar\alpha) \right) 
\label{DBshearnewsDir0}
\ea
where $(\lambda,\alpha)$ and $(\lambda',\alpha')$ are in $\ca G_{TE}$. Using the properties of $\ca G_{TE}$, we can rewrite these brackets in an explicitly anti-symmetric form as
\be\label{DBshearnewsAS}
\{\Gamma_{\lambda,\alpha}, \Gamma_{\lambda',\alpha'}\}_D = \pi \re \left[\int\!\epsilon_\scri\left( \bar\lambda \Lambda' + 2\bar\lambda \alpha' - \alpha\pu \bar\alpha' \right) + \int_S\!\epsilon_S\, \alpha_+\bar\alpha'_-\right] - \cdots
\ee
where the ``$\cdots$'' denotes the previous terms with primed and unprimed parameters exchanged.

If only the first entry is tangentially extensible, but the second is merely smooth, we can define directional Dirac brackets as in \eqref{Dirbrackets0}. 
In this case, only formula \eqref{DBshearnewsDir0} holds, while \eqref{DBshearnewsAS} may not. So,
\be\label{DBshearnewsDir}
\{\Gamma_{\lambda,\alpha}; \Gamma_{\lambda',\alpha'}\}_D = -2\pi \re \int\!\epsilon_\scri\left( \bar\lambda' (\Lambda + \alpha + \alpha_-) - \alpha' (\bar\lambda + \pu \bar\alpha) \right) 
\ee
with $(\lambda, \alpha) \in \ca G_{TE}$ and $\Gamma_{\lambda',\alpha'}$ smooth.\footnote{Note that $\Gamma_{\lambda, \alpha}$ can be smooth even if the smearings are not. For example, $\Gamma_{\delta_x,0} = \re\,\sigma(x)$ is smooth. The formula above, for the directional brackets, extends to such smearings. \label{nonsmoothsmear}}

\subsection{Supertranslation flow}
\label{SecSuperTflow}

We are now ready to compute the (directional) Dirac bracket between the supertranslation flux, and the localized shear.
We decompose the supertranslation flux \eqref{STfluxNP} into 
two parts, ``hard'' and ``soft'',
\ba
F_T^\text{\it hard} &:= -\frac{1}{4\pi} \re\int\!\epsilon_\scri\, T N \bar N \\
F_T^\text{\it soft} &:= -\frac{1}{4\pi} \re\int\!\epsilon_\scri\, T \eth^2 N 
\ea
The soft part is the one relevant for the inhomogeneous part of the transformation, and the ``factor of 2'' issue.

The analysis of tangential extensibility of $F_T$ is completely analogous to what we have done for the shear-news functions. Nevertheless, we will present the details here again, so as to make it clear that the non-linearity of $F_T^\text{\it hard}$ does not add any new layers of complexity.
The symplectic flow of $F_T^\text{\it hard}$, denoted by $X_T^\text{\it hard}$, acts on the fields as
\be
X_T^\text{\it hard}\sigma = -2T \bar N \,,\quad
X_T^\text{\it hard}N = 0
\ee
This flow is smooth and respects the boundary conditions since $T\bar N \in \ca F_0$.
The tangential extension, $\wt F_T^\text{\it hard} = F_T^\text{\it hard} + K_{\varphi_T^\text{\it hard}}$, must satisfy
\be
\{F_T^\text{\it hard}, K_{\varphi'}\} = - \{K_{\varphi_T^\text{\it hard}}, K_{\varphi'}\} \,,\,\,\,\text{\it for all $\varphi'\in \ca K$}
\ee
which leads to
\be
\frac{1}{4\pi}\re\int\!\epsilon_\scri\, \pu  \varphi' (-2T N) = -\frac{1}{2\pi}\re\int\!\epsilon_\scri\, \varphi_T^\text{\it hard} \pu \bar\varphi'
\ee
Thus,
\be
\pu\left(\varphi_T^\text{\it hard} - T \bar N\right) = 0
\ee
Since both $N$ and $\varphi_T^\text{\it hard}$ vanish asymptotically, we get
\be
\varphi_T^\text{\it hard} =  T \bar N
\ee
Consequently,
\be
\wt F_T^\text{\it hard} = F_T^\text{\it hard} + K_{T \bar N}
\ee
Notice that, unlike before, the smearing parameter of $K$ is field-dependent. This causes no issues since the Dirac brackets are always to be evaluated on-shell, and the variation of the smearing parameter in $\delta K_\varphi = \frac{1}{4\pi} \re\int\!\epsilon_\scri\, \delta \varphi (N + \pu  \bar\sigma) + \cdots$ never contributes to any relevant computation.

Integrating by parts on the sphere, the soft term becomes simply
\be
F_T^\text{\it soft} = -\frac{1}{4\pi} \re\int\!\epsilon_\scri\, \eth^2T  N = -\frac{1}{4\pi} \Gamma_{0,\eth^2T}
\ee
and, as $T$ is constant in $u$, $(0,\eth^2T) \in \ca G_{TE}$. 
Therefore, both the hard and soft parts of the supertranslation charge are tangentially extensible, which means we are allowed to compute their Dirac brackets with any other tangentially extensible function on phase space, and also their directional Dirac brackets with any smooth function.
In particular, for any smooth $\Gamma_{\lambda,\alpha}$, we have
\ba 
\{F_T^\text{\it hard}; \Gamma_{\lambda,\alpha}\}_D &= \{\wt F_T^\text{\it hard}; \Gamma_{\lambda,\alpha}\} \no
&= -\delta \Gamma_{\lambda,\alpha}( X_T^\text{\it hard}) - \delta \Gamma_{\lambda,\alpha}(X_{\varphi_T^\text{\it hard}}) \no
&= - \re\int\!\epsilon_\scri\, \bar\lambda (- 2T \bar N ) -\re\int\!\epsilon_\scri\,  \left(\bar\lambda(T\bar N) + \bar\alpha\pu(T\bar N)\right) \no
&= \re\int\!\epsilon_\scri\,T N \left(\lambda + \pu\alpha\right) \label{Tharddirbrackets}\\
\{F_T^\text{\it soft}; \Gamma_{\lambda,\alpha}\}_D &= - \frac{1}{4\pi}\{\Gamma_{0,\eth^2T}; \Gamma_{\lambda,\alpha}\}_D \no
&= \frac{1}{2} \re\int\!\epsilon_\scri\left(\bar\lambda \eth^2(T + T_-) - \alpha \pu\eth^2T\right) \no
&=  \re\int\!\epsilon_\scri\,\bar\lambda\, \eth^2T \label{Tsoftdirbrackets}
\ea
where we have used that $T$ is $u$-independent, so $T = T_-$.
Putting it together we get
\be\label{Tdirbrackets}
\{F_T; \Gamma_{\lambda,\alpha}\}_D =  \re\int\!\epsilon_\scri\,T N \left(\lambda + \pu\alpha\right) + \re\int_S\!\epsilon_S\,\bar\Lambda_+\, \eth^2T
\ee

We may also consider localized (complex) shears,
\be
\sigma(x) = \Gamma_{\delta_x,0} + i \Gamma_{i\delta_x,0}
\ee
where $\delta_x(x') := \delta(x', x)$.\footnote{According to footnote~\ref{nonsmoothsmear}, formula \eqref{DBshearnewsDir} for the directional brackets extends to such smearings. Also, as discussed in Sec.~\ref{SecComplexB}, the brackets can be extended to complex-valued functions (modulo a subtlety explained there).}
Then,
\be\label{Txisigmax}
\{F_T; \sigma(x)\}_D = T(x)\bar N(x) + \eth^2T(x)
\ee
Note that, on-shell, this matches with $-\delta_T\sigma$, given in \eqref{STfluxNP}.

To make it more evident, let us rewrite the expression above in its tensorial form. From \eqref{sabs} we have
\ba
\{F_T; \sigma_{ab}(x)\}_D &= - \{F_T; \bar\sigma\}_D m_a m_b - \{F_T; \sigma\}_D \bar m_a \bar m_b \no
&= - ( T N + \bar\eth^2T) m_a m_b - ( T \bar N + \eth^2T) \bar m_a \bar m_b \no
&= -\frac{1}{2} T N_{ab} - (\bar m^{a'} \bar m^{b'} D_{a'}D_{b'}T\, m_a m_b + m^{a'} m^{b'} D_{a'}D_{b'}T \,\bar m_a \bar m_b) \no
&= -\frac{1}{2} T N_{ab} - D_{\langle a} D_{b \rangle}T \no
&= -T \pu \sigma_{ab} - D_{\langle a} D_{b \rangle}T
\ea
where $D_{\langle a} D_{b \rangle}T := D_a D_bT - \frac{1}{2} D_c D^cT\, q_{ab}$ is the trace-free part of $D_a D_bT$. (The argument, $x$, of all local functions on the right column is omitted.)
To obtain the second line we conjugated \eqref{Txisigmax}, using the reality of $F_T$ (and of the Dirac brackets); for the third line we used \eqref{NabN} and the fact that $\eth^2T = m^a m^b D_a D_bT$, since $T$ has spin-weight $0$; for the fourth line, one can verify the equality between the quantity in parenthesis with $D_{\langle a} D_{b \rangle}T$ by simply contracting both objects with all pairs of $m^a$ and $\bar m^a$; for the last line, we just used the constraint.
This is the right transformation law for the shear under supertranslations.
In terms of the variable, $C_{ab} := -2\sigma_{ab}$, sometimes used in the literature, we have
\be
\{F_T; C_{ab}\}_D = -T \pu  C_{ab} + 2 D_{\langle a} D_{b \rangle}T
\ee
This is the ``factor of $2$'', in front of the inhomogeneous term, missed in the original analysis of \cite{He:2014laa}.
In the discussion, Sec.~\ref{SecDiscussion}, we will comment on the underlying reason for why our analysis obtains the correct factor.

\section{Reduced phase space}
\label{SecRedPS}

The reduced phase space, $\ca P_R$, is the space of physical degrees of freedom, in which the constraints are satisfied and any gauge is quotiented out. 
For second-class constraints, there is no gauge, so the reduced phase space can be identified with the constraint surface itself,
\be
\ca P_R \simeq \ca P_K
\ee
The symplectic form $\Omega_R$ on $\ca P_R$ is simply the restriction of $\Omega$ to the constraint surface,
\be
\Omega_R = \Omega\big|_{\ca P_K}
\ee
The Dirac brackets should be, in principle, equivalent (on-shell) to the Poisson brackets associated with the induced symplectic form on the reduced phase space.
Here we show that the same results from the previous section, Sec.~\ref{SecDirbracketsscri}, can be obtained from the perspective of the reduced phase space.
We can see this as a check of consistency for our Dirac brackets formalism.

Since $\sigma$ uniquely determines $N$, on-shell, the reduced phase space can be parametrized by $\sigma$ alone. 
In fact, $\ca P_R \simeq \ca F_0$.
The (induced) symplectic form is
\be\label{wtOmega}
\Omega_R = -\frac{1}{4\pi} \re\int \epsilon_\scri\,  \pu \delta \bar\sigma\, \delta\sigma
\ee
Inserting a (real) smooth vector field $V$ into $\Omega_R$, we get
\ba
I_{V} \Omega_R &= -\frac{1}{4\pi} \re\int \epsilon_\scri \left( \pu  (V \bar\sigma)\, \delta\sigma - V \sigma\, \pu \delta\bar\sigma \right) \no
&=  -\frac{1}{4\pi} \re\int \epsilon_\scri \left( \pu  (V \bar\sigma)\, \delta\sigma - V \bar\sigma\, \pu \delta\sigma \right) \no
&=  -\frac{1}{2\pi} \re\int \epsilon_\scri\, \pu  (V \bar\sigma) \,\delta\sigma +\frac{1}{4\pi} \re\int_{S^+}\!\epsilon_S\, V \bar\sigma\, \delta\sigma - \frac{1}{4\pi} \re\int_{S^-}\!\epsilon_S\, V \bar\sigma\, \delta\sigma 
\label{IwtVwtOmega}
\ea
To verify that this is (weakly) non-degenerate, we must show that the only vector $V$ which annihilates $\Omega_R$ is zero. 
The bulk condition implies
$\pu (V \bar\sigma) = 0$, that is, $V\bar\sigma$ should be constant in $u$.
The vanishing of the boundary terms requires that $V\bar\sigma$ must be purely magnetic asymptotically, and the preservation of the boundary conditions requires it to be purely electric asymptotically, and thus it must vanish asymptotically. Consequently, the only solution is $V = 0$, confirming the (weak) non-degeneracy.

Consider the shear-news function $\Gamma_{\lambda,\alpha}$, which is expressed on-shell as
\be
\Gamma_{\lambda,\alpha} \doteq \re\int\!\epsilon_\scri \left(\bar\lambda\sigma - \alpha\pu \bar\sigma \right)  
\ee
This is the most general linear function of the shear compatible with the symplectic form \eqref{wtOmega}, since $\Omega_R$ is an integral of differentials of $\sigma$ and $\pu\sigma$.
Its variation gives
\be
\delta\Gamma_{\lambda,\alpha} = \re\int\!\epsilon_\scri \left(\bar\lambda + \pu \bar\alpha \right) \delta\sigma - \re\int_S\!\epsilon_S\, \bar\alpha_+ \delta\sigma_+ + \re\int_S\!\epsilon_S\, \bar\alpha_- \delta\sigma_-
\ee
where we have assumed that $\alpha$ has well-defined asymptotic limits, $\alpha_\pm$. This assumption is necessary for regularity, given the form of \eqref{IwtVwtOmega}.
Note that the pair of parameters $(\lambda,\alpha)$ defining $\Gamma_{\lambda,\alpha}$ is over-complete, on-shell, in the sense that the transformation
\be\label{wgauge}
(\lambda, \alpha) \mapsto (\lambda + \pu w, \alpha - w)
\ee
for any smooth function $w$ with well-defined {\it purely magnetic} asymptotic limits,
does not change the shear-news function. In fact,
\be
\Gamma_{\lambda + \pu w,\alpha-w} =  \Gamma_{\lambda,\alpha} + \re\int_S\!\epsilon_S\, \bar w_+\sigma_+ - \re\int_S\!\epsilon_S\, \bar w_-\sigma_- = \Gamma_{\lambda,\alpha}
\ee
We can use this arbitrariness to impose, without loss of generality, that $\alpha$ has {\it purely electric} asymptotic limits.

The symplectic flow of $\Gamma_{\lambda,\alpha}$, $V_{\lambda,\alpha}$, must then satisfy
\ba
& \lambda + \pu \alpha = \frac{1}{2\pi} \pu (V_{\lambda,\alpha}\sigma) \label{redbulkeq}\\
&\alpha_+ - a_+ = \frac{1}{4\pi} (V_{\lambda,\alpha}\sigma)_+ \label{redboundeq+}\\
& \alpha_- - a_- =  \frac{1}{4\pi} (V_{\lambda,\alpha}\sigma)_-  \label{redboundeq-}
\ea
where $a_\pm$ are two arbitrary, purely magnetic scalars on $S^\pm$. Since $(V_{\lambda,\alpha}\bar\sigma)_\pm$  and $\alpha_\pm$ must be purely electric, we need $a_\pm = 0$. 
Integrating \eqref{redbulkeq} we get
\be\label{wtVsnw}
V_{\lambda,\alpha}\sigma = 2\pi \left(\Lambda + \alpha + c\right)
\ee
where $c$ is $u$-independent. Now taking the limit $u\to-\infty$, and using \eqref{redboundeq-}, we find
\be
c = \alpha_- 
\ee
Taking the limit $u\to+\infty$, and using \eqref{redboundeq+}, we obtain
\be\label{redsnparcond}
\Lambda_+ = \alpha_+ - \alpha_-
\ee
That is precisely the condition defining $\ca G_{TE}$. 
Here it appears as the condition for $\Gamma_{\lambda,\alpha}$ to be regular (associated to the smooth flow $V_{\lambda,\alpha}$).

Note, however, that in the reduced phase space there are no separate conditions on $\lambda$ and $\alpha$, as the boundary conditions are preserved provided only that $\Lambda + \alpha + \alpha_- \in \ca F_0$ or, equivalently (assuming that $\alpha_- \in \ca E$), $\lambda + \pu\alpha \in \ca F_1$.
Nevertheless, this does not mean that there are ``more'' regular shear-news functions on the reduced phase space than there are tangentially extensible shear-news functions on the kinematical phase space, from the perspective of their action on-shell. In fact, the action of $\Gamma_{\lambda,\alpha}$ on the reduced phase space is given, as in \eqref{wtVsnw}, precisely by the combination $\Lambda + \alpha + \alpha_-$, and via a ``gauge'' transformation of the kind \eqref{wgauge} one could always enforce that $\lambda \in \ca F_1$ and $\alpha \in\ca F_0$, because of the condition \eqref{redsnparcond}. 
We thus confirm that regularity in the reduced phase space is equivalent, on-shell, to tangential extensibility in the kinematical phase space.

We can compute the (reduced) Poisson brackets of shear-news functions
\ba
\{\Gamma_{\lambda,\alpha}, \Gamma_{\lambda',\alpha'}\}_R &= - \delta\Gamma_{\lambda',\alpha'} (V_{\lambda,\alpha}) \no
&= - \re\int\!\epsilon_\scri \left(\bar\lambda' V_{\lambda,\alpha}\sigma - \bar\alpha'\pu (V_{\lambda,\alpha}\sigma) \right)  \no
&= -2\pi \re\int\!\epsilon_\scri \left(\bar\lambda' (\Lambda + \alpha + \alpha_-) - \bar\alpha'(\lambda + \pu \alpha) \right) 
\ea
where $(\lambda,\alpha)$ and $(\lambda',\alpha')$ are in $\ca G_{TE}$.
Note that this matches exactly with the Dirac brackets computed in \eqref{DBshearnewsDir0}.
Or, if $\Gamma_{\lambda',\alpha'}$ is merely smooth, it also matches with the directional formula \eqref{DBshearnewsDir}.
Hence, this confirms the validity of the Dirac brackets construction.

It is worth mentioning that the set of shear-news functions with $\lambda = \pu\alpha$, $\Gamma_{\pu\alpha,\alpha}$, is  equivalent to the set of all shear-news functions in the reduced phase space.
As mentioned in \eqref{wgauge}, the smearing parameters $(\lambda,\alpha)$ are over-complete on-shell, and in that discussion we only made use of the boundary values of $w$ to set $\alpha_\pm$ to purely electric values asymptotically. Using the remaining freedom (i.e., $w$ satisfying $w_\pm=0$) we can require that $\alpha - w = \beta$ and $\lambda + \pu w = \pu\beta$. This always has a (unique) solution $w = -\frac{1}{2}\left(\Lambda - \alpha + \alpha_-\right)$, implying that
\be\label{overcompletelabels}
\Gamma_{\lambda, \alpha} \doteq \Gamma_{\pu\beta, \beta} \,,\quad \text{\it with } \beta = \frac{1}{2}\left(\Lambda + \alpha + \alpha_-\right).
\ee
In fact, this parametrization of the shear-news functions, by a single parameter $\beta$, is the most natural from the perspective of the reduced phase space.
Instead of starting with a function and asking for its regularity conditions, one could directly start with a vector field and ask for conditions that ensure it is Hamiltonian. Let $V_\beta$ be the vector tangent to the reduced phase space acting as
\be
V\sigma = 4\pi\beta
\ee
where $\beta \in \ca F_0$ is field-independent. Inserting it into the symplectic form gives
\ba
- I_V\Omega_R &= \re\int\!\epsilon_\scri\left(\pu\bar\beta \delta\sigma - \bar\beta \pu\delta\sigma \right) \no
&= \delta\left[\re\int\!\epsilon_\scri\left(\pu\bar\beta \sigma - \bar\beta \pu\sigma\right)\right] \no
&= \delta\Gamma_{\pu\beta,\beta}
\ea
showing that this vector is Hamiltonian and associated with $\Gamma_{\pu\beta,\beta}$.

\section{Algebra of proper observables}
\label{Secalgebra}

In classical physics, with a finite number of degrees of freedom, observables are generally defined as real (smooth) functions on the phase space. If gauge is present, then we also require that these functions are constant along the gauge orbits. All observables then form an algebra with product given by the Poisson (or Dirac, in the presence of second-class constraints) brackets. However, for theories with an infinite number of degrees of freedom, we have emphasized that not all smooth functions can be associated with a symplectic flow, and therefore they generally will not form an algebra. Only the regular (or tangentially extensible) ones have well-defined Poisson (or Dirac) brackets, and they form an algebra.\footnote{If $F$ and $G$ are two functions symplectically associated with vector fields $X_F$ and $X_G$, then $\{F,G\}$ is associated with the vector field $-[X_F, X_G]$.} 
The {\it algebra of proper observables}, with product given by the Dirac brackets, will be called $\scr A$.

In quantum mechanics, observables are (self-adjoint) operators on a Hilbert space, and they have a natural $C^*$-algebra structure, which includes all possible commutators between operators. From a quantization perspective, it is assumed that the commutator between two quantum observables should agree, at first order in $\hbar$, with the Poisson (or Dirac) brackets of the corresponding classical observables. In view of this, it is natural to consider that a quantum observable will only have a classical counterpart if it corresponds to a proper observable. Therefore, we expect that proper observables are the only ones relevant for quantization, and in this section we discuss some properties of their algebra. In particular, we will show that the shear-news functions form a complete subset of proper observables, and that any proper observable that does not commute with the memory is capable of ``measuring'' the Goldstone mode (even though the Goldstone mode itself is not a proper observable).

\subsection{The full algebra}
\label{SecEntireA}

Consider a smooth real function on the phase space, $F: \ca P\to\bb R$. The procedure for determining the conditions for its tangential extensibility is exactly analogous to what was done for the shear-news functions in Sec.~\ref{SecTEshearnews}, and further illustrated with the supertranslation charge in Sec.~\ref{SecSuperTflow}. 
In particular, in order for it to be (kinematically) regular, given the symplectic structure \eqref{OmASNP}, it must have well-defined, smooth functional derivatives (without boundary terms),
\be\label{properobsfuncder}
\delta F = \int\!\epsilon_\scri \left( \frac{\delta F}{\delta\sigma(x)}\delta\sigma(x) + \frac{\delta F}{\delta\bar N(x)}\delta\bar N(x) + \frac{\delta F}{\delta\bar\sigma(x)}\delta\bar\sigma(x) + \frac{\delta F}{\delta N(x)}\delta N(x) \right)
\ee
The reality condition implies
\be
\frac{\delta F}{\delta\bar\sigma(x)} = \ol{\frac{\delta F}{\delta\sigma(x)}} \quad \text{\it and }\quad \frac{\delta F}{\delta\bar N(x)} = \ol{\frac{\delta F}{\delta N(x)}}
\ee
It is tangentially extensible if and only if
\be\label{FTEcond}
p \in \ca P_K \mapsto \left(\frac{\delta F}{\delta\bar\sigma},\, \frac{\delta F}{\delta N}\right)\!\bigg|_p \in \ca G_{TE} \,\,\, \text{\it is smooth}
\ee
where the functional derivatives are evaluated at each point $p$ of the constrained phase space. 
The corresponding flow of its tangential extension is, analogously to \eqref{wtXlambdaalphasigma} and \eqref{wtXlambdaalphaN},
\ba
& \wt X_F\sigma(x) = 4\pi\left(\frac{\delta F}{\delta\bar\Sigma} + \frac{\delta F}{\delta N} + \frac{\delta F}{\delta N_-}\right) \\
&\wt  X_F N(x) = -4\pi \left(\frac{\delta F}{\delta\sigma} + \pu \frac{\delta F}{\delta\bar N} \right) 
\ea
where
\ba
\frac{\delta}{\delta\Sigma(x)} &:= \int_{-\infty}^u\!du'\, \frac{\delta}{\delta\sigma(u',\ang)} \\
\frac{\delta}{\delta N_\pm(\ang)} &:= \lim_{u\to\pm\infty} \frac{\delta}{\delta N(u,\ang)}
\ea
and similarly for conjugations.
Given two tangentially extensible functions, $F$ and $G$, their Dirac bracket is, analogously to \eqref{DBshearnewsAS},
\be\label{DBcompalg}
\{F, G\}_D = 4\pi \re \left[\int\!\epsilon_\scri\!\left( \frac{\delta F}{\delta\bar\sigma} \frac{\delta G}{\delta\Sigma} + 2\frac{\delta F}{\delta\sigma} \frac{\delta G}{\delta N} - \frac{\delta F}{\delta N}\pu\frac{\delta G}{\delta\bar N} \right) + \int_S\!\epsilon_S\, \frac{\delta F}{\delta N_+} \frac{\delta G}{\delta\bar N_-}\right] - \cdots
\ee
where the ``$\cdots$'' corresponds to the same expression except with $F$ and $G$ swapped.
Notice that for the shear-news function we have
\be
\frac{\delta \Gamma_{\lambda,\alpha}}{\delta\bar\sigma} = \frac{\lambda}{2} \qquad\text{\it and }\qquad \frac{\delta \Gamma_{\lambda,\alpha}}{\delta N} = \frac{\alpha}{2}
\ee
and the formula above reduces to \eqref{DBshearnewsAS} when $F = \Gamma_{\lambda, \alpha}$ and $G = \Gamma_{\lambda',\alpha'}$.

The full algebra of proper observables is thus
\be
\scr A = \left\{ F \in C^\infty(\ca P, \bb R);\, F\, \text{\it satisfies \eqref{properobsfuncder} and \eqref{FTEcond}}\right\}
\ee
with the product given by the Dirac brackets,
\be
F * G := \{F,G\}_D
\ee
expressed in \eqref{DBcompalg}.

\subsection{Memory and Goldstone probes}
\label{SecMemGP}

Each point in the physical phase space, namely the constraint surface $\ca P_K$, is characterized by a shear $\sigma$ with well-defined purely electric asymptotic limits $\sigma_\pm$. 
It is of interest to describe observables associated with these asymptotic modes (also called ``soft modes'').

Let us refer to the difference, $\sigma_+ - \sigma_-$, as the {\it memory mode}.
We define the {\it memory observables} by
\be\label{memoryonshell2}
\mem_\xi := \re\int \epsilon_\scri\, \eth^2\xi \, N = \Gamma_{0,\eth^2\xi}
\ee
where $\xi \in C^\infty(S^2, \bb R)$. Notice that, on-shell, this is equal to\footnote{Due to the boundary conditions of $\sigma_+ - \sigma_-$ being purely electric, we can drop the real part operator. In fact, as $\sigma_+ - \sigma_-$ can be written as $\eth^2\varsigma$ for some real scalar $\varsigma$, 
\[
\mem_\xi \doteq - \re\int \epsilon_S\, \bar\eth^2\xi \,\eth^2\varsigma = - \re\int \epsilon_S\, \eth^2\bar\eth^2\xi \,\varsigma = - \int \epsilon_S\, \eth^2\bar\eth^2(\re\,\xi) \,\varsigma = - \int \epsilon_S\, \bar\eth^2\xi \left(\sigma_+ - \sigma_-\right)
\]\label{footnotememreal}} 
\be\label{memoryonshell}
\mem_\xi \doteq - \re\int \epsilon_S\, \bar\eth^2\xi \left(\sigma_+ - \sigma_-\right)
\ee
Even though these two functions coincide on the reduced phase space, only \eqref{memoryonshell2} is a proper observable (i.e., tangentially extensible), whereas \eqref{memoryonshell} even fails to be regular (since its variation is supported at the boundaries).
We can say that smearings of the memory mode, as in \eqref{memoryonshell}, can be {\it directly measured} by memory observables. More generally, we say that a function is {\it directly measured by a proper observable} if it is equal on-shell to a proper observable.

The important question raised in \cite{He:2014laa} is whether there are variables conjugated to the memory. The natural candidate considered there was $\s_+$, referred to as the {\it Goldstone mode}. (Alternatively, $\sigma_-$ or the sum $\sigma_+ + \sigma_-$ can also be taken as the ``counterpart'' of the memory; naturally, any two independent linear combinations of $\sigma_+$ and $\sigma_-$ are sufficient for describing the asymptotic modes).
Functions of $\sigma_+$ alone would, similarly to \eqref{memoryonshell}, fail to be regular and therefore they cannot be proper observables. The natural question is whether any such function of $\sigma_+$ (or $\sigma_-$ or $\sigma_+ + \sigma_-$) could be directly measured by proper observables. The answer, unlike the case for the memory, is negative. In fact, we have the following:
\begin{theorem}\label{nogravitonobs}
The only (smooth, real) functionals of $\sigma_+$ and $\sigma_-$ that are equal (on-shell) to proper observables are functionals of the memory mode alone.
\end{theorem}
\begin{proof}
Let $\Theta(\sigma_+, \sigma_-)$ be a real-valued smooth functional of the asymptotic modes and assume that on the constraint surface $\ca P_K$ it is equal to a proper observable $F$. The exterior derivative of their difference must be annihilated by any vector tangent to $\ca P_K$, and from Lemma~\ref{lemma1formPK} (see App.~\ref{AppLemma1formPK}) we have 
\be\label{ThetaFCf}
\delta\left(\Theta - F\right) \doteq \delta K_\varphi
\ee
for some integral representative $\varphi$ of a (real, continuous) distribution on $\ca F_1$, which can depend on the dynamical variables.\footnote{The topology of $\ca F_1$ and associated distributions will be discussed in Sec.~\ref{Secfuncanalysis}. However, for the purposes of this proof, it is acceptable to think of $\varphi$ as a spin-weight $2$ complex-valued function on $\scri$ such that $K_\varphi$ is well-defined and smooth. Notice that $\varphi$ need not be a function in $\ca K$, since we are not requiring that $K_\varphi$ generates a regular symplectic flow. Here $\doteq$ means that the forms are based at points in $\ca P_K$, but the variations are unrestricted.}
Since $\Theta$ depends only on the asymptotic modes, its differential must be concentrated at the asymptotic boundaries,
\be
\delta\Theta = \re\left[\int_{S_+}\!\epsilon_S\, \vartheta_+(\ang) \delta\sigma_+(\ang) + \int_{S_-}\!\epsilon_S\, \vartheta_-(\ang) \delta\sigma_-(\ang) \right]
\ee
for some (possibly field-dependent) distributions $\vartheta_\pm$ on $\ca E$.
As we see from \eqref{properobsfuncder}, the differential of $F$ contains no boundary terms
\be
\delta F = 2\,\re\int\!\epsilon_\scri \left( \frac{\delta F}{\delta\sigma(x)}\delta\sigma(x) + \frac{\delta F}{\delta\bar N(x)}\delta\bar N(x) \right)
\ee
In order for \eqref{ThetaFCf} to possibly have a solution, the differential of $K_\varphi$ needs to have a similar form, so $\varphi$ must be associated with a smooth function with well-defined asymptotic limits, thereby allowing integration by parts. We then have
\be
\delta K_\varphi \doteq \frac{1}{4\pi}\re \int\! \epsilon_\scri \left(\bar\varphi\delta\bar N - \pu \bar\varphi \delta\sigma\right) + \frac{1}{4\pi}\re\left[\int_{S_+}\!\epsilon_S\, \bar\varphi_+ \delta\sigma_+ - \int_{S_-}\!\epsilon_S\, \bar\varphi_- \delta\sigma_- \right]
\ee
We obtain the following equations (on-shell)
\ba
\frac{\delta F}{\delta\bar N(x)} = - \frac{1}{8\pi}\bar\varphi \qquad &\text{\it and }\qquad \frac{\delta F}{\delta\sigma(x)} = \frac{1}{8\pi}\pu \bar\varphi \\
\vartheta_+ = \frac{1}{4\pi}\bar\varphi_+\big|_{\ca E} \qquad &\text{\it and }\qquad \vartheta_- = -\frac{1}{4\pi}\bar\varphi_-\big|_{\ca E}\label{varthetapmeq}
\ea
But notice that
\be
\frac{\delta F}{\delta\bar N_+} - \frac{\delta F}{\delta\bar N_-}  = - \frac{1}{8\pi}\left(\bar\varphi_+ - \bar\varphi_-\right) \quad\,\, \text{\it and }\quad\,\, \int_{-\infty}^{+\infty}\!du\,\frac{\delta F}{\delta\sigma(x)} = \frac{1}{8\pi}\left(\bar\varphi_+ - \bar\varphi_-\right)
\ee
which implies, from \eqref{FTEcond}, that
\be
\varphi_+ = \varphi_-
\ee
Equations \eqref{varthetapmeq} then yield
\be
\vartheta_+ = - \vartheta_- =: \vartheta
\ee
We conclude that
\be
\delta\Theta \doteq \re\int_{S}\!\epsilon_S\, \vartheta\, \delta\!\left(\sigma_+ - \sigma_-\right)
\ee
meaning that $\Theta$ must depend only on the memory mode.
\end{proof}

Establishing that the natural and promising concept of ``Goldstone mode" is actually not a proper observable is one of the main results of our paper. The next result is that notwithstanding the lack of a Goldstone mode observable, memory is not central. As we now show, it is in fact possible to find infinitely many proper observables that do not commute with memory. The point is that there is not a canonical one that can be singled out, hence the question of ``what is the quantity conjugated to memory'' turns out to have a more subtle answer.  
This leads us to the concept of {\it Goldstone probes},
\begin{definition}
A ``Goldstone probe'' is any proper observable that is capable of indirectly measuring the Goldstone mode, in the sense that its value is sensitive to changes in the state along (at least some) Goldstone directions. More precisely, $\Pi : \ca P \to \bb R$ is a Goldstone probe if and only if
\ba
&\quad i)\,\,\text{$\Pi$ is a proper observable;} \no
&\quad\!\! ii)\,\,\text{$\delta\Pi(V) \neq 0$ for some vector field $V$, tangent to $\ca P_K$, satisfying $\pu(V\sigma) = 0$.} \nonumber
\ea
\end{definition}
Given a vector $V$ tangent to $\ca P_K$, satisfying $\pu(V\sigma) = 0$, we have
\ba
\delta \Pi(V) &= 2\,\re\int\!\epsilon_\scri \left( \frac{\delta \Pi}{\delta\sigma}\, V\sigma + \frac{\delta \Pi}{\delta\bar N}\, V\bar N \right) \no
&\doteq 2\,\re\int\!\epsilon_\scri \,\frac{\delta \Pi}{\delta\sigma}\, V\sigma \no
&= 2\,\re\int_S\!\epsilon_S \left(\int_{-\infty}^{+\infty}\!du\,\frac{\delta \Pi}{\delta\sigma}\right) V\sigma \label{deltaPiV}
\ea
Since $\Pi$ is a proper observable, we conclude from \eqref{FTEcond} that it is a Goldstone probe if and only if
\be\label{GPprop1}
\int_{-\infty}^{+\infty}\!du\,\frac{\delta \Pi}{\delta\sigma} = \frac{\delta \Pi}{\delta\bar N_+} - \frac{\delta \Pi}{\delta\bar N_-} \ne 0
\ee
In particular, this means that there are {\it no} Goldstone probes constructed only from the shear or only from the news.
Note also that
\be
\{\Pi, \mem_\xi\}_D = 8\pi \re \int_S \epsilon_S\, \eth^2\xi \left( \frac{\delta \Pi}{\delta\bar N_+} - \frac{\delta \Pi}{\delta\bar N_-}\right)
\ee
which means that $\Pi$ is a Goldstone probe if and only if it does not commute with a memory observable.

A simple class of Goldstone probes is provided by shear-news functions.
As \eqref{GPprop1} shows, a shear-news function $\Gamma_{\lambda,\alpha}$ is a Goldstone probe if and only if $\alpha_+ \ne \alpha_-$.\footnote{Upon using the constraint and integrating by parts the news term, the resulting functions contain both a ``hard'' shear contribution and endpoint (soft) shear terms. A structurally similar mixed ``hard-soft'' variable was introduced in electromagnetism in \cite{flanagan2023classical}, as part of a choice of phase-space coordinates that symplectically separates hard modes from a memory-conjugate pair. Here, instead, such mixed functions arise from the characterization of proper observables and form an infinite family of Goldstone probes.}
Note that there is no ``canonical'' conjugate to the memory, but any such shear-news function is a conjugate to memory observables,
\be
\{\Gamma_{\lambda,\alpha}, \mem_\xi\}_D = 4\pi \re \int_S \epsilon_S\, \eth^2\xi \left( \bar\alpha_+ - \bar\alpha_-\right)
\ee

In order to measure the Goldstone mode, $\sigma_+$, to arbitrary precision, it is necessary to consider an infinite number of Goldstone probes. Nevertheless, there are certain families of Goldstone probes that are more efficient in measuring the soft modes.
For example, consider the sequence
\be
\Pi_{\tau,u_0,n} := \Gamma_{\lambda_n^{\tau,u_0}, \alpha_n^{\tau,u_0}}\,,\quad n \in \bb N
\ee
with
\ba
\lambda_n^{\tau,u_0}(x) &= \frac{n}{\sqrt{\pi}}e^{-n^2(u-u_0)^2} \,\eth^2\tau \\
\alpha_n^{\tau,u_0}(x) &= \frac{1}{1 + e^{-n(u-u_0)}} \,\eth^2\tau
\ea
where $\tau \in C^\infty(S,\bb R)$ and $\eth^2\tau \ne 0$. For each $n$, $(\lambda_n^{\tau,u_0}, \alpha_n^{\tau,u_0}) \in \ca G_{TE}$, so $\Pi_{n,u_0}$ is a proper observable. Moreover, $(\alpha_n^{\tau,u_0})_+ - (\alpha_n^{\tau,u_0})_- = \eth^2\tau$, so $\Pi_{\tau,u_0,n}$ are Goldstone probes, whose bracket with the memory gives
\be
\{\Pi_{\tau,u_0,n}, \mem_\xi\}_D = 4\pi \re \int_S \epsilon_S\, \eth^2\xi \,\bar\eth^2\tau
\ee
In the limit $n\to\infty$, $\lambda_n^{\tau,u_0}$ tends to a Dirac delta centered at $u_0$, and $\alpha_n^{\tau,u_0}$ tends to a Heaviside function supported on $u>u_0$. Thus, 
\be
\Pi_{\tau,u_0} := \lim_{n\to\infty}\Pi_{\tau,u_0,n} = \re\int_S\!\epsilon_S\, \bar\eth^2\tau \left( \sigma(u_0,\ang) + \int_{u_0}^{+\infty}\!du\, \bar N(u,\ang) \right) 
\ee
This can be interpreted as a {\it dressed} version of the shear at a fixed time $u=u_0$.
Since the boundary conditions assume that $N$ vanishes asymptotically, the ``dressing'' term contributes little to the value of $\Pi_{u_0}$ at large $u_0$.
Note that it is not (strictly) a proper observable because the symplectic flow it would generate is not smooth, but rather a ``generalized observable'' in the same sense as the local news (as discussed at the end of Sec.~\ref{SecTEshearnews}).
In the limit $u_0 \to\infty$ we get
\be
\Pi_{\tau,+} :=  \lim_{u_0\to\infty}\Pi_{\tau,u_0} = \re\int_S\!\epsilon_S\, \bar\eth^2\tau \, \sigma_+ 
\ee
This is not an observable, even in the weak sense (i.e., relaxing the condition of smoothness), because it does not satisfy the dual constraint \eqref{dualconstraint}.
The lesson, nevertheless, is that one should be able to measure the value of $\sigma_+$, to a fairly good degree of accuracy, by measuring some $\Pi_{\tau,u_0,n}$ with sufficiently large $n$ and $u_0$. 

One should {\it not} be tempted to interpret the construction above as a way of ``defining'' the brackets $\{\Pi_{\tau,+},\,\cdot\,\}_D$ via $\lim_{u_0\to\infty} \lim_{n\to\infty} \{\Pi_{\tau,u_0,n},\,\cdot\,\}_D$. 
In fact, there are infinitely many other families of proper observables that also converge pointwise to $\Pi_{\tau,+}$ but define different brackets in the limit. For example, let $G$ be any other Goldstone probe and define the sequence
\be
\Upsilon_{\tau,u_0,n} := \Pi_{\tau,u_0,n} + \frac{1}{n}\arctan(nG)\,,\quad n \in \bb N
\ee
For each $n$, being a sum of a proper observable and a smooth function of a proper observable, $\Upsilon_{\tau,u_0,n}$ is a proper observable. 
Since $\arctan$ is bounded, $\Upsilon_{\tau,u_0,n}$ converges pointwise to $\Pi_{\tau,u_0}$ as $n \to \infty$, and consequently
\be
\Upsilon_{\tau,+} := \lim_{u_0\to\infty} \lim_{n\to\infty} \Upsilon_{\tau,u_0,n} =  \re\int_S\!\epsilon_S\, \bar\eth^2\tau \, \sigma_+ 
\ee
However,
\be
\{\Upsilon_{\tau,u_0,n}, \mem_\xi\}_D = \{\Pi_{\tau,u_0,n}, \mem_\xi\}_D + \frac{1}{1+n^2G^2}\{G, \mem_\xi\}_D
\ee
and therefore
\be
\lim_{n\to\infty} \big(\{\Upsilon_{\tau,u_0,n}, \mem_\xi\}_D - \{\Pi_{\tau,u_0,n}, \mem_\xi\}_D\big) = \left\{\begin{array}{cl}
\{G, \mem_\xi\}_D &\text{\it if $G=0$}\\
0 &\text{\it if $G\ne0$}
\end{array}\right.
\ee
revealing a disagreement between the brackets defined in the limit.
This is further evidence that there is no consistent and unambiguous manner to define proper brackets for the Goldstone mode, in accordance with the fact that there are no proper observables that are equal on-shell to functionals of it alone (Theorem~\ref{nogravitonobs}).

We note that
it is essential to consider non-linear observables in order to construct sequences that converge pointwise to $\Pi_{\tau,+}$ but define different brackets in the limit. 
Intuitively, two sequences of observables with this property must differ by tiny ripples whose amplitude vanishes in the limit, but oscillate rapidly in the limit so that (functional) derivatives disagree --- linear observables (i.e., shear-news functions) are too rigid for this to occur.\footnote{Suppose two sequences of shear-news functions converge pointwise to the same function. They must differ by a sequence $\Gamma_{\lambda_n,\alpha_n}$ that converges pointwise to zero as $n\to\infty$. For each $n$ we have 
\[
\{F,\Gamma_{\lambda_n,\alpha_n}\}_D = -4\pi\re\!\int\!\epsilon_\scri\left(\bar\lambda_n f+ \bar\alpha_n g\right),\,\, \text{where } f = \frac{\delta F}{\delta\bar\Sigma} + \frac{\delta F}{\delta N} + \frac{\delta F}{\delta N_-} \,\,\text{and }\, g = - \frac{\delta F}{\delta\bar\sigma} - \pu\frac{\delta F}{\delta N}
\]
where $F$ is a proper observable.
Notice that $f\in \ca F_0$, $g \in \ca F_1$ and $\pu f + g = 0$, so $(f,g) \in \ca P_K$. Thus, since $\Gamma_{\lambda_n,\alpha_n} = \re\int\!\epsilon_\scri\left(\bar\lambda_n \sigma+ \bar\alpha_n \bar N\right) \to 0$ as $n\to\infty$, pointwise for all $(\sigma,\bar N) \in \ca P_K$, $\{F,\Gamma_{\lambda_n,\alpha_n}\}_D \to 0$ as $n\to\infty$, pointwise on phase space, for each $F$.} 
In fact, if a sequence of {\it linear} proper observables $G_n$ converges pointwise to a smooth function $G$, the limit of the brackets with any proper observable $F$ will converge to the directional brackets
\be
\lim_{n\to\infty} \{ F, G_n\}_D = - \lim_{n\to\infty} \{ G_n, F\}_D = \{F; G\}_D
\ee
Therefore, any sequence of linear functions $\Pi_{\tau,+,n}$ converging pointwise to $\Pi_{\tau,+}$ will define the same brackets with the memory in the limit
\be
\lim_{n\to\infty}\{\Pi_{\tau,+,n}, \mem_\xi\}_D = -\{\mem_\xi; \Pi_{\tau,+}\}_D = 4\pi \re \int_S \epsilon_S\, \eth^2\xi \,\bar\eth^2\tau
\ee
Interestingly, upon adapting the notation, this matches exactly that posited in \cite{He:2014laa}. 
We conclude that the Goldstone-memory bracket posited in \cite{He:2014laa} can be understood as a directional bracket derived from the Ashtekar-Streubel symplectic structure.
But we again emphasize that, despite this stability of the limit when restricted to sequences of shear-news functions, the directional bracket is not a proper bracket, and more importantly it is not possible to include these Goldstone modes into the algebra of observables. For instance, it is not possible to define brackets {\it among} Goldstone modes via such a limiting procedure, like
\be\label{Pitautau'}
\{\Pi_{\tau,+}, \Pi_{\tau',+}\}_D \stackrel{?}{:=} \lim_{n\to\infty}\lim_{n'\to\infty}\lim_{u_0\to\infty}\lim_{u_0'\to\infty} \{\Pi_{\tau,u_0,n}, \Pi_{\tau',u_0',n'}\}_D
\ee
because the order in which one takes the limits will affect the result. For example, taking the $n, n'$ limits first in \eqref{Pitautau'}, one gets
\be
\{\Pi_{\tau,+}, \Pi_{\tau',+}\}_D \stackrel{?}{:=} 4\pi\lim_{u_0\to\infty}\lim_{u_0'\to\infty} \re\int_S\!\epsilon_S\,\eth^2\tau \bar\eth^2\tau'\, \text{sign}(u_0-u_0')
\ee
which depends on whether the $u_0$ or $u_0'$ limit is taken first, or if they are taken simultaneously.
Possibly, even the choice of sequences of shear-news functions will lead to different results.

\subsection{Complete subalgebras}
\label{SecCompleteSubA}

A subalgebra of proper observables, $\scr U \subset \scr A$, is said to be {\it complete} if all proper observables can be expressed, on-shell, as functions of elements of $\scr U$. A more manageable, but equivalent, definition is
\begin{definition}
A complete subalgebra of proper observables $\scr U$ is a subalgebra of $\scr A$ which separates points of the 
constraint surface.
\end{definition}
\noindent
That is, the specification of the values of all functions in $\scr U$ uniquely determines a state (i.e., point) in $\ca P_K$.\footnote{It is assumed that $\scr A$ separates points in $\ca P_K$, which is true in the present case. If that were not true, there would be states indistinguishable by all physical observables, and it would be natural to quotient the phase space by this ``gauge'' so that $\scr A$ becomes point separating.}
We are not sure, however, whether in infinite dimensions this definition automatically implies that all proper observables in $\scr A$ can be expressed as {\it smooth} functions of elements of $\scr U$, and we will not attempt to establish that for our examples below (moreover, we do not see a reason for demanding this stronger condition).

Consider the subset of shear-news functions with $\lambda = \pu\alpha$,
\be
\ca U := \big\{\Gamma_\alpha := \Gamma_{\pu\alpha,\alpha}\,;\,\, \text{\it with } \alpha \in \ca F_0 \big\}
\ee
As mentioned in \eqref{overcompletelabels}, this subset is actually equivalent, on-shell, to the set of all shear-news functions.
We will show that $\ca U$ is point-separating in $\ca P_K$.
Since $\ca P_K$ is a linear space and $\Gamma_\alpha$ is a linear function on $\ca P_K$, it suffices to show that $\ca U$ separates $\sigma = 0$, i.e., that there are no points $(\sigma,N) \in \ca P_K$, other than $(0,0)$, for which all $\Gamma_\alpha \in \ca U$ vanish. We have
\ba\label{Ga}
\Gamma_\alpha &= \re\int\!\epsilon_\scri \left(\pu\bar\alpha\, \sigma + \bar\alpha \bar N\right) \no
&\doteq \re\int\!\epsilon_\scri \left(\pu\bar\alpha\, \sigma - \bar\alpha \pu\sigma\right) \no
&= -2\,\re\!\int\!\epsilon_\scri \,\bar\alpha \pu\sigma  + \re\!\int_S\!\epsilon_S\left(\bar\alpha_+\sigma_+ - \bar\alpha_-\sigma_-\right).
\ea
If this vanishes for all $\alpha \in \ca F_0$, in particular it vanishes for all compactly supported $\alpha$, which implies $\pu\sigma = 0$. Then, for $\sigma(x) = \sigma_0(\ang)$, we have
\be
\Gamma_\alpha\Big|_{\sigma = \sigma_0} = \re\!\int_S\!\epsilon_S\left(\bar\alpha_+ - \bar\alpha_-\right)\sigma_0
\ee
Now considering $\alpha$'s with $\alpha_+ \ne \alpha_-$ implies $\sigma_0 = 0$. Hence, there is a unique point ($\sigma = 0$) for which all $\Gamma_\alpha$ vanish, concluding the proof.

The Dirac brackets between elements of $\ca U$ can be directly computed from \eqref{DBshearnewsDir0}, which gives
\be\label{scrUprod}
\{\Gamma_\alpha, \Gamma_{\alpha'}\}_D = 4\pi \,\re\!\int\!\epsilon_\scri\left(\alpha'\pu\bar\alpha - \alpha \pu\bar\alpha' \right)
\ee
Therefore, together with the constant function $1$ (i.e., $p\mapsto 1$), we can define the ``Heisenberg'' subalgebra of proper observables
\be
\scr U := \ca U \oplus_S 1 = \big\{\Gamma_\alpha + a 1\,;\,\, \text{\it with  $\alpha \in \ca F_0$ and $a\in\bb R$} \big\}
\ee
where $\oplus_S$ indicates the semi-direct sum with the central element $1$ (we will omit the $1$ in what follows, denoting $a 1$ simply by $a$). The product on $\scr U$ is given, from \eqref{scrUprod}, by
\be\label{scrUprod1}
(\Gamma_\alpha + a ) * (\Gamma_{\alpha'} +a') := \{\Gamma_\alpha + a, \Gamma_{\alpha'} + a'\}_D =  4\pi \,\re\!\int\!\epsilon_\scri\left(\alpha'\pu\bar\alpha - \alpha \pu\bar\alpha' \right)
\ee
where $a,a'\in \bb R$. In this manner, $\scr U$ is a complete subalgebra of $\scr A$.
This algebra looks at first sight the same as the one considered in
\cite{Ashtekar84book,Ashtekar:1981bq}, but note that the space
of smearings here is larger. In particular, $\alpha$ does not vanish as $u\to \pm\infty$, which is crucial to make the memory not central. Also, the elements of $\scr U$ are constructed from a particular mixture of shear and news, as opposed to those previously considered which involve only news \cite{Ashtekar84book,Ashtekar:1981bq} or only shear \cite{He:2014laa,Ashtekar:2018lor,Strominger:2013jfa} (and consequently would not admit smearings with $\alpha_+ \ne \alpha_-$).

Notice that $\scr U$ is not {\it minimal}, in the sense that there are non-dense subalgebras of $\scr U$ which are also complete.\footnote{When $\scr U$ is indexed by a family of smearings, such as in this case where it is indexed by $\ca F_0 \times \bb R$, the most natural topology on $\scr U$ is the quotient topology induced by the map $\ca F_0\times \bb R \to \scr U$, $(\alpha, a) \mapsto \Gamma_\alpha + a$. The natural topology on $\ca F_0$ is discussed near \eqref{seminorms0}, and the definition of quotient topology is given in footnote~\ref{quotienttopology}. However, for the purposes of this section, one can think of ``denseness'' in a loose sense.} 
In fact, inspecting the proof above that $\ca U$ is point-separating, we only really needed that $\alpha$ is arbitrary for finite $u$ (or, at least, possibly excluding a subset of isolated points), and that $\alpha_+ \ne \alpha_-$. Therefore, the subalgebras $\scr U_+$ and $\scr U_-$, defined by
\be
\scr U_\pm := \big\{\Gamma_\alpha + a 1\,;\,\, \text{\it with $\alpha \in \ca F_0$, $\alpha_\mp = 0$ and $a\in\bb R$} \big\} 
\ee
are non-dense in $\scr U$ but also complete. 
The product on $\scr U_\pm$ is again given by the Dirac brackets, as in \eqref{scrUprod1}.
Technically, these subalgebras are still not minimal: for example, one could further restrict to smearings $\alpha$ that vanish on a finite subset of points in $\scri$. Nevertheless, it appears that there are no physically natural subalgebras smaller than $\scr U_\pm$. In that sense, we will refer to $\scr U_\pm$ (and even $\scr U$) as ``quasi-minimal'' complete subalgebras.
Generally, a suitable choice of a minimal or quasi-minimal complete subalgebra is expected to play a central role in the quantization.\footnote{Even in finite-dimensional phase spaces, only a ``sufficiently small'' complete subalgebra (like the Heisenberg algebra of positions and momenta) is assumed to undergo canonical quantization, in view of Groenewold-Van Hove type obstructions.} 

While $\scr U_\pm$ do not include memory observables, all the necessary information about the memory is encoded in these subalgebras.
In fact, it is possible to ``generate'' the memory observables via a limiting procedure. Define the sequence of observables
\be
\mem_{\xi,n} := \frac{1}{2} \Gamma_{\chi_n\eth^2\xi} \doteq \re\int\!\epsilon_\scri\, \chi_n \eth^2\xi N 
\ee
where $\chi_n(u)$ are smooth, compactly-supported real functions equal to $1$ in the interval $[-n,n]$ and with absolute value bounded by $1$. 
For each $n\in\bb N$, $\mem_{\xi,n}$ belongs to both $\scr U_\pm$.
Since $N$ approaches zero faster than $|u|^{-1-\varepsilon}$ as $u \to \pm\infty$, we see that
\be
\lim_{n\to\infty} \mem_{\xi,n} \doteq \mem_{\xi}
\ee
Moreover, this limit is ``uniform'' in the observable algebra, in the sense that it commutes with the Dirac brackets.
In fact, for any proper observable $F$ we have
\be
\{F,\mem_{\xi,n}\}_D = 2\pi\,\re\!\int\!\epsilon_\scri\, \chi_n \eth^2\xi \left( \frac{\delta F}{\delta\sigma} + \pu\frac{\delta F}{\delta\bar N}\right)
\ee
and since the term in parenthesis belongs to $\ca F_1$, the only contribution to the integral, as $n\to\infty$, will come from the interval $u\in[-n,n]$, where $\chi_n(u)=1$. Therefore,
\be
\lim_{n\to\infty}\{F,\mem_{\xi,n}\}_D  = \{F,\mem_{\xi}\}_D
\ee
showing that memory observables can be approximated by sequences in $\scr U_\pm$ that converge ``uniformly'' in $\scr A$.

It is interesting to note that, in all these subalgebras, the supertranslations act automorphically on-shell (i.e., they preserve the subalgebras). To see this, note from \eqref{Tdirbrackets} that
\ba
\{F_T, \Gamma_{\alpha}\}_D &=  2\re\int\!\epsilon_\scri\,T N \,\pu\alpha + \re\int_S\!\epsilon_S\left(\alpha_+ - \alpha_-\right) \bar\eth^2T \no
&= 2\Gamma_{0, T\pu\alpha} + \re\int_S\!\epsilon_S\left(\alpha_+ - \alpha_-\right) \bar\eth^2T \no
&\doteq \Gamma_{T\pu\alpha} + \re\int_S\!\epsilon_S\left(\alpha_+ - \alpha_-\right) \bar\eth^2T
\ea
where in the last line we used the (on-shell) equivalence described in \eqref{overcompletelabels}, $\Gamma_{0, T\pu\alpha} \doteq \Gamma_{\pu\beta, \beta}$ with $\beta = \frac{1}{2}\left(0+T\pu\alpha + (T\pu\alpha)_-\right) = \frac{1}{2}T\pu\alpha$.
In particular, the {\it Bondi flux}, $F_1$, acts as
\be
\{F_1, \Gamma_{\alpha}\}_D \doteq \Gamma_{\pu\alpha}
\ee
and it commutes with the memory observables.

\section{Complex and distributional brackets}
\label{SecDDIB}

It is possible to extend the notion of brackets to complex and localized functions, such as $N(x)$ and $\sigma(x)$. These should not be understood as standard brackets between observables, since such functions are not tangentially extensible (they would, in principle, generate non-real and non-smooth flows, or even fail to generate any flow). Nevertheless, such brackets can still be given precise and useful definitions, serving as tools for computing well-defined brackets between proper observables through appropriate smearings and linear combinations.
This will provide the accurate replacement for \eqref{NNbra}, in the presence of memory.

\subsection{Complex brackets}
\label{SecComplexB}

We define the complex shear-news function by
\ba
\Gamma_{\lambda, \alpha}^{\bb C} &:= \int\!\epsilon_\scri\left(\bar\lambda\sigma + \bar\alpha \bar N\right) \no
&= \re\int\!\epsilon_\scri\left(\bar\lambda\sigma + \bar\alpha \bar N\right) + i\, \im\int\!\epsilon_\scri\left(\bar\lambda\sigma + \bar\alpha \bar N\right) \no
&= \re\int\!\epsilon_\scri\left(\bar\lambda\sigma + \alpha  N\right) + i\, \re\int\!\epsilon_\scri\left(\overline{i\lambda}\sigma + i\alpha N\right) \no
&= \Gamma_{\lambda, \alpha} + i\, \Gamma_{i\lambda, i\alpha}
\ea
Then, we are {\it formally} allowed to compute Dirac brackets involving $\Gamma_{\lambda, \alpha}^{\bb C}$ as follows
\be
\{\Gamma_{\lambda, \alpha}^{\bb C}, \Gamma_{\lambda', \alpha'}^{\bb C}\}_D := 
\{\Gamma_{\lambda, \alpha}, \Gamma_{\lambda', \alpha'}\}_D 
- \{\Gamma_{i\lambda, i\alpha}, \Gamma_{i\lambda', i\alpha'}\}_D
+ i \{\Gamma_{\lambda, \alpha}, \Gamma_{i\lambda', i\alpha'}\}_D
+ i \{\Gamma_{i\lambda, i\alpha}, \Gamma_{\lambda', \alpha'}\}_D
\ee
provided that $(\lambda,\alpha)$ and $(\lambda',\alpha')$ are in $\ca G_{TE}$.\footnote{Note that if $(\lambda, \alpha) \in \ca G_{TE}$, then $(i\lambda,i\alpha)$ would belong to a complexified version of $\ca G_{TE}$. That is, in particular, $\int\!du\, i\lambda(u) = (i\alpha)_+ - (i\alpha)_-$. This is why it is important to define $\Gamma_{\lambda,\alpha}^{\bb C}$ with the conjugated news (i.e., $\bar\alpha \bar N$); if instead we had used $\alpha N$, we would have $\Gamma_{\lambda,\alpha} = \Gamma_{\lambda, \alpha} + i\, \Gamma_{i\lambda, -i\alpha}$, and $(i\alpha, -i\lambda)$ would not generally satisfy the integral condition. Note, however, that $\ca G_{TE}$ is fundamentally a real vector space of complex-valued spin-weighted functions, and its complexification would not be compatible with the condition that $\alpha$ must have purely electric asymptotic limits. The implications of this ``abuse'' will be discussed at the end of this subsection.}
(For directional brackets, the second entry can be any smooth $\Gamma_{\lambda', \alpha'}^{\bb C}$.)
We find,
\be\label{GammaCGammaC}
\{\Gamma_{\lambda, \alpha}^{\bb C}, \Gamma_{\lambda', \alpha'}^{\bb C}\}_D = 0
\ee
Similarly we can compute brackets involving also $\bar\Gamma_{\lambda, \alpha}^{\bb C}$. The mixed brackets give
\be\label{GammaCbarGammaC}
\{\Gamma_{\lambda, \alpha}^{\bb C}, \bar\Gamma_{\lambda', \alpha'}^{\bb C}\}_D = -4\pi\int\!\epsilon_\scri\left(\lambda' (\bar\Lambda + \bar\alpha + \bar\alpha_-) - \alpha'(\bar\lambda + \pu \bar\alpha) \right)
\ee
The brackets commute with complex-conjugation, so
\be
\{\bar\Gamma_{\lambda, \alpha}^{\bb C}, \Gamma_{\lambda', \alpha'}^{\bb C}\}_D = \overline{\{\Gamma_{\lambda, \alpha}^{\bb C}, \bar\Gamma_{\lambda', \alpha'}^{\bb C}\}_D}
\ee
and $\{\bar\Gamma_{\lambda, \alpha}^{\bb C}, \bar\Gamma_{\lambda', \alpha'}^{\bb C}\}_D = 0$.

Formula \eqref{GammaCbarGammaC} holds both for regular or directional brackets, but for regular brackets it is possible to re-express it in a manifestly anti-symmetric form,
\ba
\{\Gamma_{\lambda, \alpha}^{\bb C}, \bar\Gamma_{\lambda', \alpha'}^{\bb C}\}_D =&\,\, 2\pi\int\!\epsilon_\scri\left(\bar\lambda (\Lambda' + 2\alpha') - \lambda'(\bar\Lambda + 2\bar\alpha) -\bar\alpha\partial_u\alpha' + \alpha'\partial_u\bar\alpha \right) \no
&+ 2\pi\int_S\!\epsilon_S \left(\bar\alpha_+\alpha'_- - \alpha'_+\bar\alpha_-\right)
\label{GammaCbarGammaCas}
\ea
with $(\lambda,\alpha)$ and $(\lambda',\alpha')$ in $\ca G_{TE}$.

Let us emphasize that the complex brackets, defined in this manner, are improper and should be understood only as a tool to produce real brackets (via appropriate linear combinations). The reason for this is that the complex variables do not satisfy the boundary conditions on the smearing, therefore they are not proper observables. 
It is instructive to see some blatant contradictions that can arise if one ignores this fact and tries to treat them as observables.
For example, let us consider the Dirac brackets between a generic complexified shear-news function and the {\it memory} observable, defined in \eqref{memoryonshell}. The memory $\mem_\xi$, with real $\xi$, can be expressed {\it on-shell} both as $\Gamma_{0,\eth^2\xi}^{\bb C}$ or as $\bar\Gamma_{0,\eth^2\xi}^{\bb C}$. (As shown in footnote~\ref{footnotememreal}, this equality makes use of the boundary conditions for the shear, which is asymptotically purely electric.) But we see that $\{\Gamma_{\lambda, \alpha}^{\bb C}, \Gamma_{0,\eth^2\xi}^{\bb C}\}_D \ne \{\Gamma_{\lambda, \alpha}^{\bb C}, \bar\Gamma_{0,\eth^2\xi}^{\bb C}\}_D$.
This is a violation of the fundamental property of Dirac brackets, where two functions that are equal on-shell should have the same brackets (on-shell) with all other observables.
On the other hand, for a real shear-news function $\Gamma_{\lambda, \alpha}$, we have $\{\Gamma_{\lambda, \alpha}, \Gamma_{0,\eth^2\xi}^{\bb C}\}_D = \{\Gamma_{\lambda, \alpha}, \bar\Gamma_{0,\eth^2\xi}^{\bb C}\}_D$, as expected.
Since real observables are the only ones with physical relevance, it is still legitimate to use complexified brackets as a computational tool, despite their intrinsic inconsistencies, as long as only linear combinations producing real functions are considered. 
A similar contradiction can be derived for the conjectured Goldstone bracket of \cite{He:2014laa}, if one tries to derive it integrating the shear-news bracket.

The root of the improperness of the complexified ``observables'' can be traced to the boundary conditions. The asymptotic limits of the shear must be purely electric, which translates into $\bar\eth^2\sigma_\pm = \eth^2\bar\sigma_\pm$. A phase space vector field $X$ respecting this boundary condition must satisfy $\bar\eth^2(X\sigma_\pm) = \eth^2(X\bar\sigma_\pm)$. For real flows we have $X\bar\sigma = \ol{X\sigma}$, so the boundary condition is respected provided that $X\sigma_\pm$ are purely electric. For the shear-news functions this required that $\alpha_\pm$ are purely electric. On the other hand, in the complexified phase space, $X\sigma$ and $X\bar\sigma$ need not be related. In fact, the symplectic flow of $\Gamma_{\lambda, \alpha}^{\bb C}$ gives $X\sigma = 0$ and $X\bar\sigma = 4\pi\bar\alpha$; the boundary condition is only respected if $\eth^2\bar\alpha_\pm = 0$, which implies $\alpha_\pm = 0$. Thus, generic complexified shear-news functions, as defined in this section, tend not to respect the boundary conditions. Notice that if we had considered boundary conditions where the shear has arbitrary asymptotic limits, including magnetic components, this issue would not occur.

\subsection{Distributional brackets}
\label{SecDistbrackets}

We wish to define brackets between {\it localized} shears and news, such as 
\be\label{NsDB}
\{\sigma(x), N(x')\}_D \,,\quad \{N(x), \bar N(x')\}_D \,,\quad \{\sigma(x), \sigma(x')\}_D \,,\quad \text{\it etc...}
\ee
These are to be understood as distributional objects, and can be used to compute proper brackets only when suitably smeared over. 
Although it would be entirely sufficient to work only with proper observables, it is useful to compute these distributional brackets for the purpose of comparison with the standard literature, where these often appear.

\subsubsection{Definition}
\label{SecDefDistB}

One might have thought that these brackets could possibly be defined by simply taking shear-news functions with Dirac delta smearing parameters. However, there are significant issues with such an approach.
First, the only reasonable manner to implement a smearing like $(\delta_x,0)$ would be to realize it as a sequence of smooth functions inside $\ca G_{TE}$, for example $(\lambda_n, 0)$ where $\lambda_n$ is a sequence of Gaussians that become sharper and sharper. But note that such a sequence cannot be constructed within $\ca G_{TE}$ since the integral condition \eqref{GTEintcond} would be violated by the limit of the sequence, 
\be
\int_{-\infty}^{+\infty}\!du'\, \delta_x(x') = \delta(\ang',\ang) \ne 0 
\ee
Consequently, $\sigma(x) = \Gamma_{\delta_x,0}^{\bb C}$ is not tangentially-extensible even in this limiting sense.
Second, one could think that at least the news-news brackets could be defined as
\be
\{N(x),\bar N(x')\}_D \stackrel{?}{:=} \{ \bar\Gamma_{0, \delta_x}^{\bb C}, \Gamma_{0, \delta_{x'}}^{\bb C}\}_D
\ee
However, as we will see in this section, that would fail to capture the non-local aspects of the Dirac brackets on $\scri$, since the Dirac delta is compactly supported.

The appropriate manner to proceed is to notice that these localized brackets are fundamentally distributional objects, and must always be treated ``under the integral sign'', with suitable smearing functions. 
Consider a complex smooth function $F$. In order for it to be regular (in the complex sense), given the symplectic structure \eqref{OmASNP}, it must have well-defined, smooth functional derivatives (without boundary terms),
\be
\delta F = \int\!\epsilon_\scri \left( \frac{\delta F}{\delta\sigma(x)}\delta\sigma(x) + \frac{\delta F}{\delta\bar N(x)}\delta\bar N(x) + \frac{\delta F}{\delta\bar\sigma(x)}\delta\bar\sigma(x) + \frac{\delta F}{\delta N(x)}\delta N(x) \right)
\ee
Moreover, similarly to \eqref{FTEcond}, $F$ is tangentially extensible (in the complex sense) if and only if
\ba
p \in \ca P_K &\mapsto \left(\frac{\delta F}{\delta\bar\sigma},\, \frac{\delta F}{\delta N}\right)\!\bigg|_p \in \ca G_{TE} \,\,\, \text{\it is smooth} \label{FGconstraints} \\
p \in \ca P_K &\mapsto \left(\frac{\delta F}{\delta\sigma},\, \frac{\delta F}{\delta \bar N}\right)\!\bigg|_p \in \bar{\ca G}_{TE} \,\,\, \text{\it is smooth}
\ea
We {\it define} distributions $\{\sigma(x), \sigma(x')\}_D$, etc, called {\it proper distributional brackets}, such that 
for any two tangentially extensible functions, $F$ and $G$, we have
\be\label{DistBracketsDef}
\{F,G\}_D = \int\!\epsilon_\scri\! \int\!\epsilon'_\scri \left( \frac{\delta F}{\delta\sigma(x)} \{\sigma(x), \sigma(x') \}_D \frac{\delta G}{\delta\sigma(x')} + \frac{\delta F}{\delta\sigma(x)} \{\sigma(x), \bar N(x') \}_D \frac{\delta G}{\delta\bar N(x')} + \cdots \right)
\ee
where the ``$\cdots$'' include all other analogous combinations of a functional derivative of $F$ and a functional derivative of $G$.
Similarly, we can define {\it directional distributional brackets} such that, for any tangentially extensible $F$ and smooth $G$, we have
\be\label{DirDistBracketsDef}
\{F;G\}_D = \int\!\epsilon_\scri\! \int\!\epsilon'_\scri \left( \frac{\delta F}{\delta\sigma(x)} \{\sigma(x); \sigma(x') \}_D \frac{\delta G}{\delta\sigma(x')} + \frac{\delta F}{\delta\sigma(x)} \{\sigma(x); \bar N(x') \}_D \frac{\delta G}{\delta\bar N(x')} + \cdots \right)
\ee
where the ``$\cdots$'' have the same meaning as before.
In both cases, we don't expect the procedure to give necessarily unique answers: in fact, the constraints \eqref{FGconstraints} may induce zero modes in this inversion.
Given that we defined the distributional brackets only  as a way to reconstruct brackets, these ambiguities are immaterial and any solution is equally valid.

Note that in the right-hand side of the expressions above, the functional derivatives only appear outside the brackets. Therefore, the phase space dependence of the functional derivatives is not relevant for determining the distributional brackets, and we can simply focus on brackets between (complex) shear-news functions
\ba
\{\Gamma_{\lambda, \alpha}^{\bb C},\Gamma_{\lambda', \alpha'}^{\bb C}\}_D &= \int\!\epsilon_\scri\! \int\!\epsilon'_\scri \left( \bar\lambda(x) \bar\lambda'(x') \{\sigma(x), \sigma(x') \}_D + \bar\lambda(x) \bar\alpha'(x') \{\sigma(x), \bar N(x') \}_D + \cdots \right) \no
\{\Gamma_{\lambda, \alpha}^{\bb C},\bar\Gamma_{\lambda', \alpha'}^{\bb C}\}_D &= \int\!\epsilon_\scri\! \int\!\epsilon'_\scri \left( \bar\lambda(x) \lambda'(x') \{\sigma(x), \bar\sigma(x') \}_D + \bar\lambda(x) \alpha'(x') \{\sigma(x), N(x') \}_D  + \cdots \right)
\ea
for all $(\lambda, \alpha),\,(\lambda', \alpha') \in \ca G_{TE}$. The same applies for directional brackets.

A complication in solving the equations above, for the distributional brackets, is that the space of smearings is a constrained product of two spaces, as in \eqref{GTEdef},
\be
\ca G_{TE} \simeq \ca F_1 \times_{\wt K} \ca F_0
\ee
It is thus convenient to decompose $\ca G_{TE}$ into two independent factors. Notice that every $(\lambda,\alpha) \in \ca G_{TE}$ can be written as
\be
(\lambda, \alpha) = (\rho, 0) + (\pu \alpha, \alpha)
\ee
with $\int\!du\,\rho(x) = 0$. Thus
\be
\ca G_{TE} \simeq \ca F_1^0 \times \ca F_0
\ee
where 
\be\label{caF10def}
\ca F_1^0 := \Big\{\rho \in \ca F_1;\,\, \text{\it such that } \int\!du\,\rho(x) = 0 \Big\} 
\ee
In terms of this parametrization of the shear-news functions, we have
\ba
\{\Gamma_{\rho+\pu\alpha,\alpha}^{\bb C}, \Gamma_{\rho'+\pu\alpha',\alpha'}^{\bb C}\}_D &= \Big\{\int\!\epsilon_\scri \left( (\bar\rho(x) +\pu\bar\alpha(x))\sigma(x) +\bar\alpha(x)\bar N(x)\right), \no
&\qquad \int\!\epsilon'_\scri \left( (\bar\rho'(x') +\pup\bar\alpha'(x'))\sigma(x') +\bar\alpha'(x')\bar N(x')\right) \Big\}_D \no
&\!\!\!\!\!\!\!\!\!\!\!\!\!\!\!\!\!\!\!\!\!\!\!\!\!\!\!\!\!\!\!\!\!\!\!\!\!\!\!\!\!\!\!\!\!
= \int\!\epsilon_\scri\! \int\!\epsilon'_\scri \left(
\bar\rho(x)\bar\rho'(x') \{ \sigma(x), \sigma(x')\}_D + \pu\bar\alpha(x)\bar\alpha'(x')\{\sigma(x),\bar N(x')\}_D + \cdots \right)
\label{GammaCGammaCdist0}
\ea
and
\ba
\{\Gamma_{\rho+\pu\alpha,\alpha}^{\bb C}, \bar\Gamma_{\rho'+\pu\alpha',\alpha'}^{\bb C}\}_D &= \Big\{\int\!\epsilon_\scri \left( (\bar\rho(x) +\pu\bar\alpha(x))\sigma(x) +\bar\alpha(x)\bar N(x)\right), \no
&\qquad \int\!\epsilon'_\scri \left( (\rho'(x') +\pup\alpha'(x'))\bar\sigma(x') +\alpha'(x') N(x')\right) \Big\}_D \no
&\!\!\!\!\!\!\!\!\!\!\!\!\!\!\!\!\!\!\!\!\!\!\!\!\!\!\!\!\!\!\!\!\!\!\!\!\!\!\!\!\!\!\!\!\!
= \int\!\epsilon_\scri\! \int\!\epsilon'_\scri \left(
\bar\rho(x)\rho'(x') \{ \sigma(x), \bar\sigma(x')\}_D +\pu\bar\alpha(x)\alpha'(x')\{\sigma(x), N(x')\}_D + \cdots \right)
\label{GammaCbarGammaCdist0}
\ea
This allows us to reduce the problem to a set of independent equations (one involving $\rho$ and $\rho'$, another involving  $\alpha$ and $\alpha'$, etc).

\subsubsection{Functional analysis}
\label{Secfuncanalysis}

In order to properly define these distributions, we must establish a solid functional theory for the spaces of smearings under consideration. Since it is essential that the smearing parameters $\alpha$ do not vanish asymptotically, the distributions will behave in some aspects differently than the distributions commonly used in physics (which are almost always based on compactly-supported or fast-decaying test functions).

The relevant spaces of smearings are
\be\label{restrictedsmearings}
\alpha \in \ca F_0 \,\,,\quad
\lambda \in \ca F_1 \,\,,\quad
\rho \in \ca F_1^0
\ee
which were defined respectively in \eqref{caF0def}, \eqref{caF1def} and \eqref{caF10def}.
These spaces have natural Fr\'echet topologies, which we describe now. 
The space $\ca F_0$ can be given the topology associated with the sequence of seminorms 
\be\label{seminorms0}
|\!|f|\!|_{0,n} :=  \sum_{\scr D}^{|\scr D| \le n} \sup_{S^2} |\scr D f_-| + \sum_{k \ge 1, \scr D}^{k + |\scr D| \le n} \sup_\scri \left|(1+u^2)^{\frac{k+\varepsilon}{2}}\pu^k\scr D f \right|
\ee
where $f\in \ca F_0$, $n\ge 1$ and $|\scr D|$ is the number of derivatives in $\scr D \in \text{\rm words}(\eth, \bar\eth)$.\footnote{This is equivalent to the construction of \cite{Ashtekar:1981bq}, up to the minor modification mentioned earlier that we need to exclude $k=0$ in order to allow for memory.}
The space $\ca F_1$ can be given the quotient topology\footnote{Given a topological space $\ca F$ and a quotient map $q: \ca F \to \ca G$, the {\it quotient topology} on $\ca G$ is the finest topology which makes $q$ continuous. 
If $\ca F$ is a topological vector space whose topology is defined by a family of seminorms $|\!|f|\!|_{\ca F,n}$, then the quotient topology on $\ca G$ is defined from the family of seminorms $|\!|g|\!|_{\ca G,n} := \text{inf}_{q(f)=g}|\!|f|\!|_{\ca F,n}$~\cite{bourbaki2013topological}.
\label{quotienttopology}}
induced by the $u$-derivative map,
\be
\pu : \ca F_0 \to \ca F_1
\ee
This topology can be described in terms of the sequence of seminorms (see footnote~\ref{quotienttopology})
\be\label{seminorms1}
|\!|g|\!|_{1,n} = \inf_{\pu f = g} |\!|f|\!|_{0,n} = \sum_{k \ge 0, \scr D}^{k + |\scr D| < n} \sup_\scri \left|(1+u^2)^{\frac{k+1+\varepsilon}{2}}\pu^k\scr D g \right|
\ee
where $g\in \ca F_1$ and $n\ge 1$.
Finally, the space $\ca F_1^0$ is given the topology it inherits as a subspace of $\ca F_1$. The seminorms characterizing its topology are thus the same as in \eqref{seminorms1}.
With the topologies above, the spaces $\ca F_0$, $\ca F_1$ and $\ca F_1^0$ are Fr\'echet.

These spaces satisfy the inclusion relation
\be\label{Aninclusion}
\ca F_0 \supset \ca F_1 \supset \ca F_1^0
\ee
and the inclusion maps are continuous.  However, note that the topology of $\ca F_1$ given above is not equal to the one it would inherit as a subspace of $\ca F_0$. 
Finally, we recall that these spaces are to be regarded as real vector spaces (of complex spin-weight $2$ functions). Accordingly, the conjugated smearing parameters (of spin-weight $-2$)  live in the respective complex-conjugated spaces
\be
\bar\alpha \in \bar{\ca F}_0 \,\,,\quad
\bar\lambda \in \bar{\ca F}_1 \,\,,\quad
\bar\rho \in \bar{\ca F}_1^0
\ee
There is a natural isomorphism between each space and its complex-conjugate.

Distributions are continuous linear functionals on a space of ``test functions''. 
Thus, distributions on a space of functions $\ca F$ are elements of $\ca F^*$, the topological dual of $\ca F$. 
Due to the inclusion property \eqref{Aninclusion}, the respective dual spaces satisfy an inverse inclusion relation\footnote{This inclusion property follows from the fact that any distribution in $\ca F_0^*$, for example, can be {\it restricted} to act on elements of $\ca F_1 \subset \ca F_0$, thereby defining a distribution in $\ca F_1^*$. As the inclusion map $i: \ca F_1 \to \ca F_0$ is continuous in the given topologies, the restriction map on distributions is the pull-back $i^*: \ca F_0^* \to \ca F_1^*$. One should not think that $\ca F_1^*$ contains ``more'' distributions than $\ca F_0^*$, as $i^*$ is not injective. So, the appropriate reading is only that ``every distribution in $\ca F_0^*$ also makes sense as a distribution in $\ca F_1^*$''.}
\be\label{Andualinclusion}
\ca F_0^* \subset \ca F_1^* \subset {\ca F_1^0}^*
\ee
Similarly, we have conjugated distributions in $\bar{\ca F}^*$ acting on $\bar{\ca F}$. (Since the definitions and properties of the conjugated distributions are analogous to those of the non-conjugated ones, we will only make statements about one or the other.)
We denote the action of a distribution $\Phi \in \ca F^*$ on elements $f \in \ca F$ by
\be
\la\Phi, f\ra = \la \Phi(x), f(x)\ra = \int\!\epsilon_\scri\, \Phi(x) f(x)
\ee
where in the first equality the bracket is used only as a short-hand notation for the integral, and the integral representation should be understood in a formal sense (as usual since $\Phi(x)$ is not a function).\footnote{Note that we include the measure $\epsilon_\scri$ in the integration, in such a way that distributions will behave more like functions rather than densities. Since the measure is a background structure on $\scri$, this choice is merely for convenience. (More precisely, as the true background structure on $\scri$ is conformal, and $\epsilon_\scri$, $\alpha$ and $\lambda$ carry conformal weights $3$, $-1$ and $-2$, respectively, then distributions in $\ca F_0^*$ will actually carry conformal weight $-2$, while distributions in $\ca F_1^*$ will carry conformal weight $-1$.)
All the results in this section can be trivially converted to the language of densities by incorporating $\epsilon_\scri$ ``inside'' the distributions.}
The distributional brackets \eqref{NsDB} we are after are actually {\it bidistributions}, since they need to be smeared with two ``test functions'', one depending on $x$ and the other on $x'$. For example, $\{N(x), \bar N(x')\}_D$ is smeared with $\alpha(x) \bar\alpha(x')$, and therefore should be regarded as (continuous) linear functionals on  the (algebraic) tensor product $\ca F_0 \otimes \bar{\ca F}_0$.\footnote{The {\it algebraic} tensor product $\ca F \otimes \ca G$ contains all {\it finite} linear combinations $f(x)g(x')$, with $f \in \ca F$ and $g \in \ca G$. This is all that is needed to compute brackets between observables. (A different concept, often encountered in Quantum Mechanics, is that of {\it topological} tensor products, including also a completion leading to more general functions of the form $\phi(x,x')$, which are not needed for our purposes.) The natural topology on $\ca F \otimes \ca G$ is the {\it projective tensor-product topology}, which is defined in terms of semi-norms
\[
|\!|W|\!|_{\ca F\otimes\ca G,n} := \inf_{W = \sum_i f_ig_i}\sum_i |\!|f_i|\!|_{\ca F,n} |\!|g_i|\!|_{\ca G,n}
\]
where $W\in \ca F\otimes\ca G$ and the topologies of $\ca F$ and $\ca G$ are defined by increasing sequences of seminorms. 
The 
dual $(\ca F\otimes\ca G)^*$ of $\ca F\otimes\ca G$ is defined with respect to this topology.
It follows that there is
a natural association between continuous linear maps $\ca F \otimes \ca G \to \bb R$ (i.e., elements of the dual) and
continuous bilinear maps $\ca F\times\ca G \to \bb R$
~\cite{treves2016topological}.}
We denote the space of such bidistributions on $\ca F\otimes\ca G$ by $(\ca F \otimes \ca G)^*$. 
The action of a distribution $\Phi \in (\ca F \otimes \ca G)^* $ on elements $fg \in \ca F \otimes \ca G$ is expressed as
\be
\la\Phi, fg\ra = \la \Phi(x,x'), f(x)g(x')\ra = \int\!\epsilon_\scri\! \int\!\epsilon'_\scri\,  \Phi(x,x') f(x) g(x')
\ee
Each distributional bracket will live on a suitable dual space, depending on the most general functions it needs to be smeared against. Namely,
\ba
\{\sigma(x), \sigma(x')\}_D &\in (\bar{\ca F}_1 \otimes \bar{\ca F}_1)^* \\
\{\sigma(x), N(x')\}_D &\in (\bar{\ca F}_1 \otimes \ca F_0)^* \\
\{N(x), \sigma(x')\}_D &\in (\ca F_0 \otimes \bar{\ca F}_1)^* \\
\{N(x), N(x')\}_D &\in (\ca F_0 \otimes \ca F_0)^*
\ea
and so on.

Now we describe some important distributions on the spaces $\ca F_0$, $\ca F_1$ and $\ca F_1^0$, highlighting their unfamiliar properties.
\begin{definition}
The ``asymptotic limits'' are two distributions in $\ca F_0^*$, denoted by $L_+$ and $L_-$, and defined as
\be
\la L_\pm(u), f(x)\ra = \int\!\epsilon_\scri\, L_\pm(u) f(x) := \int_S\!\epsilon_S\,f_\pm(\ang)
\ee
where $x = (u,\ang)$. Or, more simply, we could express its action as
\be
\int\!du\, L_\pm(u) f(x) = f_\pm(\ang)
\ee
\end{definition}
\noindent
So, essentially, $L_\pm$ ``picks'' the respective asymptotic limit of $f$. 
Note that $L_\pm$ can be restricted to $\ca F_1$ and $\ca F_1^0$, but it becomes the {\it zero} distribution on those spaces.

\begin{definition}\label{defDu}
The $u$-derivative, $\fr D_u$, of distributions can be defined as a map
\be
\fr D_u : \ca F_1^* \to \ca F_0^*
\ee
such that, for $\Phi \in \ca F_1^*$, $\fr D_u \Phi \in \ca F_0^*$ is given by
\be
\la\fr D_u\Phi, f(x)\ra := - \la \Phi, \pu f(x)\ra
\ee
\end{definition}
\noindent
The reason for using a different symbol for derivatives of distributions, as opposed to simply $\partial_u$, is because of the following peculiarity when functions are identified with distributions: 
\begin{lemma}
Any $\bar\alpha \in \bar{\ca F}_0$ defines a distribution in $\ca F_1^*$ given by
\be
\la\bar\alpha, \lambda\ra := \int\!\epsilon_\scri\, \bar\alpha(x) \lambda(x)
\ee
where the right-hand side is a standard integration of two functions.
The $u$-derivative of $\bar\alpha \in \ca F_1^*$ is a distribution in $\ca F_0^*$ given by
\be
\fr D_u\bar\alpha = \pu\bar\alpha - \bar\alpha_+ L_+ + \bar\alpha_- L_-
\ee
\end{lemma}
\begin{proof}
Since $\bar\alpha$ is bounded and $\int\!du\,|\lambda(x)|$ is finite, the integral defining $\la\bar\alpha, \lambda\ra$ converges.
We have
\ba
\la \fr D_u \bar\alpha, \alpha'\ra = - \la \bar\alpha, \pu\alpha' \ra &= - \int\!\epsilon_\scri\, \bar\alpha(x) \pu\alpha'(x) \no
&= \int\!\epsilon_\scri\, \pu\bar\alpha(x) \alpha'(x) - \bar\alpha_+ \alpha'_+ + \bar\alpha_- \alpha'_- \no
&= \la\pu\bar\alpha - \bar\alpha_+ L_+ + \bar\alpha_- L_-, \alpha'\ra
\ea
for all $\alpha' \in \ca F_0$.
\end{proof}

\begin{definition}
The Dirac delta, $\delta(x,x')$, is a bidistribution which can be defined on all tensor products $\ca F_0 \otimes \ca F_1$,  $\ca F_1 \otimes \ca F_1$, etc, except for $\ca F_0 \otimes \ca F_0$. As usual,
\be
\la \delta(x,x'), f(x)g(x') \ra := \int\!\epsilon_\scri\, f(x)g(x)
\ee
\end{definition}
\noindent
On $\ca F_0 \otimes \ca F_0$ the Dirac delta is ill-defined since the integral would generally not converge.
Naturally, we can define $\delta(u,u')$ and $\delta(\ang,\ang')$ in the obvious way, and $\delta(x,x') = \delta(u,u')\delta(\ang,\ang')$. Extending Def.~\ref{defDu} to bidistributions, in the natural way, the $u$- and $u'$-derivatives of the Dirac delta are well-defined on all tensor products, including $\ca F_0 \otimes \ca F_0$,
\ba
\la \fr D_u\delta(x,x'), f(x)g(x')\ra &= - \la \delta(x,x'), \pu f(x) g(x')\ra \\
\la \fr D_{u'}\delta(x,x'), f(x)g(x')\ra &= - \la \delta(x,x'), f(x) \pup g(x')\ra
\ea
It is worthwhile to highlight that, in $(\ca F_0 \otimes \ca F_0)^*$, we have
\[\fr D_{u'}\delta(u,u') \ne - \fr D_{u}\delta(u,u')
\]
In fact,
\ba
\la \fr D_u\delta(x,x') + \fr D_{u'}\delta(x,x'), \alpha(x)\alpha'(x')\ra &= - \int\!\epsilon_\scri\left( \pu\alpha(x) \alpha'(x) + \alpha(x)\pu\alpha'(x) \right) \no
&= -\alpha_+ \alpha'_+ + \alpha_- \alpha'_-
\ea
so we conclude that
\be
\fr D_{u'}\delta(u,u') = - \fr D_{u}\delta(u,u') - L_+(u)L_+(u') + L_-(u)L_-(u')
\ee
It is convenient to define the symmetric and anti-symmetric parts of $\fr D_u\delta(u,u')$ as
\ba
\fr D_u^+\delta(u,u') &:= \frac{\fr D_u + \fr D_{u'}}{2}\delta(u,u') = -\frac{L_+(u)L_+(u') - L_-(u)L_-(u')}{2} \\
\fr D_u^-\delta(u,u') &:= \frac{\fr D_u - \fr D_{u'}}{2}\delta(u,u') = \fr D_u\delta(u,u') +\frac{L_+(u)L_+(u') - L_-(u)L_-(u')}{2}
\ea

\begin{definition}
The Heaviside bidistribution, $H(u,u')$, is defined on any tensor product which does not include any factors of $\ca F_0$. On $\ca F_1 \otimes \ca F_1$, for example,
\be
\la H(u,u'), \lambda(x)\lambda'(x')\ra := \int\!\epsilon_\scri \int\!\epsilon'_\scri\, H(u,u')\lambda(x)\lambda'(x')
\ee
where, on the right-hand side, $H(u,u')$ is the function which equals $1$ if $u>u'$ and $0$ otherwise.
\end{definition}
\noindent
One can easily verify that $H(u,u') + H(u',u) = 1$ and that
\be\label{DuHuu'}
\fr D_u H(u,u') = \delta(u,u') - L_+(u)
\ee
which also shows an unfamiliar behavior compared to distributions on compactly-supported test functions.

\subsubsection{Formulas for the brackets}
\label{DistBracketsFormulas}

Equipped with the  distributions constructed above, we can implement the procedure described in Sec.~\ref{SecDefDistB} to derive the distributional brackets.
Here we display the formulas for the distributional brackets, both in the proper (anti-symmetric) and directional forms. 
For the derivations, see App.~\ref{DistBracketsEval}.

The proper (anti-symmetric) distributional brackets can be used to compute Dirac brackets between tangentially extensible functions, according to \eqref{DistBracketsDef}. They are given by
\ba
\{\sigma(x), \bar\sigma(x')\}_D &= 2\pi\delta(\ang,\ang') \left( H(u,u') - H(u',u)\right) \\
\{\sigma(x), N(x')\}_D &= 4\pi \delta(x,x') + 2\pi\delta(\ang,\ang') L_-(u') \\
\{\bar N(x),N(x')\}_D &= -4\pi\,\fr D_u^-\delta(x,x') \\
\{\sigma(x),\sigma(x')\}_D &= 0 \\
\{\sigma(x),\bar N(x')\}_D &= 0 \\
\{\bar N(x),\bar N(x')\}_D &= 0
\ea
and all other brackets can be  deduced using
anti-symmetry and reality.
We note that the distributional brackets cannot be uniquely determined, due to the constraints satisfied by the allowed smearings. In App.~\ref{DistBracketsEval} we characterize the most general solution, and the expressions above can be regarded as the ``minimal choice''.
Remarkably, the shear-shear bracket matches the naive one that can be deduced ignoring boundary terms, see e.g. \cite{He:2014laa,Ashtekar:2018lor}. We do not have any intuitive reason as to why, but it is in any case clear that it would not be part of a consistent algebra without the non-local corrections to the other brackets (nor does it affect in any way the statement that the localized shear is not an observable).

The directional distributional brackets can be used to compute directional brackets where the first entry is a tangentially extensible function and the second is a smooth function, according to \eqref{DirDistBracketsDef}. They are given by
\ba
\{\sigma(x); \bar\sigma(x')\}_D &= - 4\pi \delta(\ang,\ang')H(u',u) \\
\{\bar N(x); \bar\sigma(x')\}_D &= -4\pi\delta(x,x') - 4\pi \delta(\ang,\ang') L_-(u) \\
\{\sigma(x); N(x')\}_D &= 4\pi \delta(x,x') \\
\{\bar N(x); N(x')\}_D &= -4\pi\, \fr D_u\delta(x,x') \\
\{\bar\sigma(x); \bar\sigma(x')\}_D &= 0 \\
\{\bar\sigma(x); N(x')\}_D &= 0 \\
\{ N(x); \bar\sigma(x')\}_D &= 0 \\
\{ N(x); N(x')\}_D &= 0
\ea
and all other brackets can be deduced using 
reality. Again, a ``minimal choice'' was made and the general solution is described in App.~\ref{DistBracketsEval}.

As a check of consistency, it is instructive to rederive the action of the supertranslation charge on the localized shear. We have
\ba
\{F_T^\text{\it hard}; \sigma(x)\}_D &= -\frac{1}{4\pi} \int\!\epsilon_\scri'\, T(x') \{N(x')\bar N(x'); \sigma(x)\}_D \no
&= -\frac{1}{4\pi} \int\!\epsilon_\scri'\, T(x') \left(N(x')\{\bar N(x'); \sigma(x)\}_D + \bar N(x')\{N(x'); \sigma(x)\}_D \right) \no
&= -\frac{1}{4\pi} \int\!\epsilon_\scri'\, T(x') \bar N(x') \left( -4\pi\delta(x',x) - 4\pi \delta(\ang',\ang) L_-(u') \right) \no
&= T(x) \bar N(x) \\
\{F_T^\text{\it soft}; \sigma(x)\}_D &= -\frac{1}{8\pi} \int\!\epsilon_\scri'\, \eth^2T(x') \{N(x'); \sigma(x)\}_D + \frac{1}{8\pi} \int\!\epsilon_\scri'\, \bar\eth^2T(x') \{\bar N(x'); \sigma(x)\}_D \no
&= -\frac{1}{8\pi} \int_\scri\!\epsilon_\scri'\, \eth^2T(x') \left( -4\pi\delta(x',x) - 4\pi \delta(\ang',\ang) L_-(u') \right) \no
&= \frac{1}{2} \eth^2\!\left(T + T_-\right) \big|_x \no
&=  \eth^2T(x) \label{TsoftDist}
\ea
where, in the last line, we used that $T$ is time-independent (so $T = T_-$). This agrees with \eqref{Txisigmax}.

A point to emphasize is that, as anticipated in the beginning of this section, we see that it would not be valid to define $\{N(x); \,\cdot\,\,\}_D$ by replacing $N(x)$ with $\bar\Gamma_{0,\delta_x}^{\bb C}$.
By taking $\alpha = \delta_x$ before computing the brackets, one loses all information about the asymptotic behavior of $\alpha$, and consequently fails to capture the non-local contributions to the brackets, encoded in the distributions $L_\pm$. 
In particular, if one were to compute $\{N(x'); \sigma(x)\}_D$ as $\{\bar\Gamma_{0,\delta_{x'}}^{\bb C}; \Gamma_{\delta_x,0}^{\bb C}\}_D$, they would obtain $\tfrac{1}{2}\eth^2T(x)$ instead of \eqref{TsoftDist}, for $T_-$ would be missed, as in \cite{He:2014laa}.

\section{Discussion}
\label{SecDiscussion}

In this paper we have developed a careful analysis of the algebra of proper observables on the Ashtekar-Streubel phase space of radiative gravitational degrees of freedom, with boundary conditions that account for distinct memory states. 
With this, we believe to have provided a more solid framework on which questions related to the gravitational dynamics at null infinity can be more clearly posed and, perhaps, answered. 
In particular, since quantization fundamentally refers to the algebra of proper observables, the results from this work may offer some insights into the quantum memory problem. 
In this discussion, we recapitulate the main conclusions from the paper and comment on some of their aspects and possible implications.

As we have emphasized, in an infinite dimensional phase space it is of essential importance to take into account the weak non-degeneracy of the symplectic form.
In particular, since the Ashtekar-Streubel phase space has second-class constraints, $\bar N = - \pu\sigma$, the algebra of {\it physical} observables is associated with the Dirac brackets.
In this paper, we have proposed a rigorous procedure to compute Dirac brackets in field theories subjected to an infinite set of second-class constraints, including non-trivial boundary conditions. The procedure starts by defining the largest set of regular constraints via suitable smearings of the local constraints. Then, given any regular function on the kinematical phase space, one 
extends it (if necessary and possible) by adding a regular constraint in such a way that the resulting symplectic flow is tangent to the constraint surface. Functions for which such an extension is possible are said to be tangentially extensible, and Dirac brackets between such functions are defined as Poisson brackets between their tangential extensions. 
As we have clarified, the Ashtekar-Streubel phase space has the peculiar feature that not all regular functions are tangentially extensible.
In fact, for the shear-news functions, $\Gamma_{\lambda, \alpha}$, the smearing parameters were subjected to a ``dual'' constraint, $\int\!du\, \lambda = \alpha_+ - \alpha_-$. 
This procedure allows us to compute the Dirac bracket between any
tangentially extensible shear-news functions.

Two key results of our constructions are that we
clarified the status of the (localized) shear as an observable, and how to recover
the correct action of supertranslations.
We showed that the localized shear, $\sigma(x)$, is not tangentially extensible even in a non-smooth sense, since $(\lambda,\alpha) = (\delta_x,0)$ does not satisfy the dual constraint ($1\ne 0$). That is, there is not even a ``distributional'' symplectic flow associated with $\sigma(x)$, and consequently there is no consistent way to define its Dirac brackets with other observables.
This explains why computations based on a bracket like \eqref{NNbra}, but involving the shear, are ill-motivated.
In particular, 
the Dirac bracket ``$\{F_T, \sigma(x)\}_D$'' is not well-defined, and one should not expect to obtain the correct action of the supertranslation on the shear in this manner. 
Nonetheless, we have explained that it is possible to define ``directional brackets'', where only one entry is a proper observable. Applying this concept to $\{F_T; \sigma(x)\}_D$, we obtained the correct  action of the supertranslation on the shear,  with the right ``factor of $2$''. 
This result crucially relies on properly including non-local aspects associated with the tangential extensions into the Dirac brackets.
Namely, we found that $\{F_T^\text{\it soft}; \sigma(x)\} = - \frac{1}{2}\eth^2(T(x) + T_-(\ang))$, which contains the non-local term $T_-$.
It is the time-independence of $T$ that implies $T + T_- = 2T$, and recovers the correct ``factor of $2$''.
Note that any attempt to compute this directional bracket based on first evaluating brackets between local fields, like $\{N(x'); \sigma(x)\}$, will fail to capture such non-local aspects. In particular, if one first represents the soft term by smearing $N(x')$ with $\eth^2\delta(x',x'')$ and then computes the resulting brackets, the boundary contribution is lost because the Dirac delta is compactly supported. Proceeding with integration against $T(x'')$ would thus give only $\{F_T^\text{\it soft}; \sigma(x)\} = - \frac{1}{2}\eth^2T(x)$, as the term $T_-$ went missing.
This provides a concrete understanding of the missing ``factor of $2$'' in \cite{He:2014laa} from the perspective of the proper Dirac brackets.

One of the main results of this paper was to identify the complete algebra of proper (gravitational) observables at $\scri$, which includes non-central memory, with the given boundary conditions. The algebra,
$\scr A$, is characterized by all tangentially extensible functions on the Ashtekar-Streubel phase space. The Dirac analysis for non-linear observables is essentially analogous to that applied for shear-news functions. In particular, the condition for tangential extensibility becomes simply that the map from the constrained phase space to the pair of functional derivatives in $\bar\sigma$ and $N$ lands into $\ca G_{TE}$ and is smooth.
An important result is that the only functionals of the ``soft modes'' (i.e., asymptotic values of the shear) which are equal, on-shell, to proper observables are functionals of the memory mode alone.
This means that the Goldstone mode cannot be directly measured by any proper observables. Accordingly, since Dirac brackets involving them are technically ill-defined, any results that rely on some ``ad hoc'' prescription for their brackets must be taken with reasonable caution.
While the Goldstone mode cannot be directly made into an observable, we have described a large family of proper observables, the Goldstone probes, that play a similar role. Namely, they are symplectically conjugated to the memory (albeit not in a canonical way) and are thus capable of indirectly measuring the Goldstone mode.
In this manner, $\scr A$ is not blind to the presence of the Goldstone mode, and any physically meaningful information about this mode is contained in the algebra. 
In the paper, we also discussed sequences of Goldstone probes that ``approximate'' the Goldstone mode. As expected, these sequences 
do not converge to proper observables. At decreasing levels of ``regularity'', they would first converge  to a ``dressed shear'', $\sigma(x) + \int_u^\infty\!du'\, \bar N(u',\ang)$, which can be thought of as a ``generalized observable'' like the localized news (i.e., associated with a non-smooth flow that formally satisfies the dual constraint), and second they would converge to $\sigma_+$ as $u\to \infty$, which is completely irregular, just like the undressed local shear. 
Consequently, even if one could naturally define a ``dressed Goldstone'' as a limit of the ``dressed shear'', there is no consistent way to define brackets for such quantities.
We made this point particularly explicit by showing
that different sequences of proper observables with the same pointwise limits representing the Goldstone mode can lead to different limits for the Dirac brackets.
But it should be stressed that this only happens when considering sequences of non-linear observables, as the limits defined from sequences of shear-news functions always define the same limiting brackets between the Goldstone mode and other proper observables, which agree with the proposal in \cite{He:2014laa}. 
These limiting brackets are actually directional brackets, and it remains impossible to define brackets {\it among} Goldstone modes by a similar procedure
due to different orderings one could choose for the limits.

In the paper we restricted attention to the standard boundary conditions with a purely electric shear, but
our results would not be modified in any essential way if we had considered boundary conditions where the shear approaches arbitrary (well-defined) limits. Some derivations would change slightly, but not the conclusions. For example, when deriving the space of smearings $\ca K$ which make the constraints regular, we argued that the asymptotic values of the smearing $\varphi$ would have to be purely magnetic to prevent boundary terms in the variation of $K_\varphi$, and then it would also have to be purely electric to preserve the boundary conditions, implying that $\varphi_\pm = 0$. Had we considered boundary conditions with arbitrary limits, we would have directly concluded that $\varphi_\pm = 0$ in order to prevent boundary terms in the variation of $K_\varphi$. By inspecting all the derivations in the paper, one can verify that the dual constraint $\wt K$, the formulas for the Dirac brackets, the correct action of supertranslations and the formulas for the distributional brackets would remain unaltered for boundary conditions with arbitrary asymptotic shears. (Other kinds of boundary conditions, like periodic conditions, could however affect the results.)

Let us also comment on the complex and distributional brackets that we introduced in this paper. Those are extensions of the Poisson (or Dirac) brackets to include complex functions and ``localized'' fields, to be understood as tools for computing brackets of real, proper observables. 
The complex brackets were defined simply as an algebraic extension of proper brackets to complex-valued functions. We stress that they are not regular brackets since they generally do not preserve the boundary conditions of purely electric shears. 
Regarding the distributional brackets, it is interesting to point out that there are two fundamentally different notions which are often conflated. For illustration, consider a standard scalar field theory whose phase space contains a field $\phi(x)$ and its conjugate momentum $\pi(x)$, with symplectic form $\Omega = \int\!dx\,\delta\pi(x)\,\delta\phi(x)$.
Let $\Phi_f := \int\!dx\, f(x)\phi(x)$ and $\Pi_g := \int\!dx\, g(x)\pi(x)$ be proper smearings of the fields, where $f$ and $g$ fall off sufficiently fast towards the boundaries. The Poisson brackets would be $\{\Phi_f, \Pi_g\} = \int\!dx\, f(x)g(x)$, and the two notions for distributional brackets are the following:
\begin{enumerate}
\item Define $\{\phi(x), \pi(x')\} := \{\Phi_{\delta_x}, \Pi_{\delta_{x'}}\} = \delta(x,x')$. This is equivalent to attempting to derive $\{\phi(x), \pi(x')\}$ directly from the symplectic form;

\item Define $\{\phi(x), \pi(x')\}$ as distributional objects that allow one to recover the proper brackets by integrating functional derivatives against them. Namely, given two proper observables $F$ and $G$, let 
\[
\{F,G\} = \int\!dx\int\!dx'\, \frac{\delta F}{\delta \phi(x)} \{\phi(x), \pi(x')\} \frac{\delta G}{\delta \pi(x')} + \cdots
\]
In this case, it follows that $\{\phi(x), \pi(x')\} = \delta(x,x')$.
\end{enumerate}
Notice that the two notions give the same result for the bracket $\{\phi(x), \pi(x')\}$, equal to $\delta(x,x')$, and this happens in many familiar situations. 
However, in the case of $\scri$ they would disagree, since notion $1$ would fail to capture the non-local aspects of the brackets, as mentioned before.
We advocate that notion $2$ is the most useful one, and that is how we defined the distributional brackets between shears and news. 
In fact, we obtain boundary contributions associated with smearings that do not vanish asymptotically.
With these contributions naturally incorporated, it is possible to re-obtain the correct action of the supertranslations on the shear. 
We suggest that any results in the literature derived from ``localized'' brackets (defined from notion $1$) should be revisited in light of the corrected form of the distributional brackets.

Finally, it would also be valuable to study in greater mathematical detail the method for constructing Dirac brackets for infinite-dimensional systems introduced here, and understand under which conditions, topologies, etc, it would apply for any given infinite dimensional phase space. For the purposes of this paper, we contented ourselves with establishing its consistency with the results obtained from the reduced phase space.

The results presented in this paper apply to the BMS group and purely electric boundary conditions at early and late times. Our methods can also be applied to different boundary conditions, or larger symmetries corresponding to weaker fall-off conditions.
Let us mention some possible directions of further developments. It has been recently shown that the symplectic center-less representation of the BMS algebra holds also for the partial flux between two arbitrary cross sections \cite{Rignon-Bret:2024gcx,Rignon-Bret:2024wlu}, thanks to the general covariance of the Ashtekar-Streubel symplectic potential \cite{Wald:1999wa,Odak:2022ndm}. It would then be interesting to apply our methods to the partial radiative phase space between two cross sections, where the non-trivial boundary conditions may now include a non-vanishing news, in order to have a well-defined notion of Dirac bracket also in this context. Larger symmetries, such as eBMS \cite{Barnich:2010ojg,Barnich:2011mi,Rignon-Bret:2024gcx}, gBMS \cite{Campiglia:2014yka,Compere:2018ylh} or BMSW \cite{Freidel:2021yqe}, all require corner improvements of the symplectic form, as well as potentially different boundary conditions. These are likely to give additional non-local boundary contributions to the distributional brackets, as proposed for instance in \cite{Campiglia:2024uqq} for gBMS. Our methods can be applied also to these larger symmetries, and it would be interesting to do so to understand their algebra of proper observables, especially since these larger symmetries are also motivated by quantum investigations. For further discussions on the hard/soft factorization of the phase space, see e.g.
\cite{Himwich:2020rro,Arkani-Hamed:2020gyp,Campiglia:2021bap,Campiglia:2020qvc,Guevara:2021abz,Donnay:2022hkf,Rignon-Bret:2024mef,Ciambelli:2024swv}.

Our classical analysis of the algebra of observables leads to some insights into the quantum memory problem. 
The attempts so far to quantize the radiative gravitational degrees of freedom have failed to incorporate states with arbitrary memory within a single, separable Hilbert space (e.g. \cite{Ashtekar:1981sf,Ashtekar84book,He:2014laa,Strominger:2017zoo,Ashtekar:2018lor,Prabhu:2022zcr}; see also \cite{Campiglia:2021bap,Rignon-Bret:2024mef,Strominger:2013jfa,Himwich:2020rro,Arkani-Hamed:2020gyp,Guevara:2021abz,Donnay:2022hkf,Ciambelli:2024swv,Freidel:2022skz,Riello:2023ptb,Prabhu:2024zwl,Prabhu:2024lmg,Agrawal:2025bsy,Donnay:2026urd,danielson2026horizons}). 
However, they usually contain at least one of the following features:
\begin{enumerate}
\item The representation is constructed out of the algebra of the shear {\it or} news only;
\item A Fock space construction based on a split of positive/negative frequencies, based on a Fourier decomposition;
\item A quantization based on the algebra of hard modes plus the ({\it ad hoc}) algebra of soft modes, which includes the Goldstone mode.
\end{enumerate}
Each of these features leads to a shortcoming, either technical or conceptual, if one accepts that only proper observables should be quantized.
Regarding point 1, any algebra of proper observables constructed only out of the shear or the news, 
e.g. $\{\Gamma_{\lambda,0}\}$ or $\{\Gamma_{0,\alpha}\}$, will not contain any Goldstone probes because of the dual constraint. Therefore all memory observables will belong to the center of the algebra, and by Schur's lemma the memory will be proportional to the identity in any (complex) irreducible representation. 
Point 2 suffers from the same issue,
as it is necessary to restrict to smearing parameters that fall off sufficiently rapidly as $u\to\pm\infty$, 
thereby excluding all Goldstone probes from the algebra.
Point 3 fails to take seriously the principle that only proper observables should be quantized, since the Goldstone mode (as usually defined) is not a proper observable, casting doubt on the conceptual foundation of such approaches.

It has been proposed \cite{Prabhu:2022zcr} that perhaps one should abandon the goal of constructing a Hilbert space and instead be content with only an algebraic notion of the $S$-matrix. However, we believe that until the points above are carefully addressed in the quantization, the idea of a standard $S$-matrix, mapping between separable Hilbert spaces, should not be entirely discarded.
In canonical quantization, it is generally necessary to restrict to a minimal (or quasi-minimal) subalgebra of observables, to avoid Groenewold-Van Hove obstructions.
We have shown that the subalgebras consisting of shear-news functions $\Gamma_\alpha := \Gamma_{\pu\alpha, \alpha}$, with $\alpha \in \ca F_0$, are complete, and they can even be restricted further to subalgebras where $\alpha \in \ca F_0^\pm$ (i.e., where $\alpha_+$ {\it or} $\alpha_-$ vanishes).
We do not claim that it is possible to construct a separable Hilbert space by quantizing such subalgebras, but it would be interesting to pursue this idea further.

As a preliminary step in this direction, let us observe the following.
In Fock quantization, one usually constructs the $1$-particle Hilbert space from a split of the solution space into positive and negative frequencies, and lets the vacuum be annihilated by all negative frequency modes, and particles be created by acting on the vacuum with the positive modes. 
The main obstruction here is that such
a negative/positive frequency split is defined from the Fourier transform, which requires one to work in a space of functions that fall off sufficiently rapidly at infinity, thereby excluding all Goldstone probes.
Nevertheless, it is understood 
\cite{Ashtekar:1980yw,Najmi:1984vn,Najmi:1985omb} that the fundamental ingredient in Fock quantization is a K\"ahler structure on the space of solutions (or, in our case, on the space of smearings), turning the (anti-symmetric) symplectic form appearing in the commutation relations,
\[
\{\Gamma_{\a}, \Gamma_{\a'}\}_D = 4\pi\,\re\!\int\!\eps_\scri\left(\a' \pu\bar\a - \bar\a \pu\a'\right) =: \omega(\a,\a')
\]
into a (positive-definite) inner-product for the Hilbert space.\footnote{The negative/positive frequency split is just a particular realization of a complex structure that does this job in the usual scenarios.}
Thus, if there exists a complex structure $J$ on $\ca F_0$ (or $\ca F_0^\pm$) compatible with $\omega$, one can define annihilation and creation operators
\[
a_\a^{} = \frac{1}{\sqrt{2}}\left(\Gamma_\a + i \Gamma_{J\a}\right) \,,\,\quad  a_\a^\dag = \frac{1}{\sqrt{2}}\left(\Gamma_\a - i \Gamma_{J\a}\right)
\]
and it would follow that $\la\a|\a\ra = \omega(\a,J\a) \ge 0$, where $|\a\ra := a_\a^\dag|0\ra$ and $a_\a|0\ra = 0$.
To ensure invariance of the vacuum (up to a phase) under global time translations, 
it seems necessary (assuming no anomalies in the commutator of the Bondi flux, $F_1$, with the shear-news functions) to require also compatibility with $u$-differentiation, $[\pu, J] = 0$.
This may pose a real obstruction to the existence of such a complex structure. But, if such a $J$ exists, then it is promising that this proposal would imply that the Bondi flux is bounded from above (also assuming no anomalies).

\subsection*{Acknowledgements}
We would like to thank Andrew Strominger and Temple He for bringing up the question of the ``problem of $2$'' during a seminar, and Abhay Ashtekar, Glenn Barnich, Luca Ciambelli, Laurent Freidel, Ted Jacobson and Simon Langenscheidt for discussions on the topics of this paper.
This work was supported in part by Perimeter Institute for Theoretical Physics. 
Research at Perimeter Institute is supported in part by the Government of Canada through the Department of Innovation, Science and Economic Development and by the Province of Ontario through the Ministry of Colleges and Universities.

\vskip 5em

\appendix
\section{Dirac brackets in infinite dimensions}
\label{AppDirbrackets}

In this appendix we propose a rigorous method for computing Dirac brackets on infinite dimensional phase spaces, like those of field theories, subjected to an infinite set of (second-class) constraints.
While it is directly inspired by a well-known formula in finite dimensions, we have not seen in the literature its enunciation in this form, for infinite dimensional phase spaces.

Let us first review the finite dimensional case, with a finite number of constraints $K_i$, $i \in \{1,2,\ldots s\}$. This set of constraints is said to be of \emph{second class} if their Poisson brackets matrix,
\be
\Delta_{ij} := \{K_i, K_j\}
\ee
is invertible (on shell). 
Then the Dirac bracket
between two smooth phase space functions $F$ and $G$ is defined by
\be\label{DBdiscrete}
\{F, G\}_D := \{F, G\} - \sum_{ij} \{F, K_i\} \Delta^{-1}_{ij} \{K_j, G\}
\ee
This bracket has all the typical properties of a Poisson bracket, and it is designed so as to ``project out'' the flow generated by the constraints away from the constraint surface, i.e., $\{F, K_i\}_D \doteq 0$ for all $F$ and constraint $K_i$. This means that the Dirac brackets depend only on the on-shell value of their arguments, i.e., if $F \doteq F'$ 
then $\{F, G\}_D \doteq \{F', G\}_D$ for all $G$. In fact, it is true that $\{F, G\}_D$ is equal, on-shell, to the Poisson brackets associated with the induced symplectic form on the constraint surface --- this is the defining property of the Dirac brackets.

In the infinite dimensional case the constraints may be labeled with a collection of discrete and continuum indices, $K_{i\alpha}$, and the ``naive'' generalization of the Dirac brackets, for smooth functions $F$ and $G$ would take the form
\be\label{Diracnaive}
\{F, G\}_D = \{F, G\} - \sum_{ij} \int\! d\alpha\! \int\! d\beta\, \{F, K_{i\alpha}\} \Delta^{-1}_{i\alpha,j\beta} \{K_{j\beta}, G\}
\ee
where $\Delta^{-1}_{i\alpha,j\beta}$ must be interpreted as a Green's function for the differential equation related to the operator $\Delta_{i\alpha,j\beta} = \{K_{i\alpha},K_{j\beta}\}$. 
There are, however, some subtleties with this formula. First, it is necessary to determine what is the correct measure used in the ``sum'' over $\alpha$ and $\beta$.
Second, it is necessary to take into careful account the boundary conditions used for the Green's function $\Delta^{-1}_{i\alpha,j\beta}$.
Moreover, there is a serious obstruction in making \eqref{Diracnaive} even well-defined.
In infinite dimensions, the symplectic form is usually only {\it weakly} non-degenerate, so not all smooth functions on the phase space can be associated with a (smooth) Hamiltonian vector field, and therefore their Poisson brackets cannot be defined from the symplectic structure. In particular, $\Delta_{i\alpha,j\beta}$ and  $\{F,K_{i\alpha}\}$ are often ill-defined. 
In order to construct {\it regular} constraints it is necessary to smear $K_{i\alpha}$ with suitable functions $f(i\alpha)$, defining $K_f := \sum_i\int\! d\alpha \, f(i\alpha) K_{i\alpha}$. For a field theory, $\alpha$ may label points in a Cauchy slice, and the $f$'s may need to satisfy special boundary (or asymptotic) conditions in order for $K_f$ to be a regular constraint. 

The cleanest way to deal with these complications is to note a convenient, and well-known \cite{henneaux1992quantization}, reformulation of formula \eqref{DBdiscrete}. Consider the following:

\begin{theorem}
Let $F$ and $G$ be regular phase space functions and $\{K_i,\, i=1,2,\ldots s\}$ be a finite set of regular second class constraints. Then there exists a differentiable function $\wt F$, called the ``tangential extension of $F$'', such that $\wt F \doteq F$ and 
\be
\{\wt F, K_i\} \doteq 0\,,\,\, \text{for all $i$}
\ee
Moreover, it follows that
\be
\{F, G\}_D \doteq \{\wt F, \wt G\} \doteq \{\wt F, G\} \doteq \{F, \wt G\}
\ee
\end{theorem}
\begin{proof}
Since the set of constraints is regular (i.e., any point in the constraint surface has a phase space neighborhood  that can be covered with coordinates $\{K_i\}\cup \{\chi_\mu\}$, such that the constraint surface is the level set $K_i = 0$, for all $i$, and $\{\chi_\mu\}$ are transverse coordinates), then $\wt F \doteq F$ implies that 
\be
\wt F = F + \varphi_i K_i
\ee
for some set of functions $\varphi_i$ (with the convention of summing over repeated indices). We wish to solve $\{\wt F, K_j\} \doteq 0$, which is equivalent to
\be
\{\wt F, K_j\} = \mu_{jk} K_k
\ee
for some set of functions $\mu_{jk}$. Thus we have
\be
\{F + \varphi_i K_i, K_j\} = \{F, K_j\} + \varphi_i \{K_i, K_j\} +  \{\varphi_i, K_j\} K_i = \mu_{jk} K_k
\ee
but since $\{K_i, K_j\} = \Delta_{ij}$ is invertible we have
\be
\varphi_l = - \{F,K_j\}\Delta^{-1}_{jl} + \mu'_{jk}K_k \Delta^{-1}_{jl}
\ee
where $\mu'_{jk} := \mu_{jk} - \{\varphi_k,K_j\}$ is some function. Therefore the tangential extension of $F$ always exists and is given explicitly by
\be
\wt F = F - \{F,K_i\}\Delta^{-1}_{ij} K_j + \text{\it ``terms quadratic in $K_i$''}
\ee
Note that the quadratic terms in $K_i$ are irrelevant in any on-shell computation of Poisson brackets. Finally,
\be
\{\wt F, G\} = \{F - \{F,K_i\}\Delta^{-1}_{ij} K_j, G\} \doteq \{F, G\} - \{F,K_i\}\Delta^{-1}_{ij} \{K_j, G\} = \{F, G\}_D
\ee
Moreover, since $\{\wt F, K_i\} = 0$, then $\{\wt F, G\} \doteq \{\wt F, \wt G\}$,
proving the theorem.
\end{proof}

This alternative expression for the Dirac brackets has a great advantage in the generalization for field theories: there is no sum over continuous indices appearing in the final formula, there is no explicit mention of $\Delta^{-1}$ and, most importantly, it is naturally adapted for working with smeared constraints. We thus propose the following:

\begin{proposal}
Consider a phase space $\mathcal P$ subjected to a collection of regular constraints $K_\varphi$ whose labels $\varphi$ are elements of a given linear space of functions $\ca K$.\footnote{More precisely, in this language the constraints are described by a linear map $\ca K \rightarrow C^\Omega(\mathcal P, \mathbb R)$, where $C^\Omega(\mathcal P, \mathbb R)$ is the space of symplectically regular real functions on $(\mathcal P, \Omega)$, meaning that $\delta K_\varphi = - I_{X_\varphi}\Omega$ has at least one smooth vector field solution $X_\varphi$.} We say that the constraints are ``second class'' if 
\be\label{secondclass}
\{K_\varphi, K_{\varphi'}\} \doteq 0\,\text{ for all $\varphi' \in \ca K$} \quad \Longleftrightarrow \quad \varphi = 0
\ee
Let $F$ be a regular function on the phase space. If possible, define the ``tangential extension of $F$'', denoted by $\wt F$, as the function given by
\be
\wt F = F + K_\varphi
\ee
where $\varphi \in \ca K$, if it exists, is the unique solution of 
\be\label{varphieq}
\{K_\varphi ,K_{\varphi'}\} \doteq - \{F, K_{\varphi'}\}\,\text{ for all $\varphi' \in \ca K$}
\ee
This is equivalent to demanding that $\wt F \doteq F$ and $\{\wt F, K_{\varphi'}\}\doteq 0$ for all $\varphi' \in \ca K$. Note that $\varphi$ may depend on the phase space point. 
If a solution $\varphi$ exists, then $F$ is said to be ``tangentially extensible''.
Note that $\wt F$ is regular since it is the sum of two regular functions. 
If $F$ and $G$ are two tangentially extensible functions, their Dirac bracket is given by
\be
\{F, G\}_D \doteq \{\wt F, \wt G\} \doteq \{\wt F, G\} \doteq \{F, \wt G\}
\ee
which is a well-defined, unambiguous formula. 
\end{proposal}

The reason for assuming that $\ca K$ is a linear space is that we wish to think of $\varphi$ as ``smearing functions'' over the set of constraints. If $K_\varphi$ and $K_{\varphi'}$ are two regular constraints, then $a K_\varphi +  a' K_{\varphi'}$ is another regular constraint, for any $a, a' \in \mathbb R$ (or $\bb C$, if the phase space is complex). Therefore, it is natural to require that $a\varphi + a' \varphi' \in \ca K$ and define $K_{a\varphi + a' \varphi'} := a K_\varphi + a' K_{\varphi'}$. 
From the linearity of $\varphi \mapsto K_\varphi$, we can easily see that the second-class condition \eqref{secondclass} implies that if a solution of \eqref{varphieq} exists, then it must be unique: if \eqref{varphieq} had two solutions $\varphi_1$ and $\varphi_2$, then by subtracting the equations we would have, for all $\varphi' \in \ca K$,
\be
0 \doteq \{K_{\varphi_2}, K_{\varphi'} \} - \{K_{\varphi_1}, K_{\varphi'} \}= \{K_{\varphi_2}-K_{\varphi_1}, K_{\varphi'} \} =  \{K_{\varphi_2-\varphi_1}, K_{\varphi'} \} 
\ee
implying that $\varphi_2 - \varphi_1 = 0$. 
The existence of a solution is not guaranteed, so generally we expect to have the following inclusion of function spaces
\[
\text{\it tangentially extensible} \subset \text{\it regular} \subset \text{\it smooth}
\]
However, we have not attempted to prove this relation.\footnote{It is conceivable, for example, that there are tangentially extensible functions which are not regular, i.e., $F + K_\varphi$ may be regular for some $\varphi \notin \ca K$. However, this does not need to be considered in the case of this paper, since we showed that our prescription produces all possible proper observables (confirmed by the reduced phase space analysis).}

\section{A lemma related to Theorem~\ref{nogravitonobs}}
\label{AppLemma1formPK}

In this appendix we prove a lemma used in the proof of Theorem~\ref{nogravitonobs}.
To obtain \eqref{ThetaFCf}, one takes $\mu = \delta(\Theta - F)$ in the lemma below.

\begin{lemma}\label{lemma1formPK}
If a continuous 1-form $\mu : T\ca P \to \bb R$ is annihilated by all vectors tangent to $\ca P_K$, then $\mu \doteq \delta K_\varphi$, where $\varphi$ is an integral representative of some (possibly field-dependent) real, continuous distribution on $\ca F_1$.
\end{lemma}
\begin{proof} 
Viewing the local constraint $K$ in \eqref{constraint1} as a map from $\ca P$ to $\ca F_1$, namely $p \mapsto (\bar N + \pu\sigma)|_p$, its exterior derivative can be expressed as a linear map $\delta K : T\ca P \to \ca F_1$,
\be\label{dKmap}
\delta K(X) = X\bar N + \pu(X\sigma)
\ee
A vector $V$ is tangent to $\ca P_K$ if and only if $\delta K(V) = 0$. That is, $T\ca P_K \simeq \textrm{ker}(\delta K)$. 
Given any real vector $X \in T_p\ca P$, with $p\in \ca P_K$, there exists a $V$ such that $V\sigma = - X\sigma$, since both sets can span arbitrary values of the shear compatible with its boundary conditions. It follows that we can
define $\wt X := X + V$ such that $\wt X\sigma = 0$, and this vector satisfies
\ba
\wt X\sigma &= X\sigma + V\sigma = 0 \\
\wt X\bar N &= X\bar N + V\bar N = X\bar N - \pu(V\sigma) = X\bar N + \pu(X\sigma) = \delta K(X) \label{lemmawtXbarN}
\ea
Now let $\Phi$ be a (possibly field-dependent) real, continuous distribution on $\ca F_1$ given by
\be
\Phi(f)\big|_p := \mu(X_f)\big|_p
\ee
where $f\in \ca F_1$ and $X_f$ is a vector defined by
\be
X_f\sigma = 0\,,\quad X_f\bar N = f
\ee
Since $\mu$ is annihilated by $\textrm{ker}(\delta K)$, then for every $X \in T_p\ca P$, with $p\in \ca P_K$, 
\be
\mu(X) = \mu(\wt X) = \Phi(\delta K(X))
\ee
where we have used \eqref{lemmawtXbarN}.
Formally, the distribution $\Phi$ can be represented as an integral
\be
\Phi(f) = \frac{1}{4\pi}\re\int\!\epsilon_\scri\, \bar\varphi(x) f(x)
\ee
where $\varphi$ is said to be an {\it integral representative} of $\Phi$.
In this manner, the relation above is expressed as
\be
\mu(X) = \frac{1}{4\pi}\re\int\!\epsilon_\scri\, \bar\varphi\, \delta K(X)
\ee
and, as this is true for every $X \in T_p\ca P$, with $p\in \ca P_K$,
\be
\mu = \frac{1}{4\pi}\re\int\!\epsilon_\scri\, \bar\varphi\, \delta K \doteq \delta\left[\frac{1}{4\pi}\re\int\!\epsilon_\scri\, \bar\varphi K \right] = \delta K_\varphi
\ee
where ``$\doteq$'' means here that the forms are based at $p\in \ca P_K$ (but they are not restricted to $T_p^*\ca P_K$).
\end{proof}

\section{Evaluation of the distributional brackets}
\label{DistBracketsEval}

In this appendix we present the derivation of the formulas for the distributional brackets, displayed in Sec.~\ref{DistBracketsFormulas}.
We will begin with some general results which will be useful, then proceed with the directional brackets (whose derivation is slightly simpler) and finally discuss the proper brackets.

In solving for the distributional brackets, we will often encounter equations where a distribution defined on a certain space of test functions vanishes when restricted to a subspace. For example, $\{\bar\sigma(x), \bar\sigma(x')\}_D$ is a bidistribution in $(\ca F_1 \otimes \ca F_1)^*$ and, as we will see, the equation it satisfies is
\be
\int\!\epsilon_\scri\!\int\!\epsilon'_\scri\,\rho(x)\rho'(x') \{ \bar\sigma(x), \bar\sigma(x')\}_D = 0 \quad\text{\it for all }\, \rho\rho' \in \ca F_1^0 \otimes \ca F_1^0
\ee
The following theorem characterizes the solutions of the analogous equation for univariate distributions.
\begin{theorem}\label{theodist0}
If $\Phi \in \ca F_1^*$ and 
\be\label{Phirho0cond}
\la\Phi, \rho\ra = 0 \quad \text{for all $\rho \in \ca F_1^0$}
\ee
then
\be
\Phi(x) = \phi(\ang)
\ee
for some distribution $\phi \in \ca E^*$.\footnote{Recall that $\ca E := \eth^2C^\infty(S^2, \bb R)$ is the space of purely electric, complex spin-weight $2$ smooth functions on the sphere. Any distribution $\phi$ in $\ca E^*$ is also naturally a distribution in $\ca F_1^*$, acting as
\[
\la\phi(\ang), \lambda(x)\ra = \int\!\epsilon_\scri\, \phi(\ang)\lambda(x) = \int_S\!\epsilon_S\, \phi(\ang) \int_{-\infty}^{\infty}\!du\,\lambda(u,\ang) = \la \phi(\ang), \Lambda_+(\ang)\ra_S
\]
where $\Lambda_+(\ang) = \int_{-\infty}^{\infty}\!du\,\lambda(u,\ang) \in \ca E$ and $\la \,\cdot\,, \cdot\,\ra_S$ is the pairing for distributions acting on $\ca E$.\label{ftphiext}}
\end{theorem}
\begin{proof}
Any distribution in $\ca F_1^*$ which vanishes when restricted to $\ca F_1^0$ descends to a distribution on the quotient, $\ca F_1/\ca F_1^0$.
Consider the $u$-integration map, $I_u : \ca F_1 \to \ca E$, given by
\be
I_ug(\ang) := \int\!du\, g(u,\ang)
\ee
for $g\in \ca F_1$. Then $I_u$ is surjective onto $\ca E$, and $\ca F_1^0 = \text{\rm ker}(I_u)$, so
\be
\ca F_1/\ca F_1^0 \simeq \text{\rm im}(I_u) = \ca E
\ee
with $I_u$ realizing the quotient map.
The distributions $\Phi \in \ca F_1^*$ satisfying \eqref{Phirho0cond} are thus the pullbacks of distributions $\phi \in \ca E^*$ under $I_u$, $\Phi = I_u^*\phi$.
That is, for $\lambda\in \ca F_1$,
\be
\la\Phi, \lambda\ra = \la\phi, I_u\lambda\ra_S = \la\phi, \Lambda_+\ra_S
\ee
where $\la \,\cdot\,, \cdot\,\ra_S$ is the pairing on $\ca E$.
With a slight abuse of notation, according to footnote~\ref{ftphiext}, this is written simply as $\Phi(x) = \phi(\ang)$.
\end{proof}
This result can also be extended to bidistributions,
\begin{theorem}\label{theodist}
If $\Phi \in (\ca F_1 \otimes \ca F_1)^*$ and 
\be\label{Theorholambda}
\la\Phi(x,x'), \rho(x)\lambda(x')\ra = 0 \quad \text{\it for all }\, \rho\lambda \in \ca F_1^0 \otimes \ca F_1
\ee
then
\be\label{Theorholambdasol}
\Phi(x,x') = \phi(\ang,x')
\ee
where $\phi(\ang,x') \in (\ca E \otimes \ca F_1)^*$. 
(The distribution $\phi(\ang,x')$ on the right-hand side of the equality above is understood as the pull-back of a distribution in $(\ca E \otimes \ca F_1)^*$ under $I_u \times \textrm{id} : \ca F_1 \otimes \ca F_1 \to \ca E \otimes \ca F_1$, where $\textrm{id}$ is the identity map.)

If $\Phi \in (\ca F_1 \otimes \ca F_1)^*$ and 
\be\label{Theorhorho}
\la\Phi(x,x'), \rho(x)\rho'(x')\ra = 0 \quad \text{\it for all }\, \rho\rho' \in \ca F_1^0 \otimes \ca F_1^0
\ee
then
\be\label{Theorhorhosol}
\Phi(x,x') = \phi(\ang, x') + \phi'(x,\ang')
\ee
where $\phi(\ang,x') \in (\ca E \otimes \ca F_1)^*$ and $\phi'(x,\ang') \in (\ca F_1 \otimes \ca E)^*$.
(Again, in the equality above, these are understood as pull-backs under $I_u \times \textrm{id}$ and $\textrm{id} \times I_u$, respectively.) 
\end{theorem}
\begin{proof}
The distributions in \eqref{Theorholambdasol} and \eqref{Theorhorhosol} are clearly solutions to equations \eqref{Theorholambda} and \eqref{Theorhorho}, respectively. We need to show that these are the general solutions.
Choose a fixed compactly-supported smooth function $\chi:\bb R\to \bb R$ satisfying $\int\!du\,\chi(u) = 1$ and define the map $S:\ca E\to\ca F_1$ by
\be
S(\varsigma)(x) := \chi(u)\varsigma(\ang)
\ee
We see that $S$ is a right-inverse for $I_u$ (i.e., $I_u \circ S = \textrm{id} : \ca E \to\ca E$) and $P := S\circ I_u : \ca F_1\to\ca F_1$ is a projection map. Then $R := \textrm{id} - P$ is a projector from $\ca F_1$ to $\ca F_1^0$. In fact,
\be
R(\lambda)(x) = \lambda(x) - \chi(u)\Lambda_+(\ang)
\ee
and $\int\!du\,R(\lambda) = 0$.
Moreover, it is surjective onto $\ca F_1^0$ (in fact, $\rho = R\rho$). 
Then, given any bidistribution in $(\ca F_1\otimes\ca F_1)^*$ we can decompose
\ba
\la\Phi, \lambda\lambda'\ra &= \la \Phi, (P+R)\lambda\,(P+R)\lambda'\ra \no
&= \la \Phi, P\lambda\,P\lambda'\ra + \la \Phi, P\lambda\,R\lambda'\ra + \la \Phi, R\lambda\,P\lambda'\ra + \la \Phi, R\lambda\,R\lambda'\ra
\label{PhiPRdecomp}
\ea
If $\Phi$ vanishes when restricted to $\ca F_1^0\otimes\ca F_1$, the last two terms will vanish and thus
\be
\la\Phi, \lambda\lambda'\ra = \la\Phi, (P\lambda)\lambda'\ra
\ee
where we regrouped the first two terms into one.
Now define the distribution $\phi \in (\ca E\otimes\ca F_1)^*$ by
\be
\la \phi(\ang,x'), \varsigma(\ang)\lambda'(x')\ra_S := \la \Phi(x,x'), S\varsigma(x)\lambda'(x')\ra
\ee
and note that
\be
\la (I_u\times\mathrm{id})^*\phi, \lambda\lambda'\ra = \la \phi, (I_u\lambda)\lambda'\ra_S = \la \Phi, (SI_u\lambda)\lambda'\ra = \la \Phi, (P\lambda)\lambda'\ra = \la\Phi, \lambda\lambda'\ra
\ee
showing that $\Phi = (I_u\times\mathrm{id})^*\phi$. Or, with an abuse of notation, $\Phi(x,x') = \phi(\ang,x')$.

Consider now the case where $\Phi$ vanishes when restricted to $\ca F_1^0\otimes\ca F_1^0$. The last term in \eqref{PhiPRdecomp} vanishes, which gives
\be
\la\Phi, \lambda\lambda'\ra = \la\Phi, (P\lambda)\lambda'\ra + \la\Phi, R\lambda\,P\lambda'\ra
\ee
where we regrouped the first two terms into one. Define $\phi \in (\ca E\otimes \ca F_1)^*$ in the same way as before, so it accounts for the first term above. For the second term, define $\phi'\in (\ca F_1\otimes\ca E)^*$ by
\be
\la \phi'(x,\ang'), \lambda(x)\varsigma(\ang')\ra_S := \la \Phi(x,x'), R\lambda(x)\, S\varsigma(\ang')\ra
\ee
By a similar computation as before, we see that $\la\Phi, R\lambda\,P\lambda'\ra = \la (\mathrm{id}\times I_u)^*\phi',\lambda\lambda'\ra$. Therefore, with the abuse of notation, $\Phi(x,x') = \phi(\ang, x') + \phi'(x,\ang')$.
\end{proof}

It will be useful to know how the $u$-derivative acts on distributions of the type $\phi(\ang,x')$. If such a $\phi$ is an element of $(\ca F_1\otimes\ca F_1)^*$, then
\ba
\la \fr D_u\phi, \alpha\lambda\ra &= - \la\phi(\ang,x'), \pu\alpha(x)\lambda(x')\ra \no
&= - \int\!\epsilon_\scri\int\!\epsilon_\scri'\, \phi(\ang,x')\pu\alpha(x)\lambda(x') \no
&= -\int_S\!\epsilon_S\int\!\epsilon_\scri'\, \phi(\ang,x')\left(\alpha_+(\ang)-\alpha_-(\ang)\right)\lambda(x') \no
&= - \int\!\epsilon_\scri\int\!\epsilon_\scri'\, \phi(\ang,x')\left(L_+(u)-L_-(u)\right)\alpha(x)\lambda(x') \no
&= \la \left(L_-(u)-L_+(u)\right) \phi(\ang,x'), \alpha(x)\lambda(x')\ra
\ea
That is, $\fr D_u\phi(x,x') = \left(L_-(u)-L_+(u)\right)\phi(\ang,x')$. An analogous formula holds for the $u'$-derivative of distributions in $(\ca F_1\otimes\ca E)^*$.

Finally, let us make a comment about the reality of distributional brackets.
Since the Dirac brackets are real, we will assume that the distributional brackets are also real, in the following sense.
The conjugation mapping between $\ca F$ and $\bar{\ca F}$, $f \mapsto \bar f$, induces a map between distributions: if $\Phi \in \ca F^*$, we define its {\it complex-conjugate} $\bar\Phi \in \bar{\ca F}^*$ by
\be
\la \bar\Phi, \bar f\ra := \ol{\la \Phi, f\ra}\,\,, \quad \text{\it for all $f\in \ca F$}
\ee
This can be automatically extended to bidistributions.
The reality of the distributional brackets is consequently expressed as, for example,
\be
\{\bar\sigma(x), N(x')\}_D = \ol{\{\sigma(x), \bar N(x')\}_D}
\ee
Of course, this is not to say that $\{\bar\sigma(x), N(x')\}_D$ is a ``real-valued distribution''; only the ``brackets themselves'' are real (i.e., they behave as if they come from a real symplectic structure).\footnote{As a matter of fact, there are no ``real-valued distributions'', strictly speaking, since conjugation is never an automorphism (it maps $\ca F^*$ to $\bar{\ca F}^*$).
However, some distributions, like the Dirac delta or the Heaviside function, act in the exact same way as their complex-conjugates --- in that sense, we can say that they are {\it real-valued distributions}. 
Thus, with a slight abuse of notation, we may write things like $\bar\delta(x,x') = \delta(x,x')$, but it should be understood that if the $\delta$ on the right-hand side lives in $(\ca F \otimes \ca G)^*$, then the $\delta$ on the left-hand side lives in $(\bar{\ca F} \otimes \bar{\ca G})^*$.}

\subsection{Directional brackets}

To evaluate the directional distributional brackets, consider the directional Dirac brackets with $\bar\Gamma_{\rho+\pu\alpha,\alpha}^{\bb C}$, where $(\rho,\alpha) \in \ca F_1^0 \times \ca F_0$, in the first entry and $\bar\Gamma_{\lambda',\alpha'}^{\bb C}$, where $(\lambda',\alpha') \in \ca F_1 \times \ca F_0$, in the second entry. We have
\ba
\{\bar\Gamma_{\rho+\pu\alpha,\alpha}^{\bb C}; \bar\Gamma_{\lambda',\alpha'}^{\bb C}\}_D &= \Big\{\int\!\epsilon_\scri \left( (\rho(x) +\pu\alpha(x))\bar\sigma(x) +\alpha(x) N(x)\right); \no
&\qquad \int\!\epsilon'_\scri \left( \lambda'(x')\bar\sigma(x') +\alpha'(x') N(x')\right) \Big\}_D \no
&\!\!\!\!\!\!\!\!\!\!\!\!\!\!\!\!\!\!\!\!\!\!\!\!\!\!\!\!\!\!\!\!\!\!\!\!\!\!
= \int\!\epsilon_\scri\!\int\!\epsilon'_\scri \Big[
\rho(x)\lambda'(x') \{ \bar\sigma(x); \bar\sigma(x')\}_D \nonumber\\[-0.8ex]
&\!+ \alpha(x)\lambda'(x') \{ N(x); \bar\sigma(x')\}_D + \pu\alpha(x)\lambda'(x') \{ \bar\sigma(x); \bar\sigma(x')\}_D \nonumber\\[0.2ex]
&\!+ \rho(x) \alpha'(x') \{ \bar\sigma(x); N(x')\}_D \nonumber\\[0.2ex]
&\!+ \alpha(x)\alpha'(x') \{ N(x); N(x')\}_D + \pu\alpha(x)\alpha'(x') \{ \bar\sigma(x); N(x')\}_D \Big] \no
&\!\!\!\!\!\!\!\!\!\!\!\!\!\!\!\!\!\!\!\!\!\!\!\!\!\!\!\!\!\!\!\!\!\!\!\!\!\!
= \int\!\epsilon_\scri\!\int\!\epsilon'_\scri \Big[
\rho(x)\lambda'(x') \{ \bar\sigma(x); \bar\sigma(x')\}_D \nonumber\\[-0.8ex]
&\!+ \alpha(x)\lambda'(x') \left( \{ N(x); \bar\sigma(x')\}_D - \fr D_u \{ \bar\sigma(x); \bar\sigma(x')\}_D \right) \nonumber\\[0.2ex]
&\!+ \rho(x) \alpha'(x') \{ \bar\sigma(x); N(x')\}_D \nonumber\\[0.2ex]
&\!+ \alpha(x)\alpha'(x') \left( \{ N(x); N(x')\}_D - \fr D_u \{ \bar\sigma(x); N(x')\}_D \right) \Big] 
\label{GammaCGammaCdistdir}
\ea
According to \eqref{GammaCGammaC}, this must evaluate to $0$. 
Each line in the expression above must independently vanish. Note that the equation coming from the first line can be written as
\be
\la \{\bar\sigma(x); \bar\sigma(x')\}_D, \rho(x)\lambda'(x')\ra = 0 \quad \text{for all }\, \rho\lambda' \in \ca F_1^0 \otimes \ca F_1
\ee
Theorem~\ref{theodist} then gives
\be
\{\bar\sigma(x); \bar\sigma(x')\}_D = \phi(\ang,x')
\ee
for some $\phi(\ang,x')\in (\ca E\otimes\ca F_1)^*$.
The second line gives
\be
\la \{ N(x); \bar\sigma(x')\}_D - \fr D_u \{ \bar\sigma(x); \bar\sigma(x')\}_D, \alpha(x)\lambda'(x')\ra = 0 \quad \text{for all }\, \alpha\lambda' \in \ca F_0 \otimes \ca F_1
\ee
Since this distribution belongs to $(\ca F_0 \otimes \ca F_1)^*$, it must be zero
\be
\{ N(x); \bar\sigma(x')\}_D - \fr D_u \{ \bar\sigma(x); \bar\sigma(x')\}_D = 0
\ee
which gives
\be
\{ N(x); \bar\sigma(x')\}_D =  \left(L_-(u) - L_+(u)\right) \phi(\ang,x')
\ee
The third line reads
\be
\la \{\bar\sigma(x); N(x')\}_D, \rho(x) \alpha'(x')\ra = 0 \quad \text{for all }\, \rho\alpha' \in \ca F_1^0 \otimes \ca F_0
\ee
A straightforward modification of Theorem~\ref{theodist} implies
\be\label{sigmaNkappaPsi}
\{\bar\sigma(x); N(x')\}_D = \wt\phi(\ang,x')
\ee
for some $\wt\phi\in (\ca E\otimes\ca F_0)^*$.
Finally, the fourth line yields
\be
\{ N(x); N(x')\}_D = \fr D_u \{ \bar\sigma(x);  N(x')\}_D = \left(L_-(u) - L_+(u)\right) \wt\phi(\ang,x')
\ee
The ``minimal solution'' corresponds to taking $\phi=\wt\phi = 0$, which gives
\ba
\{\bar\sigma(x); \bar\sigma(x')\}_D &= 0 \\
\{\bar\sigma(x); N(x')\}_D &= 0 \\
\{ N(x); \bar\sigma(x')\}_D &= 0 \\
\{ N(x); N(x')\}_D &= 0
\ea

Next, we consider directional brackets between $\Gamma_{\rho+\pu\alpha,\alpha}^{\bb C}$ and $\bar\Gamma_{\lambda',\alpha'}^{\bb C}$. We have
\ba
\{\Gamma_{\rho+\pu\alpha,\alpha}^{\bb C}; \bar\Gamma_{\lambda',\alpha'}^{\bb C}\}_D &= \Big\{\int\!\epsilon_\scri \left( (\bar\rho(x) +\pu\bar\alpha(x))\sigma(x) +\bar\alpha(x)\bar N(x)\right); \no
&\qquad \int\!\epsilon'_\scri \left( \lambda'(x')\bar\sigma(x') +\alpha'(x') N(x')\right) \Big\}_D \no
&\!\!\!\!\!\!\!\!\!\!\!\!\!\!\!\!\!\!\!\!\!\!\!\!\!\!\!\!\!\!\!\!\!\!\!\!\!\!
= \int\!\epsilon_\scri\!\int\!\epsilon'_\scri \Big[
\bar\rho(x)\lambda'(x') \{ \sigma(x); \bar\sigma(x')\}_D \nonumber\\[-0.8ex]
&\!+ \bar\alpha(x)\lambda'(x') \left( \{ \bar N(x); \bar\sigma(x')\}_D - \fr D_u \{ \sigma(x); \bar\sigma(x')\}_D \right) \nonumber\\[0.2ex]
&\!+ \bar\rho(x) \alpha'(x') \{ \sigma(x);  N(x')\}_D \nonumber\\[0.2ex]
&\!+ \bar\alpha(x)\alpha'(x') \left( \{ \bar N(x); N(x')\}_D - \fr D_u \{ \sigma(x); N(x')\}_D \right) \Big] 
\label{GammaCbarGammaCdistdir}
\ea
According to \eqref{GammaCbarGammaC}, this must evaluate to
\ba
\{\Gamma_{\rho + \pu\alpha, \alpha}^{\bb C}, \bar\Gamma_{\lambda', \alpha'}^{\bb C}\}_D &= -4\pi\int\!\epsilon_\scri\left(\lambda' (\bar P + 2\bar\alpha) - \alpha'(\bar\rho + 2\pu\bar\alpha) \right) \no
&\!\!\!\!\!\!\!\!\!\!\!\!\!\!\!\!\!\!\!\!\!\!\!\!\!\!= -4\pi \int\!\epsilon_\scri\!\int\!\epsilon'_\scri\, \delta(\ang,\ang')\Big[
\bar\rho(x)\lambda'(x') H(u',u)  \nonumber\\[-0.8ex]
&\qquad\qquad\qquad\quad + 2\bar\alpha(x)\lambda'(x') \delta(u,u') \nonumber\\[0.2ex]
&\qquad\qquad\qquad\quad - \bar\rho(x) \alpha'(x') \delta(u,u')  \nonumber\\[0.2ex]
&\qquad\qquad\qquad\quad + 2\bar\alpha(x)\alpha'(x') \fr D_u\delta(u,u') \Big] 
\ea
where $P(x) := \int_{-\infty}^u\!du'\,\rho(u',\ang)$ and, in the second equality, we have introduced the appropriate bidistributions. The following equations are implied
\ba
&\la \{\sigma(x); \bar\sigma(x')\}_D + 4\pi \delta(\ang,\ang')H(u',u), \bar\rho(x) \lambda'(x')\ra = 0 \,\,\,\, \text{\it for all }\, \bar\rho\lambda' \in {\bar{\ca F}_1^0} \otimes \ca F_1 \\
&\la \{\bar N(x); \bar\sigma(x')\}_D - \fr D_u\{\sigma(x); \bar\sigma(x')\}_D + 8\pi \delta(x,x'), \bar\alpha(x) \lambda'(x')\ra = 0 \no
&\hspace{23em} \text{\it for all }\, \bar\alpha\lambda' \in \bar{\ca F}_0 \otimes \ca F_1 \\
&\la \{\sigma(x); N(x')\}_D - 4\pi \delta(x,x'), \bar\rho(x) \alpha'(x')\ra = 0 \qquad\,\,\,\,\,\,\,\,\, \text{\it for all }\, \bar\rho\alpha' \in {\bar{\ca F}_1^0} \otimes \ca F_0 \\
&\la \{\bar N(x); N(x')\}_D - \fr D_u\{\sigma(x); N(x')\}_D + 8\pi \fr D_u\delta(x,x'), \bar\alpha(x) \alpha'(x')\ra = 0 \no
&\hspace{23em} \text{\it for all }\, \bar\alpha\alpha' \in \bar{\ca F}_0 \otimes \ca F_0
\ea
The first line gives
\be
\{\sigma(x); \bar\sigma(x')\}_D = - 4\pi \delta(\ang,\ang')H(u',u) + \psi(\ang,x')
\ee
where $\psi \in (\bar{\ca E}\otimes \ca F_1)^*$.
The second line gives
\ba
\{\bar N(x); \bar\sigma(x')\}_D &= \fr D_u\{\sigma(x); \bar\sigma(x')\}_D - 8\pi \delta(x,x') \no
&= -4\pi\delta(x,x') - 4\pi\delta(\ang,\ang')L_-(u) + \left(L_-(u) - L_+(u)\right)\psi(\ang,x')
\ea
where we have used \eqref{DuHuu'} to compute $\fr D_uH(u',u) = \fr D_u\left(1 - H(u,u')\right) =  - \delta(u,u') + L_-(u)$.
The third line gives
\be
\{\sigma(x); N(x')\}_D = 4\pi \delta(x,x') + \wt\psi(\ang,x')
\ee
where $\wt\psi(\ang,x') \in (\bar{\ca E}\otimes\ca F_0)^*$.
The fourth line gives
\ba
\{\bar N(x); N(x')\}_D &= \fr D_u\{\sigma(x); N(x')\}_D - 8\pi \fr D_u\delta(x,x') \no
&= -4\pi \fr D_u\delta(x,x') + \left(L_-(u) - L_+(u)\right)\wt\psi(\ang,x')
\ea
The ``minimal solution'' consists of taking $\psi = \wt\psi = 0$, which yields 
\ba
\{\sigma(x); \bar\sigma(x')\}_D &= - 4\pi \delta(\ang,\ang')H(u',u) \\
\{\bar N(x); \bar\sigma(x')\}_D &= -4\pi\delta(x,x') - 4\pi \delta(\ang,\ang') L_-(u) \\
\{\sigma(x); N(x')\}_D &= 4\pi \delta(x,x') \\
\{\bar N(x); N(x')\}_D &= -4\pi\, \fr D_u\delta(x,x')
\ea
All the other brackets (involving complex-conjugated entries) could be computed by considering the remaining brackets between shear-news functions, namely $\{\Gamma_{\rho+\pu\alpha,\alpha}^{\bb C}; \Gamma_{\lambda',\alpha'}^{\bb C}\}_D$ and $\{\bar\Gamma_{\rho+\pu\alpha,\alpha}^{\bb C}; \Gamma_{\lambda',\alpha'}^{\bb C}\}_D$. However, since every manipulation behaves ordinarily under complex conjugation, the reality of the Dirac brackets will extend automatically to distributional brackets. So, for example,
\be
\{N(x); \bar N(x')\}_D = \ol{\{\bar N(x); N(x')\}_D} = -4\pi\, \fr D_u\delta(x,x')
\ee
where, in the last equality, we have used the reality of the Dirac delta (and of its derivative).

\subsection{Proper brackets}

The proper distributional brackets are intended to compute the Dirac brackets between two proper observables. Since proper brackets are automatically anti-symmetric, we will explore the multiplicity of solutions for the distributional brackets to make them manifestly anti-symmetric, if possible. Of course this is not necessary, as any solution is just as good for the ultimate purpose of the distributional brackets, which is only to reconstruct brackets between proper observables --- nevertheless, the solutions with the most amount of symmetries should be preferred simply for ``aesthetic'' reasons. 
Thus, we will seek for solutions satisfying, for example,
\be
\{\sigma(x'), \bar N(x)\}_D = - \{\bar N(x), \sigma(x')\}_D
\ee
if possible.

Let us first consider the brackets between $\bar\Gamma_{\rho+\pu\alpha,\alpha}^{\bb C}$ and $\bar\Gamma_{\rho'+\pu\alpha',\alpha'}^{\bb C}$, with $(\rho,\alpha), (\rho',\alpha') \in \ca F_1^0 \times \ca F_0$. We have
\ba
\{\bar\Gamma_{\rho+\pu\alpha,\alpha}^{\bb C}, \bar\Gamma_{\rho'+\pu\alpha',\alpha'}^{\bb C}\}_D &= \Big\{\int\!\epsilon_\scri \left( (\rho(x) +\pu\alpha(x))\bar\sigma(x) + \alpha(x) N(x)\right), \no
&\qquad \int\!\epsilon'_\scri \left( (\rho'(x') +\pup\alpha'(x'))\bar\sigma(x') +\alpha'(x') N(x')\right) \Big\}_D \no
&\!\!\!\!\!\!\!\!\!\!\!\!\!\!\!\!\!\!\!\!\!\!\!\!\!\!\!\!\!\!\!\!\!\!\!\!\!\!
= \int\!\epsilon_\scri\!\int\!\epsilon'_\scri \Big[
\rho(x)\rho'(x') \{ \bar\sigma(x), \bar\sigma(x')\}_D \nonumber\\[-0.8ex]
&\!+\rho(x)\alpha'(x')\{\bar\sigma(x), N(x')\}_D + \rho(x)\pup\alpha'(x') \{\bar\sigma(x),\bar\sigma(x')\}_D  \nonumber\\[0.2ex]
&\!+ \alpha(x)\rho'(x')\{ N(x), \bar\sigma(x')\}_D + \pu\alpha(x)\rho'(x')\{\bar\sigma(x),\bar\sigma(x')\}_D \nonumber\\[0.2ex]
&\!+ \alpha(x)\alpha'(x') \{ N(x), N(x')\}_D + \pu\alpha(x)\pup\alpha'(x') \{\bar\sigma(x), \bar\sigma(x')\}_D \nonumber\\[-0.3ex]
&\!+\pu\alpha(x)\alpha'(x')\{\bar\sigma(x), N(x')\}_D + \alpha(x)\pup\alpha'(x') \{ N(x), \bar\sigma(x')\}_D \Big] \no
&\!\!\!\!\!\!\!\!\!\!\!\!\!\!\!\!\!\!\!\!\!\!\!\!\!\!\!\!\!\!\!\!\!\!\!\!\!\!
= \int\!\epsilon_\scri\!\int\!\epsilon'_\scri \Big[
\rho(x)\rho'(x') \{ \bar\sigma(x), \bar\sigma(x')\}_D \nonumber\\[-0.8ex]
&\!+\rho(x)\alpha'(x')\left( \{\bar\sigma(x), N(x')\}_D - \fr D_{u'}\{\bar\sigma(x), \bar\sigma(x')\}_D \right) \nonumber\\[0.2ex]
&\!+ \alpha(x)\rho'(x')\left( \{ N(x), \bar\sigma(x')\}_D - \fr D_{u}\{\bar\sigma(x), \bar\sigma(x')\}_D\right) \nonumber\\[0.2ex]
&\!+ \alpha(x)\alpha'(x') \big(\{ N(x), N(x')\}_D + \fr D_{u}\fr D_{u'} \{\bar\sigma(x), \bar\sigma(x')\}_D \nonumber\\[-0.3ex]
&\qquad\qquad\qquad -\fr D_{u}\{\bar\sigma(x), N(x')\}_D -\fr D_{u'}\{ N(x), \bar\sigma(x')\}_D \big) \Big]
\label{GammaCGammaCdist}
\ea
According to \eqref{GammaCGammaC}, this should vanish for all $(\rho,\alpha), (\rho',\alpha') \in \ca F_1^0 \times \ca F_0$. We thus get the following equations,
\ba
&\la\{ \bar\sigma(x), \bar\sigma(x')\}_D, \rho(x)\rho'(x')\ra = 0 \hspace{11.2em}\text{\it for all }\, \rho\rho' \in \ca F_1^0 \otimes \ca F_1^0 \no
&\la \{\bar\sigma(x), N(x')\}_D - \fr D_{u'}\{\bar\sigma(x), \bar\sigma(x')\}_D , \rho(x)\alpha'(x')\ra = 0 \qquad\!\! \text{\it for all }\, \rho\alpha' \in \ca F_1^0 \otimes \ca F_0 \no
&\la\{ N(x), \bar\sigma(x')\}_D - \fr D_{u}\{\bar\sigma(x), \bar\sigma(x')\}_D, \alpha(x)\rho'(x')\ra = 0 \qquad\text{\it for all }\, \alpha\rho' \in \ca F_0 \otimes \ca F_1^0 \no
& \la\{ N(x), N(x')\}_D + \fr D_{u}\fr D_{u'} \{\bar\sigma(x), \bar\sigma(x')\}_D -\fr D_{u}\{\bar\sigma(x), N(x')\}_D \no
&\qquad\qquad\qquad\,\, -\fr D_{u'}\{ N(x), \bar\sigma(x')\}_D, \alpha(x)\alpha'(x')\ra = 0 \qquad\text{\it for all }\, \alpha\alpha' \in \ca F_0 \otimes \ca F_0
\ea
The first line gives, using Theorem~\ref{theodist} and imposing anti-symmetry,
\be
\{ \bar\sigma(x), \bar\sigma(x')\}_D = \phi(\ang,x') - \phi(\ang',x)
\ee
where $\phi \in (\ca E\otimes \ca F_1)^*$. 
The second line yields
\be
\{\bar\sigma(x), N(x')\}_D - \fr D_{u'}\{\bar\sigma(x), \bar\sigma(x')\}_D = \wt\phi(\ang,x')
\ee
for $\wt\phi \in (\ca E\otimes \ca F_0)^*$. Evaluating the $u'$-derivative gives
\be\label{barsigmaNshortarg}
\{\bar\sigma(x), N(x')\}_D = \wt\phi(\ang,x') + \fr D_{u'}\phi(\ang,x') + \left(L_+(u') - L_-(u')\right) \phi(\ang',x)
\ee
Note that $\fr D_{u'}\phi(\ang,x')$ is a distribution of the same kind as $\wt\phi(\ang,x')$, and thus it can be eliminated via a redefinition
\be
\wt\phi(\ang,x') + \fr D_{u'}\phi(\ang,x') \to \wt\phi(\ang,x') 
\ee
In this way, we can write simply 
\be\label{barsigmaNshortarg2}
\{\bar\sigma(x), N(x')\}_D = \wt\phi(\ang,x') + \left(L_+(u') - L_-(u')\right) \phi(\ang',x)
\ee
The third line is equivalent to the second under the restriction of anti-symmetry. 
The fourth line gives
\be
\{ N(x), N(x')\}_D + \fr D_{u}\fr D_{u'} \{\bar\sigma(x), \bar\sigma(x')\}_D -\fr D_{u}\{\bar\sigma(x), N(x')\}_D -\fr D_{u'}\{ N(x), \bar\sigma(x')\}_D = 0
\ee
and, evaluating the derivatives,
\be
\{ N(x), N(x')\}_D = \left(L_+(u') - L_-(u')\right) \wt\phi(\ang',x) - \left(L_+(u) - L_-(u)\right) \wt\phi(\ang,x')
\ee
The ``minimal solution'' consists of taking $\phi = \wt\phi = 0$, which yields
\ba
\{ \bar\sigma(x), \bar\sigma(x')\}_D &= 0 \no
\{\bar\sigma(x), N(x')\}_D &= 0 \no
\{ N(x), N(x')\}_D &= 0
\ea

Next, we consider the brackets between $\Gamma_{\rho+\pu\alpha,\alpha}^{\bb C}$ and $\bar\Gamma_{\rho'+\pu\alpha',\alpha'}^{\bb C}$, with $(\rho,\alpha), (\rho',\alpha') \in \ca F_1^0 \times \ca F_0$. We have
\ba
\{\Gamma_{\rho+\pu\alpha,\alpha}^{\bb C}, \bar\Gamma_{\rho'+\pu\alpha',\alpha'}^{\bb C}\}_D &= \Big\{\int\!\epsilon_\scri \left( (\bar\rho(x) +\pu\bar\alpha(x))\sigma(x) + \bar\alpha(x) \bar N(x)\right), \no
&\qquad \int\!\epsilon'_\scri \left( (\rho'(x') +\pup\alpha'(x'))\bar\sigma(x') +\alpha'(x') N(x')\right) \Big\}_D \no
&\!\!\!\!\!\!\!\!\!\!\!\!\!\!\!\!\!\!\!\!\!\!\!\!\!\!\!\!\!\!\!\!\!\!\!\!\!\!
= \int\!\epsilon_\scri\!\int\!\epsilon'_\scri \Big[
\bar\rho(x)\rho'(x') \{ \sigma(x), \bar\sigma(x')\}_D \nonumber\\[-0.8ex]
&\!+\bar\rho(x)\alpha'(x')\{\sigma(x), N(x')\}_D + \bar\rho(x)\pup\alpha'(x') \{\sigma(x),\bar\sigma(x')\}_D  \nonumber\\[0.2ex]
&\!+ \bar\alpha(x)\rho'(x')\{ \bar N(x), \bar\sigma(x')\}_D + \pu\bar\alpha(x)\rho'(x')\{\sigma(x),\bar\sigma(x')\}_D \nonumber\\[0.2ex]
&\!+ \bar\alpha(x)\alpha'(x') \{ \bar N(x), N(x')\}_D + \pu\bar\alpha(x)\pup\alpha'(x') \{\sigma(x), \bar\sigma(x')\}_D \nonumber\\[-0.3ex]
&\!+\pu\bar\alpha(x)\alpha'(x')\{\sigma(x), N(x')\}_D + \bar\alpha(x)\pup\alpha'(x') \{ \bar N(x), \bar\sigma(x')\}_D \Big] \no
&\!\!\!\!\!\!\!\!\!\!\!\!\!\!\!\!\!\!\!\!\!\!\!\!\!\!\!\!\!\!\!\!\!\!\!\!\!\!
= \int\!\epsilon_\scri\!\int\!\epsilon'_\scri \Big[
\bar\rho(x)\rho'(x') \{ \sigma(x), \bar\sigma(x')\}_D \nonumber\\[-0.8ex]
&\!+\bar\rho(x)\alpha'(x')\left( \{\sigma(x), N(x')\}_D - \fr D_{u'}\{\sigma(x), \bar\sigma(x')\}_D \right) \nonumber\\[0.2ex]
&\!+ \bar\alpha(x)\rho'(x')\left( \{ \bar N(x), \bar\sigma(x')\}_D - \fr D_{u}\{\sigma(x), \bar\sigma(x')\}_D\right) \nonumber\\[0.2ex]
&\!+ \bar\alpha(x)\alpha'(x') \big(\{ \bar N(x), N(x')\}_D + \fr D_{u}\fr D_{u'} \{\sigma(x), \bar\sigma(x')\}_D \nonumber\\[-0.3ex]
&\qquad\qquad\qquad -\fr D_{u}\{\sigma(x), N(x')\}_D -\fr D_{u'}\{ \bar N(x), \bar\sigma(x')\}_D \big) \Big]
\label{GammaCbarGammaCdist}
\ea
According to \eqref{GammaCbarGammaCas}, this must evaluate to
\ba
\{\Gamma_{\rho +\pu\alpha, \alpha}^{\bb C}, \bar\Gamma_{\rho'+\pu\alpha', \alpha'}^{\bb C}\}_D =&\,\, 2\pi\int\!\epsilon_\scri\big( \bar\rho P' - \rho' \bar P + 3 \bar\rho \alpha' - 3 \rho'\bar\alpha \nonumber\\[-0.3ex]
&\qquad\qquad\,\, +\pu\bar\alpha P' - \pu\alpha' \bar P + 4 \pu\bar\alpha \alpha' - 4 \pu\alpha' \bar\alpha \big) \nonumber\\[0.3ex]
&\!\!\!\!\!\!\!\!\!\!\!\!\!\!\!\!\!\!\!\!\!\!\!\!\!\!= 2\pi \int\!\epsilon_\scri\!\int\!\epsilon'_\scri\, \delta(\ang,\ang')\Big[
\bar\rho(x)\rho'(x') \left( H(u,u') - H(u',u) \right)  \nonumber\\[-0.8ex]
&\qquad\qquad\qquad\,\, + \bar\rho(x)\alpha'(x') \left( 3\delta(u,u') + \fr D_{u'}H(u',u) \right) \nonumber\\[0.3ex]
&\qquad\qquad\qquad\,\, - \bar\alpha(x) \rho'(x') \left( 3\delta(u,u') + \fr D_{u}H(u,u') \right)  \nonumber\\[0.2ex]
&\qquad\qquad\qquad\,\, - 8\bar\alpha(x)\alpha'(x') \fr D_u^-\delta(u,u') \Big] 
\ea
where $P(x) := \int_{-\infty}^u\!du'\, \rho(u',\ang)$.
The following equations are implied
\ba
&\la\{ \sigma(x), \bar\sigma(x')\}_D - 2\pi\delta(\ang,\ang')\left( H(u,u') - H(u',u) \right), \bar\rho(x)\rho'(x')\ra = 0 \no
&\hspace{25em} \text{\it for all }\, \bar\rho\rho' \in {\bar{\ca F}_1^0} \otimes \ca F_1^0 \no
&\la \{\sigma(x), N(x')\}_D - \fr D_{u'}\{\sigma(x), \bar\sigma(x')\}_D - 2\pi\delta(\ang,\ang') \left(4\delta(u,u') - L_+(u')\right), \bar\rho(x)\alpha'(x')\ra = 0 \no
&\hspace{25em} \text{\it for all }\, \bar\rho\alpha' \in {\bar{\ca F}_1^0} \otimes \ca F_0 \no
&\la \{ \bar N(x), \bar\sigma(x')\}_D - \fr D_{u}\{\sigma(x), \bar\sigma(x')\}_D + 2\pi\delta(\ang,\ang') \left(4\delta(u,u') - L_+(u)\right), \bar\alpha(x)\rho'(x')\ra = 0 \no
&\hspace{25em} \text{\it for all }\, \bar\alpha\rho' \in \bar{\ca F}_0 \otimes \ca F_1^0 \no
& \la\{ \bar N(x), N(x')\}_D + \fr D_{u}\fr D_{u'} \{\sigma(x), \bar\sigma(x')\}_D -\fr D_{u}\{\sigma(x), N(x')\}_D \no
&\hspace{10em} -\fr D_{u'}\{ \bar N(x), \bar\sigma(x')\}_D + 16\pi \fr D_u^-\delta(x,x'), \bar\alpha(x)\alpha'(x')\ra = 0 \no
&\hspace{25em} \text{\it for all }\, \bar\alpha\alpha' \in \bar{\ca F}_0 \otimes \ca F_0
\label{disteqsGammaCbarGammaCproper}
\ea
The first line gives, imposing reality and anti-symmetry,
\be
\{\sigma(x), \bar\sigma(x')\}_D = 2\pi\delta(\ang,\ang')\left( H(u,u') - H(u',u) \right) + \psi(\ang,x') - \bar\psi(\ang',x)
\ee
where $\psi \in (\bar{\ca E}\otimes \ca F_1)^*$.
The second line gives
\be
\{\sigma(x), N(x')\}_D = \fr D_{u'}\{\sigma(x), \bar\sigma(x')\}_D + 2\pi\delta(\ang,\ang') \left(4\delta(u,u') - L_+(u')\right) + \wt\psi(\ang,x')
\ee
where $\wt\psi \in (\bar{\ca E}\otimes \ca F_0)^*$. Evaluating the $u'$-derivative yields
\ba
\{\sigma(x), N(x')\}_D &=  2\pi\delta(\ang,\ang') \left(2\delta(u,u') + L_-(u')\right) + \left(L_+(u') - L_-(u')\right)\bar\psi(\ang',x) \no
&\quad + \fr D_{u'}\psi(\ang,x') + \wt\psi(\ang,x')
\ea
Again, since $\fr D_{u'}\psi(\ang,x')$ is of the same kind as $\wt\psi(\ang,x')$, 
we can redefine
\be
\wt\psi(\ang,x') + \fr D_{u'}\psi(\ang,x') \to \wt\psi(\ang,x')
\ee
so the expression above becomes
\be
\{\sigma(x), N(x')\}_D =  2\pi\delta(\ang,\ang') \left(2\delta(u,u') + L_-(u')\right) + \left(L_+(u') - L_-(u')\right)\bar\psi(\ang',x) +\wt\psi(\ang,x')
\ee
From reality and anti-symmetry, we should have
\be
\{\bar N(x), \bar\sigma(x')\}_D = \ol{\{N(x), \sigma(x')\}_D} = - \ol{\{\sigma(x'), N(x)\}_D }
\ee
which is compatible with the third line of \eqref{disteqsGammaCbarGammaCproper}, so no new information is gained from it.
The fourth line gives
\ba
\{ \bar N(x), N(x')\}_D &= - \fr D_{u}\fr D_{u'} \{\sigma(x), \bar\sigma(x')\}_D +\fr D_{u}\{\sigma(x), N(x')\}_D \no
&\quad\, +\fr D_{u'}\{ \bar N(x), \bar\sigma(x')\}_D - 16\pi \fr D_u^-\delta(x,x')
\ea
Computing all derivatives yields
\ba
\{ \bar N(x), N(x')\}_D &= - 4\pi \fr D_u^-\delta(x,x') \no
&\quad\, +\left(L_-(u) - L_+(u)\right) \wt\psi(\ang,x') - \left(L_-(u') - L_+(u')\right) \ol{\wt\psi}(\ang',x)
\ea
The ``minimal solution'' corresponds to taking $\psi=\wt\psi = 0$, which gives
\ba
\{\sigma(x), \bar\sigma(x')\}_D &= 2\pi\delta(\ang,\ang') \left( H(u,u') - H(u',u)\right) \\
\{\sigma(x), N(x')\}_D &= 4\pi \delta(x,x') + 2\pi\delta(\ang,\ang') L_-(u') \\
\{\bar N(x),N(x')\}_D &= -4\pi\,\fr D_u^-\delta(x,x') 
\ea
All the other brackets can be obtained by reality and anti-symmetry.

\bibliographystyle{JHEPs}
\bibliography{biblio}
%

\end{document}